\documentclass[11pt]{article}

\usepackage{epsfig, psfrag}
\usepackage{latexsym}
\usepackage{enumerate}
\usepackage{url}
\usepackage{float}
\usepackage[hypertexnames=false,hyperfootnotes=false]{hyperref}
\usepackage{texnansi}
\usepackage{color}
\usepackage{tikz}
\usepackage[margin=10pt,font=small,labelfont=bf]{caption}
\usepackage[subrefformat=parens,labelformat=simple]{subcaption}
\usepackage{afterpage}
\usepackage{enumitem}
\usepackage[boxed]{algorithm}
\usepackage{algpseudocode}
\usepackage[normalem]{ulem}
\usepackage{lmodern}
\usepackage{booktabs}
\usepackage{sectsty}
\usepackage{ifthen}
\usepackage{amsmath}
\usepackage{amsthm}
\usepackage{amssymb}
\usepackage{amsfonts}
\usepackage{thmtools}
\usepackage{thm-restate}
\usepackage{xspace}
\usepackage{titling}
\usepackage{parskip}
\usepackage{natbib}
\usepackage{xfrac}
\usepackage{multirow}
\usepackage{bigdelim}
\usepackage{bm}
\usepackage{makecell}
\usepackage{colortbl}
\usepackage{xcolor}
\PassOptionsToPackage{algo2e,ruled}{algorithm2e}
\RequirePackage{algorithm2e}
\usepackage[textsize=tiny]{todonotes}
\newcommand{\field}[1]{\ensuremath{\mathbb{#1}}}
\newcommand{\R}{\ensuremath{\field{R}}} 
\newcommand{\I}[1]{\ensuremath{\mathbb{I}_{\left\{#1\right\}}}} 
\newcommand{\PR}{\ensuremath{\mathbb{P}}} 
\newcommand{\E}{\ensuremath{\mathbb{E}}} 
\newcommand{\defeq}{\ensuremath{\triangleq}}

\newcommand{\Ascr}{\ensuremath{\mathcal A}}

\newcommand{\Cscr}{\ensuremath{\mathcal C}}

\newcommand{\Fscr}{\ensuremath{\mathcal F}}

\newcommand{\Lscr}{\ensuremath{\mathcal L}}

\newcommand{\Nscr}{\ensuremath{\mathcal N}}

\newcommand{\Pscr}{\ensuremath{\mathcal P}}
\newcommand{\Qscr}{\ensuremath{\mathcal Q}}

\newcommand{\Xscr}{\ensuremath{\mathcal X}}
\newcommand{\Yscr}{\ensuremath{\mathcal Y}}

\DeclareMathOperator*{\argmin}{\mathrm{argmin}}
\DeclareMathOperator*{\argmax}{\mathrm{argmax}}
\newcommand{\minimize}{\ensuremath{\mathop{\mathrm{minimize}}\limits}}

\DeclareMathOperator{\var}{\textup{VaR}}
\DeclareMathOperator{\cvar}{\mathrm{CVaR}}
\DeclareMathOperator{\scvar}{\mathrm{S-CVaR}}

\DeclareMathOperator{\Prog}{\Pscr}

\newcommand{\seclabel}[1]{\noindent\textbf{\sffamily #1}}
\newcommand{\subproof}[1]{\noindent\textit{#1}}

\declaretheoremstyle[headfont=\sffamily\bfseries,bodyfont=\itshape]{thm-sf}
\declaretheorem[style=thm-sf]{theorem}
\declaretheorem[style=thm-sf]{remark}

\declaretheorem[style=thm-sf]{definition}

\declaretheorem[style=thm-sf]{corollary}
\declaretheorem[style=thm-sf]{lemma}

\declaretheorem[style=thm-sf]{proposition}

\newcommand{\proofnamest}[1]{{\normalfont\sffamily\bfseries #1}}
\renewcommand{\thmcontinues}[1]{\hyperref[#1]{continued}}
\usetikzlibrary{arrows,patterns,plotmarks,pgfplots.groupplots}
\tikzstyle{every picture} += [>=stealth]
\tikzset{axis/.style={semithick, line join=miter}}
\allsectionsfont{\sffamily}
\makeatletter
\def\@seccntformat#1{\csname the#1\endcsname.\quad}
\makeatother

\floatname{algorithm}{\normalfont\sffamily\bfseries Algorithm}
\floatstyle{ruled}
\newcommand{\emailhref}[1]{\href{mailto:#1}{\tt #1}} 
\provideboolean{fastcompile}
\newcommand{\hidefastcompile}[1]{\ifthenelse{\boolean{fastcompile}}{}{#1}}
\usepackage{pgfplots}
\usepackage{pgfplotstable}
\usetikzlibrary{calc}
\usepackage{mathtools}
\definecolor{orange}{rgb}{0.85,0.33,0.13} 
\definecolor{green}{rgb}{0.13,0.85,0.33}
\definecolor{purple}{rgb}{0.33,0.13,0.85}
\definecolor{lime}{rgb}{0.65,0.85,0.13}
\definecolor{blue}{rgb}{0.13,0.65,0.85}
\pgfplotscreateplotcyclelist{tricolor}{%
  orange,every mark/.append style={fill=orange!80!black},mark=*\\%
  green,every mark/.append style={fill=green!80!black},mark=square*\\%
  purple,every mark/.append style={fill=purple!80!black},mark=otimes*\\%
  black,mark=star\\%
  orange,every mark/.append style={fill=orange!80!black},mark=diamond*\\%
  green,densely dashed,every mark/.append style={solid,fill=green!80!black},mark=*\\%
  purple,densely dashed,every mark/.append style={solid,fill=purple!80!black},mark=square*\\%
  black,densely dashed,every mark/.append style={solid,fill=gray},mark=otimes*\\%
  orange,densely dashed,mark=star,every mark/.append style=solid\\%
  green,densely dashed,every mark/.append style={solid,fill=green!80!black},mark=diamond*\\%
}
\pgfplotsset{colormap={tricolormap}{color=(orange) color=(green) color=(purple)},
  colormap={quadcolormap}{color=(orange) color=(lime) color=(blue) color=(purple)}}
\pgfplotstableset{%
  font=\small,
  every head row/.style={before row=\toprule[1pt], after row=\midrule},
  every last row/.style={after row=\bottomrule[1pt]}}
\pgfplotsset{compat=1.15}

\usepackage{mathtools}
\mathtoolsset{showonlyrefs=true}

\renewcommand*{\thead}[1]{\bfseries \makecell{#1}}

\usepackage{setspace}
\usepackage[margin=1in]{geometry}

\setcitestyle{authoryear,round,semicolon,aysep={,},yysep={,},notesep={, }}
\provideboolean{submissionversion}
 \setboolean{submissionversion}{true}
\provideboolean{blindversion}

\ifthenelse{\boolean{submissionversion}}{
  \renewcommand{\todo}[1]{}
  \newcommand{\newedit}[1]{#1}
  \newcommand{\deledit}[1]{}
}{
  \newcommand{\newedit}[1]{{\color{green} #1}}
  \newcommand{\deledit}[1]{{\color{orange} \sout{#1}}}
}

\tikzstyle{rate} += [color=orange,very thick]
\usetikzlibrary{matrix,positioning}
\usepgfplotslibrary{groupplots}
\pgfplotsset{compat=newest}

\ifthenelse{\boolean{blindversion}}{
  \title{\textsf{\textbf{Risk-Sensitive Optimal Execution via a Conditional Value-at-Risk Objective}}}
  \author{}
  \date{}
}{
  \title{\textsf{\textbf{Risk-Sensitive Optimal Execution via a Conditional Value-at-Risk
        Objective\thanks{
          The authors wish to thank Carlos Gomez Gascon and Paul Glasserman for
          helpful discussions.}
}}}
\author{ \\
  Seungki Min \\
  KAIST \\
  \emailhref{skmin@kaist.ac.kr}
  \and \\
  Ciamac C. Moallemi \\
  Columbia University \\
 \emailhref{ciamac@gsb.columbia.edu} \\
  \and \\
  Costis Maglaras \\
  Columbia University \\
  \emailhref{c.maglaras@gsb.columbia.edu}  \\
}

\date{
 Current Revision: January 2022}
}

\begin{document}

\maketitle
\singlespacing
\begin{center}
  \bfseries\sffamily PRELIMINARY VERSION --- DO NOT DISTRIBUTE
\end{center}

\begin{abstract}
  We consider a liquidation problem in which a risk-averse trader tries to liquidate a fixed
  quantity of an asset in the presence of market impact and random price fluctuations.  When
  deciding the liquidation strategy, the trader encounters a trade-off between the transaction
  costs incurred due to market impact and the volatility risk of holding the position.  Our
  formulation begins with a continuous-time and infinite horizon variation of the seminal model of
  \citet{Almgren00}, but we define as the objective the conditional value-at-risk (CVaR) of the
  implementation shortfall, and allow for dynamic (adaptive) trading strategies.  In this setting,
  remarkably, we are able to derive closed-form expressions for the optimal liquidation strategy and
  its value function.

  Our results yield a number of important practical insights. We are able to quantify the benefit
  of adaptive policies over optimized static (pre-committed) policies. The
  relevant improvement depends only on the level of risk aversion, and grows
  without bound as the trader becomes more risk neutral.  For moderate levels of risk aversion,
  the optimal dynamic policy outperforms the optimal static policy by 5--15\%, and outperforms the
  optimal volume weighted average price (VWAP) policy by 15--25\%. This improvement is achieved
  through dynamic policies that exhibit ``aggressiveness-in-the-money'': trading is accelerated
  when price movements are favorable (to minimize risk), and is slowed when price movements are
  unfavorable (to minimize transaction costs). Overall, the optimal dynamic policies exhibit much
  better performance in the right tail of worst outcomes, relative to optimal static policies.

  From a mathematical perspective, our analysis exploits the dual representation of CVaR to
  convert the problem to a continuous-time, zero sum dynamic game. In this setting, we leverage
  the idea of the state-space augmentation, recently applied to certain discrete-time Markov
  decision processes with a CVaR objective. 
  We obtain a partial differential equation describing the
  optimal value function, which is separable and a special instance of the
  Emden--Fowler equation. This leads to a closed-form solution. As our problem is a special case
  of a continuous-time linear-quadratic-Gaussian control problem with a CVaR objective, these
  results may be interesting in broader settings.


\end{abstract}

\onehalfspacing

\section{Introduction} \label{sec:intro}

Optimal execution is a problem of significant importance for algorithmic traders in modern
financial markets. Here, a trader must decide an optimal trading strategy to buy (or sell) a fixed
quantity of an asset. Typically, there is a trade-off between trading quickly, which minimizes the
risk of adverse price movements, and trading slowly, which minimizes transaction costs.  In their
seminal paper, \citet{Almgren00} framed this liquidation problem as a mean-variance optimization
problem, that optimizes a combination of the average cost (i.e., the expected implementation
shortfall) and the variability of the cost (i.e., the variance of the implementation shortfall).
In that setting, they explicitly derive the optimal liquidation schedule, which is parameterized
by the risk-aversion level of the trader. Their framework and suggested solution have become
standards in this area, serving as a starting point for more complicated models and trading
strategies, both in the academic literature and among practitioners.

However, the \citet{Almgren00} analysis is restricted to static strategies: they only considered
deterministic schedules under which the trader pre-commits to a fixed trading schedule, i.e., algorithms
that do not adapt to changing market conditions such as the price of the asset.  This restriction
makes the analysis considerably more straightforward and tractable, and allows for closed-form
solutions. This leaves open possible additional benefit from dynamic, adaptive trading strategies,
however. Indeed, most practitioners utilize adaptive strategies. This is often done through ad hoc
or heuristic means, such as model predictive control (MPC), where the trader periodically resolves
for a new deterministic policy as time evolves.


Indeed, it has been an important objective to incorporate dynamic strategies into the
Almgren-Chriss framework, but in a more principled way. Some practitioners such as
\citet{Kissell05} have suggested a series of heuristics that are price-adaptive in a particular
way: they liquidate more aggressively when the price moves in a favorable direction, and liquidate
more slowly when the price moves in an adverse direction.  \citet{Almgren07} observed that this
``aggressiveness in-the-money'' (AIM) behavior can strictly improve on the optimal deterministic
strategy in the mean-variance criterion.  In a subsequent work, \citet{Almgren11} develop a
dynamic programming technique by which approximate solutions can be numerically obtained, and show
that \newedit{such approximate solutions exhibit} AIM, despite the lack of an analytic solution.
See also \cite{Forsyth11} for the continuous-time version of this analysis.

Another stream of work (including the present paper) introduces alternative risk criteria other
than the mean-variance objective, so as to formulate the problem into a more tractable form. For
example, \citet{Schied09} formulate the problem as an expected utility maximization problem, and
derive a Hamilton-Jacobi-Bellman (HJB) equation that characterizes the optimal adaptive
strategy.  They find that the optimal strategy is either aggressive- or passive-in-the-money,
respectively, if and only if the utility function displays increasing or decreasing risk aversion,
respectively, but an analytic solution is not available.  \citet{Gatheral11} propose an
alternative risk criterion that utilizes the time-averaged risk exposure to the price change
driven by the geometric Brownian motion (more precisely, the risk term is formulated as the time
integration of the position value process, i.e., the product of the position process and the price
process), and explicitly solve for the optimal strategy that is shown to exhibit the AIM
behavior. However, their proposed risk criteria is ad hoc, and may encourage, for example,
negative positions to reduce risk. Hence, the resulting optimal policies are not
``liquidate-only'': they may trade in both directions (buying and selling) in excess of what
is strictly necessary to reduce the position to zero. \citet{Forsyth12} investigate the use of the
quadratic variation of the position value process as a risk criteria, and observe that the classic
static solution of \citet{Almgren00} is again optimal when the price process is driven by the
arithmetic Brownian motion.  The authors also consider the geometric Brownian motion, under which
the optimal solution is numerically determined and shown to be price-adaptive, and they report its
AIM behavior through numerical examples. \citet{Lin15}
introduces a composite dynamic coherent risk measures and derive the optimal solution that is
tractable but static.  One can also consider an entropic risk measure introduced such as that of
\citet{Glasserman13}, but it can be shown that the resulting strategy is also not price-adaptive.

To summarize, the prior work in this area either (i) yields only approximate numerical solutions;
(ii) incorporates only ad hoc risk criteria; (iii) yields policies that are not liquidate-only and
trade in both directions; or (iv) illustrates no benefit from dynamic policies
over static policies. In contrast, our work simultaneously addresses all of these issues, and, to
our knowledge, is the first paper to do so.

In this paper, we consider the conditional value-at-risk (CVaR; also known as average
value-at-risk, tail conditional expectation, or expected shortfall) as an objective.
CVaR is a risk measure that quantifies the tail risk.  Given a quantile value $q \in (0,1]$
and a random variable that represents the cost, the CVaR value at level $q$ is defined as the
average of the worst $q$-fraction of the outcomes, i.e., the tail average beyond the $q^\text{th}$
quantile of the cost distribution. Smaller values of $q$ focus on performance in the worst cases,
i.e., the trader is more risk averse, while $q=1$, on the other hand, corresponds to a risk
neutral trader. Starting from the pioneering work of \citet{Artzner99}, CVaR has
received much attention for its intuitive definition and for nice mathematical properties as a
convex and coherent risk measure.

From the point of view of using CVaR as an optimization objective, the (static) optimization of
the CVaR value for a single period can be done efficiently, by utilizing its primal variational
representation \citep{Rockafellar02}.  In the multi-period setting of a Markov decision problem
(MDP), however, the dynamic optimization of the CVaR value of the total cost is challenging. To
begin, this objective is time-inconsistent. Moreover, as observed by \citet{Artzner07} and
\citet{Shapiro09}, the optimal action at a point in time may not be completely determined by the
current state of the MDP, but may depend on the entire history, and therefore the conventional
dynamic programming techniques may not work.

More recent work has adopted the idea of state augmentation to overcome this issue: by introducing
an extra state variable, an optimal policy can be sufficiently characterized as a Markov process
defined on this augmented state space.  Broadly speaking, these studies develop CVaR MDP
frameworks using two kinds of state augmentation.  The first kind introduces an extra state
variable that represents the running cost, and derives the dynamic programming principle from the
primal variational representation of CVaR, i.e.,
$\cvar_q[C] = \min_{c \in \R}\{ c + \frac{1}{q} \E[ (C-c)^+ ] \}$ \citep[see][]{Rockafellar02}.
This state augmentation scheme is adopted in, for example,
\citet{Bauerle11,Huang16,Miller17,Chow18,Backhoff20}.  The second kind of state augmentation
introduces an extra state variable that represents the quantile value, and derives the dynamic
programming principle from the dual variational representation of CVaR, i.e.,
$\cvar_q[C] = \sup_{\mathbb{Q} : \mathbb{Q} \ll \mathbb{P}, \frac{d\mathbb{Q}}{d\mathbb{P}} \leq
  q^{-1} }\E_{\mathbb{Q}}[ C ] $ \citep[see][]{Artzner99}.  The work of
\citet{Pflug16,Chow15,Chapman18,Li20} belongs to this category.

\seclabel{This paper.}  In this paper, we seek an adaptive liquidation strategy that minimizes the
CVaR value of the implementation shortfall in a continuous-time, infinite-horizon variation of the
classical Almgren-Chriss framework.
We adopt the second kind of state augmentation described above.  More specifically, we consider an
augmented state space represented as $(X_t, Q_t)$, where $X_t \in \mathbb{R}$ is the current
remaining position size of the trader, and $Q_t \in [0,1]$ is a quantile value that represents the
current level of risk \newedit{neutrality}.  We observe that the dynamic optimization of the CVaR objective
can be represented as a (continuous-time) stochastic game between the trader who controls the
position process $X_t$ and the adversary who controls the quantile process $Q_t$.  By analyzing
the Nash equilibrium of this game, we can identify the minimal CVaR value that the trader can
achieve, and specify the trader's optimal policy and the adversary's optimal policy, which are
formulated as time-stationary Markov policies on $(X_t, Q_t)$.


\seclabel{Practical contributions.}  Using our approach, we can express the optimal dynamic
liquidation strategy in an analytic closed-form. This allows us to characterize properties of the
optimal strategy, and yields a number of insights consistent with how optimal execution algorithms
are widely employed in practice.

First, we show that the optimal strategy is liquidate-only, i.e.,
trades in only one direction and it keeps liquidating until it completes the execution. This is
intuitive since, in practice, traders typically do not want to increase their positions or
establish new positions during the liquidation process. We also observe that the optimal
trading strategy becomes more aggressive and trades more quickly as the level of risk aversion
increases. This is also intuitive since price risk can be minimized by trading quickly.

Second, and more interestingly, we show that the optimal trading strategy is dynamic and
  responds to stochastic fluctuations in the price process in a way that accelerates trading (to
  minimize price risk) when price movements are favorable, and slows trading (to minimize
  transaction costs) when price movements are unfavorable --- in other words, it exhibits
  aggressiveness-in-the-money (AIM).  The intuition behind the optimality of AIM is as follows: 

 Recall that the CVaR measure is the tail average beyond the $q^\text{th}$ quantile ---
  the CVaR objective is not impacted by cost realizations below this threshold but only requires
  the average cost above this threshold to be minimized.  As an extreme case, imagine a situation
  when the price movements have been so favorable that the trader has earned a large profit from
  the holding positions and the current cumulative cost is far below the quantile threshold.  The
  CVaR objective of a trader in this scenario is not impacted by the cost (so long as it remains
  below the threshold) and the trader should thus be willing to pay a large, deterministic
  transaction cost (up to the gap between the current cumulative cost and the threshold) so as to
  complete the execution quickly and minimize the risk of the total cost exceeding the threshold.
  In the opposite case, imagine a scenario where the current cumulative cost is far above the
  threshold. Since the total cost is very likely to ultimately exceed the threshold, the CVaR
  objective of the trader becomes the expected cost and is approximately risk neutral. The trader
  is thus encouraged to slow down trading so as to minimize (deterministic) transaction costs
  thereafter.  As illustrated in these two extreme cases, the optimal strategy responds
  asymmetrically to price changes, in the way that it becomes more aggressive in an adverse
  situation, due to the intrinsic asymmetry of the CVaR measure.

Our stochastic game interpretation of this dynamic CVaR optimization problem provides an alternative (and more formal) justification of AIM.
On the augmented state space, the adversary's quantile process $Q_t$ represents the current level
of risk neutrality (starting at value $q$), and is coupled with the price process.
In particular, the adversary's optimal quantile process $Q_t^\star$ represents the probability that the current sample path is among the $q$-fraction of worst outcomes for the trader, and hence, $Q_t^\star$ increases when the price moves in an unfavorable direction.
When $Q_t^\star$ increases, consequently, the trader becomes more risk-neutral and is encouraged
to trade slowly, which is consistent with AIM.

Third, we are also able to quantify the benefit of dynamic trading by comparing the optimal dynamic
policy to two benchmark static policies: the optimal static policy and an optimized version of a
policy that trades at a constant rate. The later policy is meant to be representative of volume
weighted average price (VWAP) polices that are popular in practice. The relative improvement of
the optimal dynamic policy depends on the problem parameters only through the risk aversion $q$,
and can be characterized analytically. For moderate levels of risk aversion, the optimal dynamic
policy outperforms the optimal static policy by 5--15\%, and outperforms the optimal VWAP policy
by 15--25\%. Moreover, relative improvement is increasing in $q$ and grows without bound as
$q \nearrow 1$, i.e., as the investor becomes more risk neutral. Since most traders using optimal
execution algorithms are large investors trading a small portion of their overall portfolio over a
short time horizon, their utility functions are nearly linear from the perspective of the optimal
execution problem, hence the nearly risk neutral regime ($q \approx 1$), where the relative
benefits of dynamic trading are greatest, is also the most practically relevant regime.

Last, numerical experiments yield insight to some useful features of the dynamic optimal policy.
Namely, this policy effectively controls the worst outcomes in the (right) tail of the cost distribution, resulting in a cost distribution is very distinguished from the ones induced by deterministic policies.
This feature may appeal to the practitioners who are not particularly interested in optimizing the CVaR value.
For example, we observe that the class of CVaR-optimal dynamic policies achieves better median performance and better tail probabilities than the class of deterministic policies.
See Figure~\ref{fig:numeric-frontiers} and the accompanying discussion for further details.



\seclabel{Mathematical contributions.}  We introduce a novel and technically sound approach to
developing the CVaR MDP framework in the continuous-time setting.  To sketch our approach briefly,
we first introduce a scaled version of CVaR that allows us to avoid the ambiguity of CVaR in the
corner case (i.e., when $q=0$) and inherently induces the concavity of the objective
(Proposition~\ref{prop:scvar-properties}).  We then utilize the martingale representation theorem,
by which we can rewrite the CVaR objective as a maximization problem for an adversary who controls
the quantile process $Q_t$ against the decision maker
(Theorem~\ref{thm:martingale-representation}), and the problem can be translated into a
continuous-time stochastic game between the decision maker and the adversary who are competing
over the expected value of the risk-adjusted outcome (Theorem~\ref{thm:cvar-minimax}).  After this
step, the objective now decomposes over time, hence we can develop a CVaR dynamic programming
principle (Theorem~\ref{thm:cvar-dp}).  This naturally leads to the Hamilton-Jacobi-Bellman (HJB) partial differential
equations that characterize the optimal solution. The HJB equation is separable across the state
variables, and hence admits a closed form solution for the value function.

From the HJB equation, we are able to define candidate optimal control policies for the trader and
adversary. Unfortunately, we are not able to show that these policies as feasible. Instead, we
construct a series of feasible policies that converge pointwise to the candidate optimal policy
and in value to the optimal value function (see Theorem~\ref{thm:policy-optimality} and the
accompanying discussion).

Our approach leverages the idea of state augmentation that incorporates the quantile value as an
extra state variable, which has been suggested and utilized in prior work \citep{Pflug16, Chow15,
  Chapman18, Li20}, but in exclusively discrete-time settings. To our knowledge, our work is
the first to apply this in continuous-time, which introduces new technical challenges. We clearly
take advantage of the continuous-time setting: the martingale representation theorem, which is
intrinsically continuous-time, is fundamental in our analysis as it allows us to parameterize the
adversary's control policy with a real-valued stochastic process that can be tractably optimized.
Moreover, the choice of continuous-time is critical in this application, since the optimization of
an analogous discrete-time model would involve a Bellman equation that is not separable, and hence
is unlikely to admit an analytic solution (see the discussion at the end of \S
\ref{ssec:martingale-representation}).

As our problem is a special case of a continuous-time linear-quadratic-Gaussian control problem
with a CVaR objective, these technical results may be interesting in broader settings, and feed
into the broader emerging literature of continuous-time differential games.


\seclabel{Organization of paper.}
In \S \ref{sec:problem}, we introduce the notation and formally describe the model and the problem.
In \S \ref{sec:cvar-dp}, we develop the CVaR MDP framework for which we sequentially introduce a martingale representation of the CVaR objective, the game-theoretic representation of the problem, and the Markov policies defined on the augmented state space.
In \S \ref{sec:opt}, we derive the HJB equation, identify its solution, and characterize the optimal liquidation strategy.
In \S \ref{sec:cost-analysis}, we compare the optimal adaptive strategy with two deterministic strategies: the optimal deterministic strategy and the optimized VWAP strategy.
In \S \ref{sec:numerical}, we provide simulation results that illustrate the optimal strategy.
In the appendix, we provide the proofs that are deferred from \S \ref{sec:problem}--\S \ref{sec:cost-analysis}.

\section{Problem} \label{sec:problem}

We consider a filtered probability space $\big( \Omega, \Fscr, \mathbb{F} = ( \Fscr_t )_{t \geq 0}, \mathbb{P} \big)$, where $\mathbb{F}$ is a natural filtration of a Brownian motion $( W_t )_{t \geq 0}$ that satisfies the usual conditions.
We denote by $\Lscr^p( \Omega, \Fscr, \mathbb{P} )$ (or simply by $\Lscr^p$) the set of $\Fscr$--measurable random variables $X : \Omega \rightarrow \mathbb{R}$ such that $\mathbb{E}|X|^p < \infty$.
Given a sequence of random variables $( X_n )_{n \in \mathbb{N}}$, it is said that $X_n \stackrel{ \Lscr^p }{ \rightarrow } X$ if $\lim_{n \rightarrow \infty} \E|X_n - X|^p = 0$.
The time index set is denoted by $\mathbb{T} \defeq [0,\infty)$.
We also define $\Prog$ to be the set of progressively measurable stochastic processes in this filtered probability space.
We denote $\mathbb{R}_+ \defeq [0, \infty)$ and $\mathbb{R}_- \defeq (-\infty, 0]$.

\subsection{Model} \label{sec:model}

We consider a continuous-time and infinite-horizon version of the setting of \citet{Almgren00}.
We postulate a trader who wants to liquidate $x \in \mathbb{X}$ units of an asset over an
\emph{infinite-time horizon}.\footnote{\newedit{The choice of an infinite-horizon setting is
    important since as it will allow an analytical solution of the problem. In particular, in an
    infinite-horizon setting, the optimal policy is stationary over time, and hence time can be
    eliminated as a state variable. As we will see in subsequent sections, this results in a
    simpler Hamilton-Jacobi-Bellman (HJB) equation that can be solved analytically. Beyond this,
    in some sense the infinite-horizon setting is more elegant from a modeling perspective since
    it endogenizes the effective time horizon of the trader dynamically as a function of risk
    aversion.  That said, the optimal policy for a finite-horizon variation of our setting could
    be obtained via a numerical solution of the HJB equation, and we expect the qualitative
    insights to be no different.  }} Here, $x$ can be negative if the trader wants to acquire the
asset.  We define $\mathbb{X} \defeq [-M, M] \subset \mathbb{R}$, where $M \geq 0$ is an arbitrary
large number\footnote{We require that the position size be bounded in order to resolve technical
  difficulties (see Remark~\ref{rem:Mbound}).}  that represents the largest possible position size
that the trader is allowed to own.

\seclabel{Liquidation strategy.}
The trader's liquidation policy is represented with a real-valued continuous-time stochastic process $\pi \defeq ( \pi_t )_{t \geq 0}$, where $\pi_t \in \mathbb{R}$ specifies the \emph{liquidation rate} at time $t$.
Given an initial position size $x$ and a liquidation strategy $\pi$, the trader's \emph{position process} $X^{x,\pi} \defeq (X_t^{x,\pi})_{t \geq 0}$ is determined by
\begin{equation}
	X_t^{x,\pi} = x - \int_{s=0}^t \pi_s ds;
\end{equation}
i.e., the trader liquidates $\pi_t$ units of the asset per unit time (or acquires $-\pi_t$ units of the asset per unit time if $\pi_t < 0$).
While deferring a formal statement to the end of this subsection, we will restrict our attention to the policies under which the trader's position varies continuously over time (i.e., involves no impulse trades) and vanishes eventually (i.e., converges to zero as $t$ goes to infinity).

\seclabel{Liquidation cost.}
Following the framework of \citet{Almgren00}, we define the \emph{cost process} (or loss process) $C^{x,\pi} \defeq (C_t^{x,\pi})_{t \geq 0}$ as
\begin{equation} \label{eq:loss}
	C_t^{x,\pi} \defeq \int_{s=0}^t \frac{1}{2} \eta \pi_s^2 ds - \int_{s=0}^t \sigma X_s^{x,\pi} dW_s.
\end{equation}
The first term $\int_{s=0}^t \frac{1}{2} \eta \pi_s^2 ds$ represents the (cumulative) transaction cost incurred by the temporary price impact, where the coefficient $\eta > 0$ reflects the illiquidity of the asset.
The second term $\int_{s=0}^t \sigma X_s^{x,\pi} dW_s$ represents the (cumulative) loss incurred by the random fluctuation of the market price, where $\sigma > 0$ is the volatility of the price process and $W_t$ is the standard Brownian motion.
The total cost (i.e., total implementation shortfall) $C_\infty^{x,\pi} \defeq \lim_{t \rightarrow \infty} C_t^{x,\pi}$ is a random variable of interest that we want to minimize via a CVaR objective.


Note that we do not consider permanent price impact in our formulation.  Within the
continuous-time Almgren--Chriss framework, it is well known that the contribution of permanent
linear price impact to the implementation shortfall is path-independent; i.e., it does not depend
on the liquidation strategy as long as the strategy clears all the positions eventually.\footnote{
  Suppose that liquidating the asset at rate $v$ permanently shifts the market price at rate
  $-\lambda v$, i.e., permanent linear price impact.  Given a feasible policy $\pi$, the
  associated transaction costs can be expressed $\int_{t=0}^\infty \lambda \pi_t X_t^{x,\pi} dt$,
  which is always equal to $\frac{1}{2} \lambda x^2$, independent of the policy, given that
  $\lim_{t \rightarrow \infty} X_t^{x,\pi} = 0$.  } Therefore, without loss of generality, we can
ignore the presence of permanent impact since it does not make any difference in the trader's
decision making.  See \cite{Almgren00} for details.

\seclabel{Admissible strategies.}
We now formally define the \emph{set of admissible policies} $\Pi(x)$ as
	\begin{equation} \label{eq:Pi}
		\Pi(x) \defeq \left\{ \pi : \mathbb{T} \times \Omega \rightarrow \mathbb{R} \left| 
			\begin{array} {l}
				\pi \in \Prog, \\
				\E\left[ \left( \int_{t=0}^\infty \pi_t^2 dt \right)^2 \right] < \infty, 
				~ \E\left[ \int_{t=0}^\infty |X_t^{x,\pi}|^2 dt \right] < \infty, \\
				X_t^{x,\pi} \in \mathbb{X}, ~ \forall t \geq 0
			\end{array}
		\right. \right\}
		.
	\end{equation}
An admissible policy $\pi \in \Pi(x)$ can be dynamic and so it can adjust the trading rate adaptively to the price changes.
By this definition, impulse trades are not allowed by the constraint $\E\left[ \left( \int_{t=0}^\infty \pi_t^2 dt \right)^2 \right] < \infty$, and a non-vanishing position is also not allowed by the constraint $\E\left[ \int_{t=0}^\infty |X_t^{x,\pi}|^2 dt \right] < \infty$.
These conditions are to guarantee that the limit $C_\infty^{x,\pi} \defeq \lim_{t \rightarrow \infty} C_t^{x,\pi}$ is an integrable random variable: individual terms in \eqref{eq:loss} converge in $\Lscr^2$ as $t \rightarrow \infty$, and therefore $C_t^{x,\pi} \stackrel{ \Lscr^2 }{ \rightarrow } C_\infty^{x,\pi}$ for some $C_\infty^{x,\pi} \in \Lscr^2$.

\begin{remark}\label{rem:Mbound}
  The last condition $X_t^{x,\pi} \in \mathbb{X} \defeq [-M, M]$ is seemingly restrictive,
  but enforcing boundedness of the position path is merely a technical condition necessary for our
  later analysis. In particular, we will see that under the optimal policy, the position size will
  be monotonically decreasing towards zero over time. Therefore, we need only choose $M$ so that
  the initial position size is feasible, i.e., $|x| \leq M$. Beyond this, the choice of $M$ will
  not be a binding constraint: it will neither influence the optimal policy nor the optimal value
  function.
\end{remark}

\subsection{Scaled Conditional Value-at-Risk} \label{ssec:scvar}


The \emph{conditional value-at-risk} (CVaR) at a quantile level $q \in (0,1]$ is a mapping from $\Lscr^1(\Omega, \Fscr, \mathbb{P})$ to $\mathbb{R}$.
While there exist several definitions of CVaR in the literature, we consider the one based on its dual representation as a coherent risk measure \citep{Artzner99, Shapiro09}: given a random variable $C \in \Lscr^1$, 
\begin{equation} \label{eq:cvar}
	\cvar_q[C] \defeq \sup_{Q^\dagger \in \Qscr^\dagger(q)} \E\left[ C Q^\dagger \right]
	\quad \text{where} \quad
	\Qscr^\dagger(q) \defeq \left\{  Q^\dagger \in \Lscr^\infty(\Omega, \Fscr, \mathbb{P}), ~ 0 \leq Q^\dagger \leq \frac{1}{q}, ~ \E Q^\dagger = 1 \right\},
\end{equation}
i.e., it is defined as a maximization over a random variable $Q^\dagger$ that has bounded support $[0, 1/q]$ and an expected value of one (or equivalently, a maximization over a probability measure such that it is absolutely continuous with respect to $\mathbb{P}$ and its Radon--Nikodyn derivative with respect to $\mathbb{P}$ is upper bounded by $1/q$ almost surely).

We introduce a scaled version of the CVaR measure, which surprisingly simplifies our analysis.

\begin{definition}[Scaled Conditional Value-at-Risk]
	Given a random variable $C \in \Lscr^1(\Omega, \Fscr, \mathbb{P})$ and a quantile $q \in [0,1]$, the \emph{scaled conditional value-at-risk} (S-CVaR) at level $q$ is
	\begin{equation} \label{eq:scvar}
		\scvar_q[C] \defeq \sup_{Q \in \Qscr(q)} \E\left[ C Q \right],
	\end{equation}
	where $\Qscr(q)$ represents \emph{the risk envelope}:
	\begin{equation}
		\Qscr(q) \defeq \left\{  Q \in \Lscr^\infty( \Omega, \Fscr, \mathbb{P} ) \left| 
			~ 0 \leq Q \leq 1, 
			~ \E Q = q
			\right. \right\}.
	\end{equation}
\end{definition}

The S-CVaR measure is obtained by simply scaling the risk envelope of CVaR by $q$.
As a result, $\scvar_q[C]$ is also given by $q \times \cvar_q[C]$ for $q \ne 0$.
Nevertheless, it naturally incorporates the case $q=0$ into its definition, for which CVaR is not well defined.\footnote{
	When $q=0$, $\cvar_0[C]$ is typically defined as the essential supremum of $C \in \Lscr^1$, which can be infinite if the loss distribution has an unbounded support.
	We have defined the S-CVaR measure using the dual representation of CVaR so that we can effectively avoid the ambiguity at $q=0$.
}
It further has the following useful properties:
\begin{proposition}[Properties of S-CVaR] \label{prop:scvar-properties}
	For any random variable $C \in \Lscr^1$, $\scvar_q[C]$ satisfies the following properties:
	\begin{enumerate}[label=(\roman*)]
		\item \label{it:scvar-cvar-relationship} $\scvar_q[C] = q \times \cvar_q[C]$, for any $q \in (0,1]$.
		\item \label{it:scvar-boundary} $\scvar_0[C] = 0$ and $\scvar_1[C]= \E C$.
		\item \label{it:scvar-bounds} $\left| \scvar_q[C] \right| \leq \E |C|$ and $\scvar_q[C] \geq q \mathbb{E}C$ for any $q \in [0,1]$.
		\item \label{it:scvar-concavity} The mapping $q \mapsto \scvar_q[C]$ is concave
                  and continuous on
                  $[0,1]$.
		\item \label{it:scvar-interpretation} Suppose that $C$ is a continuous random variable whose distribution is atomless. Then,
			\begin{equation}
				\scvar_q[C] = \E\left[ C \I{ C \geq F_C^{-1}(1-q) } \right],
			\end{equation}
			where $F_C^{-1}(\cdot)$ is the inverse distribution function of $C$, and $\cvar_q[C] = \E\left[ C \left| C \geq F_C^{-1}(1-q) \right. \right]$.
	\end{enumerate}
\end{proposition}
The proof can be found in Appendix \ref{app:proof-problem}.
Properties \ref{it:scvar-cvar-relationship}--\ref{it:scvar-bounds} provide basic characterizations of S-CVaR.
The property \ref{it:scvar-interpretation} provides an interpretation of S-CVaR as a truncated average as opposed to the interpretation of CVaR as a conditional average.
We particularly highlight property \ref{it:scvar-concavity} that shows the concavity and continuity of the S-CVaR value with respect to $q$, which is a crucial property that will be exploited in our analysis.

One can interpret the definition \eqref{eq:scvar} as a maximization problem for an adversary.
This adversary selects a set of scenarios so as to maximize the average cost within the selected scenario, given a constraint that the total measure of the selected scenarios should be $q$.
Informally,\footnote{When the loss distribution has an atom at the $q^\text{th}$ quantile, the extremal random variable $Q^\star(\omega)$ may take a fractional value.} the optimized random variable $Q^\star$ is an indicator random variable such that $Q^\star(\omega) = 1$ if the scenario $\omega$ is among the worst $q$-fraction of the scenarios, and $Q^\star(\omega) = 0$ otherwise.

\subsection{Risk-sensitive Execution with a CVaR Objective}

We now introduce the CVaR risk criterion into the setting described in \S \ref{sec:model}.
In particular, we seek an adaptive strategy that minimizes the CVaR value of the implementation shortfall, given an initial position $x \in \mathbb{X}$ and a target quantile $q \in (0,1]$.
Without loss of generality, we formulate this optimization problem via an S-CVaR objective and define the \emph{value function} $V:\mathbb{X} \times [0,1] \rightarrow \mathbb{R}$ as
\begin{equation} \label{eq:V}
	V(x,q) \defeq \inf_{\pi \in \Pi(x)} \scvar_q\left[ C_\infty^{x,\pi} \right].
	\tag{$*$}
\end{equation}
Note that the above formulation includes the case $q=0$.
By Proposition \ref{prop:scvar-properties}, the value function $V(x,q)$ is well-defined at $q=0$, and the minimal CVaR value is simply given by $V(x,q)/q$ for any $q \ne 0$.
We aim to identify the optimal value function $V(x,q)$ as well as its corresponding optimal liquidation strategy $\pi^\star$.

Recall that the objective $\scvar_q[ C_\infty^{x,\pi} ]$ concerns the worst $q$-fraction of outcomes.
When $q=1$, the problem reduces to a risk-neutral liquidation problem.
When $q$ takes a smaller value, the problem is equivalent to considering a more risk-averse trader who concerns a smaller fraction of worst cases, being wary of more extreme cases.
We anticipate that the trader uses this quantile value $q \in [0,1]$ as an input to our algorithm so as to control the level of risk-aversion that he wants to achieve.
In practice, we do not expect the traders to use an extremely small quantile value such as $q = 0.01$ or $q = 0.05$: since they encounter this sort of liquidation task often, possibly on a daily basis, it would be too conservative for them to optimize their performance in the worst $1\%$ or $5\%$ of cases at a cost of sacrificing their performance in the normal $99\%$ or $95\%$ of cases.

\section{CVaR Dynamic Programming Principle} \label{sec:cvar-dp}

Using the definition of the S-CVaR measure \eqref{eq:scvar}, the risk-sensitive optimal execution problem \eqref{eq:V} can be formulated as
\begin{equation}
        V(x,q) = \inf_{\pi \in \Pi(x)} \sup_{Q \in \Qscr(q)} \E\left[ C_\infty^{x,\pi} Q \right] .
\end{equation}
As discussed in \S \ref{ssec:scvar}, we can think of an adversary who optimizes a random variable $Q \in \Qscr(q)$ so as to select the worst $q$-fraction of sample paths against the trader who employs a liquidation policy $\pi \in \Pi(x)$.
In this section, we reformulate the adversary's optimization problem as an optimization over a real-valued continuous-time stochastic process rather than a random variable, and interpret the risk-sensitive optimal control problem as a continuous-time stochastic game between the trader and the adversary.
To this end, we develop a continuous-time dynamic programming principle by exploiting the recursive structure of this game.



\subsection{Martingale Representation of CVaR Objective} \label{ssec:martingale-representation}


We consider an arbitrary random variable $C \in \Lscr^1(\Omega, \Fscr, \mathbb{P})$ and derive an alternative representation of $\scvar_q[C]$.
The results in this subsection are valid not only in the context of the liquidation problem, but also in any filtered probability space generated by a Brownian motion.

We define the \emph{adversary's policy} as a real-valued continuous-time stochastic process $\gamma \defeq ( \gamma_t )_{t \geq 0}$, which determines the \emph{adversary's quantile process} $Q^{q,\gamma} \defeq (Q_t^{q,\gamma})_{t \geq 0}$:
\begin{equation}
        Q_t^{q,\gamma} = q + \int_{s=0}^t \gamma_s \, dW_s,
\end{equation}
where $W = (W_t)_{t \geq 0}$ is the Brownian motion that drives the random price fluctuation.
We sometimes call $\gamma_t$ the \textit{quantile diffusion rate} by analogy to the liquidation rate $\pi_t$.

The \emph{set of admissible adversary's policies} is defined as
        \begin{equation} \label{eq:Gamma}
                \Gamma(q) \defeq \left\{ \gamma : \mathbb{T} \times \Omega \rightarrow \mathbb{R} ~
                \left|
                        ~ \gamma \in \Prog,
                        ~ 0 \leq Q_t^{q,\gamma} \leq 1, \forall t \geq 0
                \right. \right\}.
        \end{equation}
Given an admissible adversary's policy $\gamma \in \Gamma(q)$, its corresponding quantile process $Q^{q,\gamma}$ is a (local) martingale starting at $q \in [0,1]$ whose diffusion term is governed by $\gamma$.
In particular, the quantile process is required to take values within $[0,1]$ and, as a result, it has the following properties (the proof is provided in Appendix \ref{app:proof-cvar-dp}):
\begin{proposition}[Properties of the adversary's quantile process $Q$] \label{prop:Q-properties}
        For any $\gamma \in \Gamma(q)$,
        \begin{enumerate}[label=(\roman*)]
                \item \label{it:Q-martingale} $(Q_t^{q,\gamma})_{t \geq 0}$ is a continuous and
                  bounded martingale taking values in $[0,1]$, and hence $\E[Q_\tau^{q,\gamma}] =
                  q$ for any stopping time $\tau$.
                \item \label{it:Q-limit} $Q_\infty^{q,\gamma} \defeq \lim_{t \rightarrow \infty} Q_t^{q,\gamma}$ exists in $[0,1]$ almost surely, and also $\E[Q_\infty^{q,\gamma}] = q$.
                \item \label{it:Q-absorption} Once $Q_t^{q,\gamma}$ hits $0$ or $1$, it never escapes thereafter.
        \end{enumerate}
\end{proposition}

By the martingale representation theorem, any random variable $Q$ in the risk envelope $\mathcal{Q}(q)$ can be represented as the limit of a quantile process $Q^{q,\gamma}$ for some $\gamma \in \Gamma(q)$, and vice versa:
\begin{lemma}\label{lem:isomorphism}
        For any $q \in [0,1]$, $\Qscr(q) = \left\{ Q_\infty^{q,\gamma} | \gamma \in \Gamma(q) \right\}$.
\end{lemma}

\begin{proof}
        Consider an arbitrary $\tilde{Q} \in \Qscr(q)$, and Doob martingale $(Q_t)_{t \geq 0}$ generated by $\tilde{Q}$, i.e., $Q_t \defeq \E[ \tilde{Q} | \Fscr_t ]$ for each $t$ (we have $\lim_{t \rightarrow \infty} Q_t = \tilde{Q}$ and $Q_0 = \E\tilde{Q} = q$).
        By the martingale representation theorem \citep[Thm.~43 in Chap.~IV]{Protter}, there exists a predictable process $\gamma$ such that $Q_t = Q_0 + \int_{s=0}^t \gamma_s dW_s$.
        Since $Q_t = \E[ \tilde{Q} | \Fscr_t] \in [0,1]$ for any $t$, we have an admissible adversary's policy $\gamma \in \Gamma(q)$ and therefore $\tilde{Q} \in \left\{ Q_\infty^{q,\gamma} | \gamma \in \Gamma(q) \right\}$ and $\Qscr(q) \subseteq \left\{ Q_\infty^{q,\gamma} | \gamma \in \Gamma(q) \right\}$.

        Now consider an arbitrary $\gamma \in \Gamma(q)$ and let $\tilde{Q} \defeq Q_\infty^{q,\gamma}$.
        Trivially, $\E \tilde{Q} = Q_0 = q$ and $\tilde{Q} \in [0,1]$.
        Therefore, $Q_\infty^{q,\gamma} \in \Qscr(q)$, and hence $\Qscr(q) \supseteq \left\{ Q_\infty^{q,\gamma} | \gamma \in \Gamma(q) \right\}$.
\end{proof}

This alternative representation of the risk envelope $\mathcal{Q}(q)$ immediately leads to the following representation of the S-CVaR value, under which the adversary optimizes over the set of stochastic processes instead of the set of random variables:
\begin{theorem}[Martingale representation of the CVaR objective] \label{thm:martingale-representation}
        For any given random variable $C \in \Lscr^1(\Omega, \Fscr, \mathbb{P})$ and $q \in [0,1]$, we have
        \begin{equation} \label{eq:martingale-representation}
                \scvar_q[ C ] = \sup_{\gamma \in \Gamma(q)} \E\left[ C Q_\infty^{q,\gamma} \right].
        \end{equation}
\end{theorem}
\begin{proof}
        The result immediately follows from Lemma \ref{lem:isomorphism} and the definition of S-CVaR \eqref{eq:scvar}.
\end{proof}

To better understand this adversary's optimization problem, consider a discrete-time setting with two periods, six sample paths and $q=1/2$, as illustrated in Figure~\ref{fig:cartoon-martingale}.
The adversary is asked to assign the quantile values on the individual nodes so as to maximize the objective, $\E[CQ_2]$, given a constraint that the resulting quantile process $(Q_0, Q_1, Q_2)$ needs to be a martingale.

\begin{figure}[H]
\centering
\begin{tikzpicture}
	\node[draw, fill=lightgray!50] (s0) at (0,0) 									{$Q_0 = \frac{1}{2}$};
	
	\node[draw, xshift=4cm, yshift=1.5cm, fill=lightgray!33] (s1u) at (s0.center) 		{$Q_1 = \frac{1}{3} $};
	\node[draw, xshift=4cm, yshift=-1.5cm, fill=lightgray!66] (s1d) at (s0.center) 		{$ Q_1 = \frac{2}{3}$};
	
	\node[draw, xshift=4cm, yshift=0.7cm] (s2uu) at (s1u.center) 					{$Q_2 = 0$};
	\node[draw, xshift=4cm, yshift=0cm] (s2um) at (s1u.center) 					{$Q_2 = 0$};
	\node[draw, xshift=4cm, yshift=-0.7cm, fill=lightgray] (s2ud) at (s1u.center) 		{$Q_2 = 1$};
	
	\node[draw, xshift=4cm, yshift=0.7cm] (s2du) at (s1d.center) 					{$Q_2 = 0$};
	\node[draw, xshift=4cm, yshift=0cm, fill=lightgray] (s2dm) at (s1d.center) 			{$Q_2 = 1$};
	\node[draw, xshift=4cm, yshift=-0.7cm, fill=lightgray] (s2dd) at (s1d.center) 		{$Q_2 =1$};
	
	\node[anchor=west] at (s2uu.east) 					{$C = -12$};
	\node[anchor=west] at (s2um.east) 					{$C = -5$};
	\node[anchor=west] at (s2ud.east) 					{$C = +2$};
	
	\node[anchor=west] at (s2du.east) 					{$C = -2$};
	\node[anchor=west] at (s2dm.east) 					{$C = +5$};
	\node[anchor=west] at (s2dd.east) 					{$C = +12$};
	
	\node at (s0 |- 0,-3cm) {$t=0$};
	\node at (s1u |- 0,-3cm) {$t=1$};
	\node at (s2uu |- 0,-3cm) {$t=2$};
	
	\draw[->] (s0.east) -- (s1u.west);
	\draw[->] (s0.east) -- (s1d.west);
	\draw[->] (s1u.east) -- (s2uu.west);
	\draw[->] (s1u.east) -- (s2um.west);
	\draw[->] (s1u.east) -- (s2ud.west);
	\draw[->] (s1d.east) -- (s2du.west);
	\draw[->] (s1d.east) -- (s2dm.west);
	\draw[->] (s1d.east) -- (s2dd.west);
\end{tikzpicture}
\caption{
        An illustration of the discrete-time version of the adversary's martingale in a two-period setting with $q=\frac{1}{2}$.
        The underlying probability space is represented with a tree, in which each node represents the current state of the system, and each branch represents a possible transition from one state to another.
        We assume that six sample paths are equally likely to be realized.
        The label next to each terminal node represents the total cost $C$ incurred along each sample path.
        The value in each node represents the value of the (optimal) adversary's martingale process at each state.
}
\label{fig:cartoon-martingale}
\end{figure}
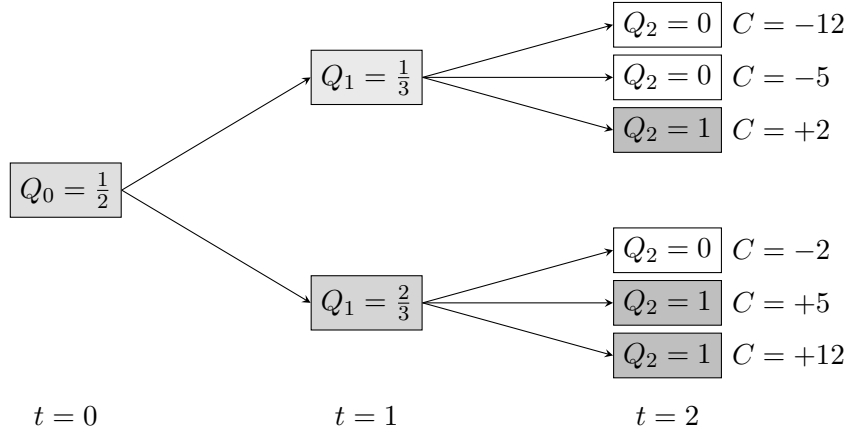

It is easy to see that the adversary's optimal solution is to select the worst $q$-fraction of sample paths, i.e., to assign the quantile value $1$ to the $q$-fraction of terminal nodes with the largest realized cost and assign the quantile value $0$ to the rest of terminal nodes.
The quantile values of the non-terminal nodes are sequentially determined via backward induction (from $t=2$ to $t=0$) by averaging the quantile values of subsequent nodes.
Here, the quantile value at each node means how likely the current sample path ends at one of the terminal nodes selected by the adversary.\footnote{
Consider the optimal martingale $Q^\star$ that solves \eqref{eq:martingale-representation}, and recall that the random variable $Q_\infty^\star(\omega)$ indicates whether the sample path $\omega$ is among the worst $q$-fraction of scenarios.
As a Doob martingale, the quantile process $Q_t^\star = \E[Q_\infty^\star | \Fscr_t ]$ is its running estizmate at time $t$, i.e., the likelihood that the current sample path will end up with one of the worst scenarios.
}

We can alternatively interpret this optimization problem as a sequential decision-making problem that the adversary assigns the quantile values in a forward direction (from $t=0$ to $t=2$).
It begins with assigning the quantile value $q$ to the root node.
Starting from the root node, the adversary is asked to allocate quantile values to the subsequent nodes given a budget constraint that the average of the quantile values assigned to these nodes should equal the quantile value at the current node.
This is repeated until all terminal nodes get assigned their quantile values.
Here, the quantile value at each node means how much fraction of sample paths can be selected thereafter, and the budget constraint guarantees the resulting quantile process be a martingale.

The latter interpretation is closed related to Theorem \ref{thm:martingale-representation}.
Indeed, for example, can consider a binomial tree model as a discrete-time approximation of the underlying Brownian motion, in which every node has two branches representing the events of price moving up or down by a small amount.
In this binomial approximation, the quantile value allocation problem that the adversary solves at each node involves only one decision variable: given the current quantile value $Q_t$, the allocation of next quantile values is of the form $\big\{ (1 + \theta_t) Q_t, (1-\theta_t) Q_t \big\}$, where $\theta_t$ is the real-valued decision variable.
Here, determining the value of $\theta_t$ is effectively equivalent to determining the diffusion rate of $Q_t$, which is in fact the adversary's policy $\gamma_t$ in our formulation.

We highlight that this martingale representation greatly simplifies the adversary's optimization problem and is only available in a continuous-time setting.
In the discrete-time setting investigated in \citet{Pflug16,Chow15,Chapman18,Li20}, for example, the adversary in each period needs to solve a much more complicated optimization that may involve many decision variables (if there are $n$ possible random outcomes at a state, the quantile value allocation problem involves $n-1$ real-valued decision variables).
It may not easy in such cases to solve the adversary's optimization problem even numerically.

\subsection{Risk-sensitive Liquidation as a Continuous-time Stochastic Game}

We now return to the liquidation problem, and define an \emph{outcome function} $J$ as a function of the trader's policy $\pi \in \Pi(x)$ and the adversary's policy $\gamma \in \Gamma(q)$ at each $x \in \mathbb{X}$ and $q \in [0,1]$:
        \begin{equation}
                J(\pi, \gamma; x, q) \defeq \E\left[ C_\infty^{x,\pi} Q_\infty^{q,\gamma} \right].
        \end{equation}
By Theorem \ref{thm:martingale-representation}, the value function \eqref{eq:V} can be formulated as
        \begin{equation}\label{eq:V2}
                V(x,q) = \inf_{\pi \in \Pi(x)} \sup_{\gamma \in \Gamma(q)} J( \pi, \gamma; x, q ).
        \end{equation}
The following theorem characterizes this value function as an equilibrium outcome of a continuous-time stochastic game between the trader and the adversary.

\begin{theorem}[CVaR optimization as a continuous-time stochastic game] \label{thm:cvar-minimax}
        The value function $V(x,q)$ is the outcome at the Nash equilibrium of the zero-sum game in which the trader and the adversary compete over the outcome $J$,
        \begin{equation}
                V(x,q)
                        = \inf_{\pi \in \Pi(x)} \max_{\gamma \in \Gamma(q)} J( \pi, \gamma; x, q )
                        = \max_{\gamma \in \Gamma(q)} \inf_{\pi \in \Pi(x)}  J( \pi, \gamma; x, q ),
        \end{equation}
        where an optimal solution $\gamma \in \Gamma(q)$ must exist for each maximization.
\end{theorem}

Theorem \ref{thm:cvar-minimax} states that the minimax solution equals the maximin solution.
This means that the value function is, as a saddle point, the equilibrium outcome at which each player simultaneously plays the best response against the other player's strategy.
This may not always hold true for a general class of risk-sensitive control problems: the convexity of the outcome function with respect to the trader's policy and the convexity of the policy space are required in our proof (Appendix \ref{app:cvar-minimax-proof}) in order to apply Sion's minimax theorem \citep{Sion58}.

The following remark provides an alternative interpretation of this game based on the Girsanov theorem.
\begin{remark} \label{rem:Girsanov} Recall that the dual representation of the CVaR
    objective \eqref{eq:cvar} involves a multiplication with a random variable $Q^\dagger$, which
    corresponds to an absolutely continuous change of measure  with Radon-Nikodym derivative $Q^\dagger$.
    In terms of the adversary's quantile process
    $Q^{q,\gamma}$, this change of measure can by written as $Q^\dagger \defeq Q_\infty^{q, \gamma}/q$.
    By the Girsanov theorem, the price process has a drift under this new measure, and is given by
    \[
      dP = \sigma d\tilde W_t + \sigma \frac{\gamma_t}{Q_t^{q,\gamma}} dt,
    \]
    where $\tilde W$ is a Brownian motion under the new measure.
    In other words, the trader's
    optimization problem given the adversary's policy $\gamma$ is can also be viewed as a risk-neutral
    liquidation problem under altered price dynamics in the presence of
    adversary who can introduce drift into the price process.
  
\end{remark}

\subsection{CVaR Dynamic Programming Principle}

In this subsection, we develop a continuous-time dynamic programming principle.
We first state a proposition that characterizes a temporal structure of the game.

\begin{proposition}[Time decomposition] \label{prop:time-decomposition}
        Fix $x \in \mathbb{X}$ and $q \in [0,1]$.
        For any trader's policy $\pi \in \Pi(x)$, adversary's policy $\gamma \in \Gamma(q)$, and stopping time $\tau$, we have
        \begin{equation}
                J( \pi, \gamma; x, q ) = \E\left[ C_\tau^{x,\pi} Q_\tau^{q,\gamma} + \E\left[ \left. \big( C_\infty^{x,\pi} - C_\tau^{x,\pi} \big) Q_\infty^{q, \gamma} \right| \Fscr_\tau \right] \right].
        \end{equation}
\end{proposition}

\begin{proof}
        Observe that $C_\tau^{x,\pi}$ is $\Fscr_\tau$-measurable and $\E[ Q_\infty^{q,\gamma} | \Fscr_\tau ] = Q_\tau^{q,\gamma}$ since  $Q^{q,\gamma}$ is a martingale.
        Utilizing the tower property, we obtain $J( \pi, \gamma; x, q )
                = \E\left[ C_\infty^{x,\pi} Q_\infty^{q,\gamma} \right]
                = \E\left[ C_\infty^{x,\pi} Q_\infty^{q,\gamma} - (C_\infty^{x,\pi} - C_\tau^{x,\pi}) Q_\infty^{q,\gamma} \right]
                = \E\left[ \E\big( C_\tau^{x,\pi} Q_\infty^{q,\gamma} | \Fscr_\tau \big) + \E\big( (C_\infty^{x,\pi} - C_\tau^{x,\pi}) Q_\infty^{q,\gamma} | \Fscr_\tau \big) \right]
                = \E\left[ C_\tau^{x,\pi} Q_\tau^{q,\gamma} + \E\big( (C_\infty^{x,\pi} - C_\tau^{x,\pi}) Q_\infty^{q,\gamma} | \Fscr_\tau \big) \right]$.
\end{proof}

Note that $C_\infty^{x,\pi} - C_\tau^{x,\pi}$ represents the cost realized after time $\tau$.
Proposition \ref{prop:time-decomposition} states that the final outcome can be decomposed into two terms: one term describes the subgame before time $\tau$, and the other term describes the subgame after time $\tau$.

Observe that, in the subgame after time $\tau$, the trader is liquidating $X_\tau^{x,\pi}$ shares and the adversary is selecting a $Q_\tau^{q,\gamma}$-fraction of the future scenarios realized thereafter.
This time decomposition naturally leads to the following dynamic programming principle:

\begin{theorem}[CVaR dynamic programming principle] \label{thm:cvar-dp}
        For any $x \in \mathbb{X}$, $q \in [0,1]$, and a stopping time $\tau$ with $\E[\tau] < \infty$, we have
        \begin{equation} \label{eq:cvar-dp}
                V(x,q) = \inf_{\pi \in \Pi(x)} \sup_{\gamma \in \Gamma(q)} \E\left[ C_\tau^{x,\pi}
                  Q_\tau^{q,\gamma} + V\left( X_\tau^{x,\pi}, Q_\tau^{q,\gamma} \right) \right].
        \end{equation}
\end{theorem}

Theorem \ref{thm:cvar-dp} provides the optimality principle in the form of Bellman's equation: at the equilibrium, the outcome after time $\tau$ can be sufficiently described by the subgame equilibrium $V(X_\tau^{x,\pi}, Q_\tau^{q,\gamma})$.
The trader is minimizing the (risk-adjusted) cost up to time $\tau$ in a consideration of his future state $X_\tau$, while the adversary is simultaneously maximizing the (risk-adjusted) cost up to time $\tau$ in a consideration of his future state $Q_\tau$, and the subgame starts at those future states.
Like Theorem \ref{thm:cvar-minimax}, Theorem \ref{thm:cvar-dp} relies on the saddle-point characterization of the equilibrium; i.e., it does not matter which player commits his policy first in the subgame.
The formal proof can be found in Appendix \ref{ssec:proof-cvar-dp}.

\subsection{$(X,Q)$-Markov Policies}

Theorem \ref{thm:cvar-dp} implies that the \emph{augmented state space} $(X_t, Q_t)$ is sufficient to describe the remaining subgame at time $t$. 
Therefore, if a policy is reasonable, its action at time $t$ ($\pi_t$ or $\gamma_t$) should be determined by the current position size $X_t$ and the current quantile value $Q_t$.
To formalize this idea, we introduce time-stationary Markov policies running on this augmented state space:

\begin{definition}[ $(X,Q)$-Markov policies ]
        We say that a trader's policy $\pi$ is an $(X,Q)$-Markov policy coupled with $\gamma$ if
        \begin{equation} \label{eq:markov1}
                \pi_t(\omega) = f\left( X_t^{x,\pi}(\omega), Q_t^{q,\gamma}(\omega) \right), \quad \forall t, \omega
        \end{equation}
        for some measurable function $f:\mathbb{R} \times [0,1] \rightarrow \mathbb{R}$.

        Similarly, an adversary's policy $\gamma$ is an $(X,Q)$-Markov policy coupled with $\pi$ if
        \begin{equation} \label{eq:markov2}
                \gamma_t(\omega) = g\left( X_t^{x,\pi}(\omega), Q_t^{q,\gamma}(\omega) \right)
        \end{equation}
        for some measurable function $g:\mathbb{R} \times [0,1] \rightarrow \mathbb{R}$.

        A policy pair $(\pi, \gamma)$ is a mutually coupled $(X,Q)$-Markov policy pair if both \eqref{eq:markov1} and \eqref{eq:markov2} hold.
\end{definition}

An $(X,Q)$-Markov policy is characterized by a function defined on the augmented state space.
The function $f$ or $g$ specifies the liquidation rate or the quantile diffusion rate when the current position size is $x$ and the current quantile level is $q$.
Recall that we have defined a policy, $\pi$ or $\gamma$, as a continuous-time stochastic process adapted to the Brownian motion, i.e., as a (progressively measurable) mapping $\pi: \mathbb{T} \times \Omega \rightarrow \mathbb{R}$. 
Strictly speaking, the function $f$ or $g$ does not completely determine one player's policy unless the other player's policy is specified.
To avoid this ambiguity, when we describe an $(X,Q)$-Markov policy, we specify the other player's policy that is coupled with it.

To better understand, consider the policies $\pi$ and $\gamma$ that are mutually coupled $(X,Q)$-Markov policies induced by functions $f$ and $g$.
Under $\pi$ and $\gamma$, the system is completely described by the coupled processes $(X_t, Q_t)_{t \geq 0}$ on the augmented state space, whose dynamics are given by the following stochastic differential equations:
\begin{equation}
        dX_t = - f\left( X_t, Q_t \right) dt, \quad dQ_t = g\left( X_t, Q_t \right) dW_t,
\end{equation}
with the initial states $X_0 = x$ and $Q_0 = q$.
Even if $\gamma$ is not an $(X,Q)$-Markov policy, we can still consider an $(X,Q)$-Markov policy $\pi$ that is induced by $f$ and coupled with $\gamma$, and then the position process will be given by $dX_t = - f\left( X_t, Q_t^{q,\gamma} \right) dt$.

Note also that the admissibility of an $(X,Q)$-Markov policy is not always guaranteed: it may fail to satisfy the admissible conditions given in \eqref{eq:Pi} or \eqref{eq:Gamma}, depending on the generating function $f$ or $g$ as well as the other player's policy coupled with it.
See the discussions before and after Theorem \ref{thm:policy-optimality}.


\section{Optimal Solution} \label{sec:opt}

In this section, we utilize the CVaR dynamic programming principle (Theorem \ref{thm:cvar-dp}) to derive a Hamilton--Jacobi--Bellman (HJB) equation for the risk-sensitive optimal execution problem, and identify the functional form of the value function and the optimal policies by solving this HJB equation.
For all propositions/theorems stated in this section, we defer their proofs to Appendix \ref{app:proof-opt}.

\subsection{Minimal CVaR Cost}

We first state Dynkin's formula that we can apply to the right-hand side of \eqref{eq:cvar-dp} in Theorem \ref{thm:cvar-dp} so as to represent it as a time-integration.

\begin{proposition}[Dynkin's formula] \label{prop:cvar-dp-dynkin}
        Consider a function $\widehat{V} : \mathbb{X} \times [0,1] \rightarrow \mathbb{R}$ such
        that \newedit{$\widehat{V} \in \Cscr^{1,2}(\mathbb{X} \times [0,1])$.}
        For any $\pi \in \Pi(x)$, $\gamma \in \Gamma(q)$, and stopping time $\tau$, we have
        \begin{align}
                \E\left[ C_{\tau}^{x,\pi} Q_{\tau}^{q,\gamma} + \widehat{V}\left( X_{\tau}^{x,\pi}, Q_{\tau}^{q,\gamma} \right) \right]
                        &= \widehat{V}\left( x, q \right)
                        \\&+ \E\left[ \int_{t=0}^{\tau} \left\{ \frac{\eta}{2} Q_t^{q,\gamma} \pi_t^2 - \widehat{V}_x\left( X_t^{x,\pi}, Q_t^{q,\gamma} \right) \pi_t \right\} dt \right]
                        \label{eq:dynkin-trader}
                        \\&+ \E\left[ \int_{t=0}^{\tau} \left\{ \frac{1}{2} \widehat{V}_{qq}\left( X_t^{x,\pi}, Q_t^{q,\gamma} \right) \gamma_t^2 - \sigma X_t^{x,\pi} \gamma_t \right\} dt \right],
                        \label{eq:dynkin-adversary}
        \end{align}
        where $\widehat{V}_x(x,q) \defeq \frac{ \partial }{ \partial x }\widehat{V} (x,q)$ and $\widehat{V}_{qq}(x,q) \defeq \frac{ \partial^2 }{ \partial q^2 }\widehat{V} (x,q)$.
\end{proposition}

For the sake of argument, suppose that the value function $V$ can be plugged into Proposition \ref{prop:cvar-dp-dynkin} in the place of $\widehat{V}$.
When considering an infinitesimal time interval (i.e., $\tau = dt$), we have
\begin{align}
        V(x,q) &\stackrel{\text{Thm \ref{thm:cvar-dp}}}{=} \inf_{\pi \in \Pi(x)} \sup_{\gamma \in \Gamma(q)} \E\left[ C_{\tau}^{x,\pi} Q_{\tau}^{q,\gamma} + V\left( X_{\tau}^{x,\pi}, Q_{\tau}^{q,\gamma} \right) \right]
                \\&\stackrel{\text{Prop \ref{prop:cvar-dp-dynkin}}}{=} \inf_{\pi \in \Pi(x)} \sup_{\gamma \in \Gamma(q)}\left\{
                        V(x,q)
                        + \left( \frac{\eta}{2} q \pi_0^2 - V_x\left( x, q \right) \pi_0 \right) dt
                        + \left( \frac{1}{2} V_{qq}\left( x, q \right) \gamma_0^2 - \sigma x \gamma_0 \right) dt
                        \right\}
                \\&= V(x,q)
                        + \inf_{\pi \in \Pi(x)} \left\{ \frac{\eta}{2} q \pi_0^2 - V_x\left( x, q \right) \pi_0 \right\} dt
                        + \sup_{\gamma \in \Gamma(q)} \left\{ \frac{1}{2} V_{qq}\left( x, q \right) \gamma_0^2 - \sigma x \gamma_0 \right\} dt.
\end{align}
Observe that the terms associated with the trader's policy $\pi$ and the terms associated with the adversary's policy $\gamma$ can be separated.
This informal argument suggests that the value function $V$ has to satisfy the following HJB equation:
\begin{equation}
        \min_{v \in \mathbb{R}} \left\{ \frac{\eta}{2} q v^2 - V_x\left( x, q \right) v \right\}
        +
        \max_{w \in \mathbb{R}} \left\{ \frac{1}{2} V_{qq}\left( x, q \right) w^2 - \sigma x w \right\}
        = 0.
\end{equation}
In the following theorem, we make this argument more formally and identify the sufficient conditions to be the optimal value function.

\begin{theorem}[Verification theorem] \label{thm:verification}
        Consider a function \newedit{$V^\star : \mathbb{X} \times [0,1] \rightarrow \mathbb{R}_+$} satisfying
        \begin{enumerate}[label=(\roman*)]
                \item \label{it:verification-HJB} \newedit{$V^\star \in \Cscr^{1,2}(\mathbb{X} \times
                    (0,1))$}, and, for any \newedit{$x \in \mathbb{X}$}
                  and $q \in (0,1)$, it satisfies
                \begin{equation} \label{eq:HJB}
                        \min_{v \in \mathbb{R}} \left\{ \frac{\eta}{2} q v^2 - V^\star_x\left( x, q \right) v \right\}
                        +
                        \max_{w \in \mathbb{R}} \left\{ \frac{1}{2} V^\star_{qq}\left( x, q \right) w^2 - \sigma x w \right\}
                        = 0,
                \end{equation}
                where $V^\star_x \defeq \frac{\partial V^\star}{\partial x}$ and $V^\star_{qq} \defeq \frac{ \partial^2 V^\star}{\partial q^2}$.

                \item \label{it:verification-boundary}
                                $V^\star(0, q) = 0$ for all $q \in [0,1]$, and $V^\star(x, 0) = V^\star(x,1) = 0$ for all \newedit{$x \in \mathbb{X}$}.

                \item \label{it:verification-symmetry}
                                $V^\star(x, q) = V^\star(-x, q)$ for all \newedit{$x \in \mathbb{X}$} and $q \in [0,1]$, and $V^\star(x,q)$ is increasing in \newedit{$x$ on $[0, M]$}.

                \item \label{it:verification-f-monotone}
                        $\frac{V^\star_x(x,q)}{q}$ is increasing in \newedit{$x$ on $\mathbb{X}$} for each $q \in (0,1)$, and decreasing in $q$ on $(0,1)$ for each \newedit{$x \in \mathbb{X}$}.

                \item  \label{it:verification-g-monotone}
                        $\frac{x}{V^\star_{qq}(x,q)}$ is decreasing in \newedit{$x$ on $[0, M]$}.
        \end{enumerate}
        Then, $V(x,q) = V^\star(x,q)$ for all $x \in \mathbb{X}$ and $q \in [0,1]$.
\end{theorem}

In Theorem \ref{thm:verification}, conditions \ref{it:verification-HJB} and \ref{it:verification-boundary} specify the HJB equation and the boundary conditions that the value function has to satisfy.
In fact, the value function $V$ can be uniquely determined by these two conditions.
However, the other conditions, \ref{it:verification-symmetry}--\ref{it:verification-g-monotone}, are also necessary to show that this value is indeed achievable within our definition of the admissible policies, $\Pi(x)$ and $\Gamma(q)$.
More specifically, condition \ref{it:verification-symmetry} asserts the symmetry and the monotonicity of the value function with respect to position size $x$, and conditions \ref{it:verification-f-monotone} and \ref{it:verification-g-monotone} assert certain behaviors of the optimal policies that are implied from the HJB equation (e.g., the optimal liquidation strategy should trade more aggressively when liquidating a larger quantity).
While these properties of the value function are natural given the problem structure, they serve as regularity conditions in our proof to resolve technical issues arising in the convergence analysis.

Observe that the optimization terms in the HJB equation \eqref{eq:HJB} are separated and each of them is a trivial quadratic optimization problem.
By solving these optimizations explicitly, the HJB equation can be translated into the following partial differential equation:
\begin{equation}
        V_x^2(x,q) \times V_{qq}(x,q) = - \sigma^2 \eta \times x^2 \times q.
\end{equation}
It turns out that this differential equation with the boundary condition \ref{it:verification-boundary} admits a separable solution.
The value function $V(x,q)$ can be represented as a product of a function of $x$ and a function of $q$, as identified in the following theorem.

\begin{theorem}[Value function] \label{thm:value-function}
  Consider a function \newedit{$V^\star: \mathbb{X} \times [0,1] \rightarrow \mathbb{R}$} defined as
        \begin{equation} \label{eq:V-opt}
                V^\star(x,q) \defeq (3/4)^{\frac{2}{3}} \times \sigma^{\frac{2}{3}} \eta^{\frac{1}{3}} \times |x|^{\frac{4}{3}} \times \varphi(q),
        \end{equation}
        where $\varphi : [0,1] \rightarrow \mathbb{R}_+$ is the solution in $\Cscr((0,1))$ to the following differential equation:
        \begin{equation} \label{eq:diff-eq}
                \varphi^2(q) \times \varphi''(q) = -q, ~ \forall q \in (0,1)
                , \quad \text{and} \quad
                \varphi(0) = \varphi(1) = 0.
        \end{equation}
        Then, $V^\star$ satisfies the conditions of Theorem \ref{thm:verification}, and hence $V(x,q) = V^\star(x,q)$ for all \newedit{$x \in \mathbb{X}$} and $q \in [0,1]$.
\end{theorem}

\begin{figure}[H]
\centering
\begin{subfigure}{.5\textwidth}
  \centering
  \includegraphics[width=\linewidth]{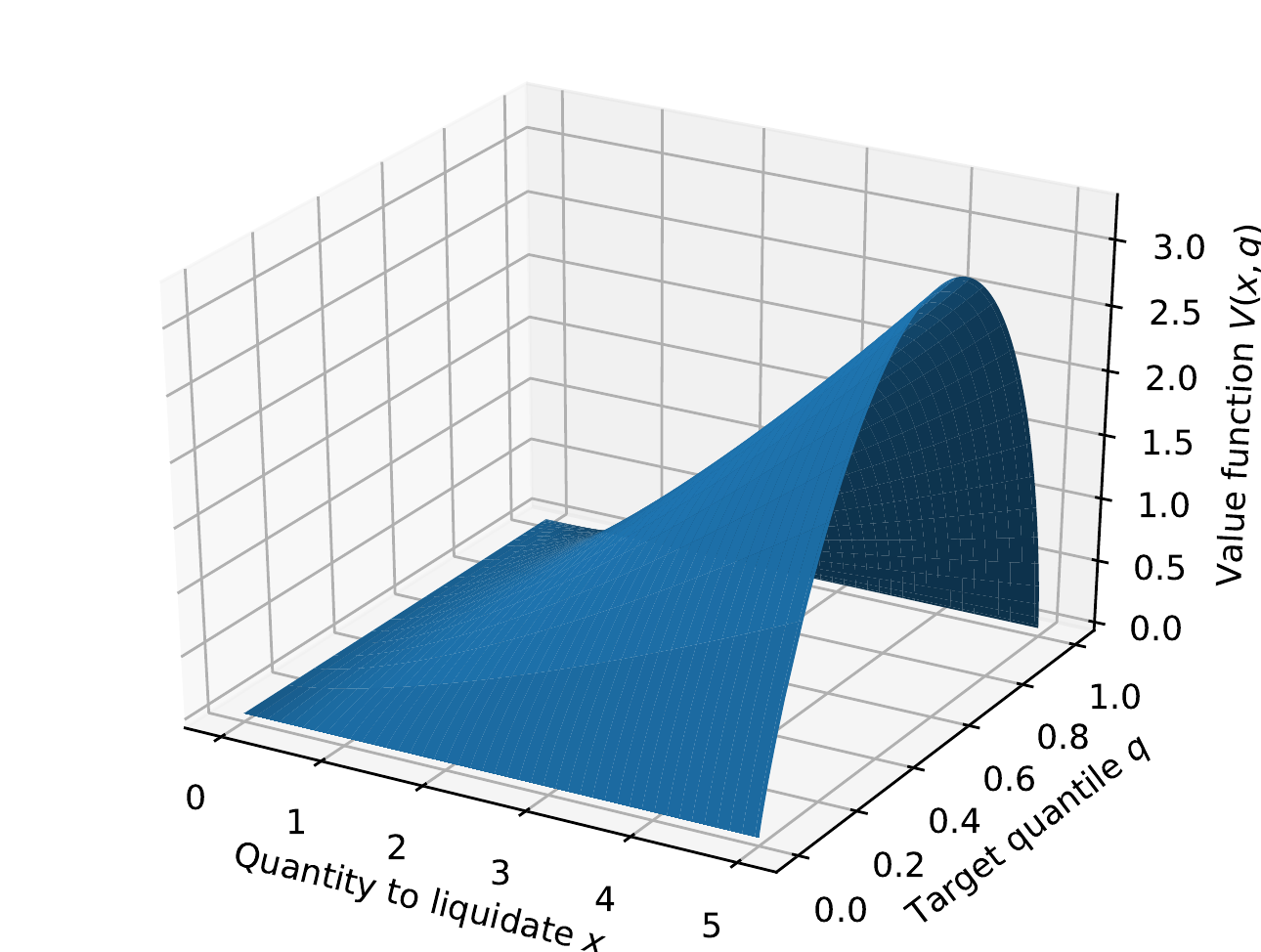}
\end{subfigure}%
\begin{subfigure}{.5\textwidth}
  \centering
  \includegraphics[width=0.9\linewidth]{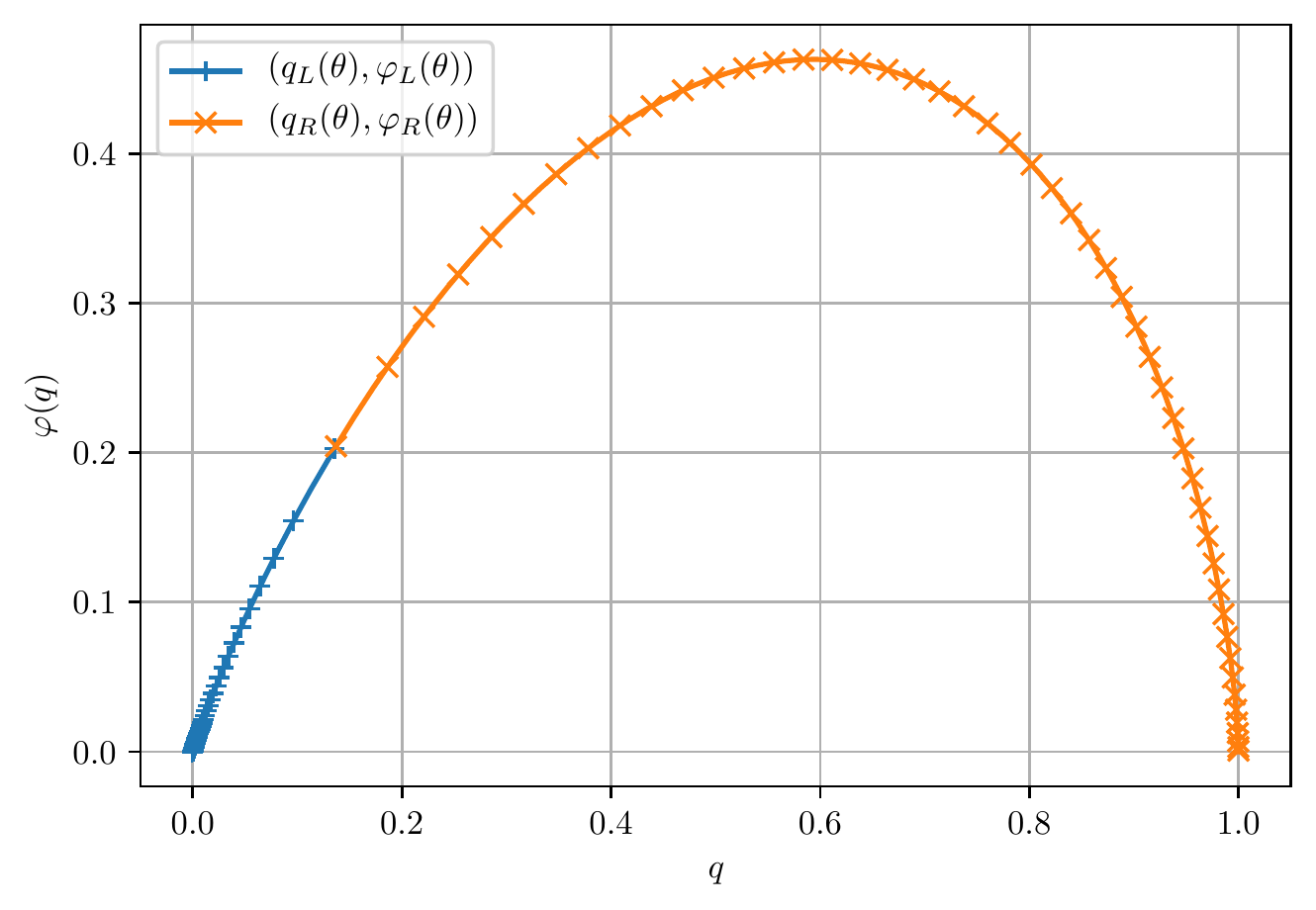}
\end{subfigure}
\caption{(Left) Value function $V(x,q)$ that represents the minimal $\scvar$ loss (i.e., $\scvar_q[C_\infty]$) at a quantile level $q$ when liquidating $x$ units of an asset given that $\sigma = \eta = 1$.
        The closed-form expression is provided in \eqref{eq:V-opt}.
        (Right) Function $\varphi(q)$ that is derived in Proposition \ref{prop:Emden-Fowler}.
        The curve $\{ (q, \varphi(q)) \}_{q \in [0,1]}$ is represented in parametric form, in which the left part of the curve is represented as $\left\{  \left( q_L(\theta), \varphi_L(\theta) \right) \right\}_{\theta \in [0, \infty]}$, and  the right part of the curve is represented as $\left\{ \left( q_R(\theta), \varphi_R(\theta) \right) \right\}_{\theta \in [0, \bar{\theta}]}$.
        }
\label{f-empirical}
\end{figure}

The differential equation of form \eqref{eq:diff-eq} is known as the \emph{Emden-Fowler equation} \citep[2.3.27]{ODEhandbook}, and its solution can be expressed in parametric form as follows:

\begin{proposition}[Parametric representation of $\varphi(q)$] \label{prop:Emden-Fowler}
        The function $\varphi(q)$ can be represented in a parametric form that admits a closed-form expression.
        Define
        \begin{equation}
                Z_L(\theta) \defeq - \frac{2}{\pi} K_{1/3}(\theta)
                , \quad
                Z_R(\theta) \defeq \sqrt{3} J_{1/3}(\theta) - Y_{1/3}(\theta)
        \end{equation}
        where $J$ and $Y$ are the first and second kinds of Bessel functions, and $K$ is the second kind of modified Bessel function.
        Further define
        \begin{equation}
                \bar{\theta} \defeq \inf\{ \theta > 0 :  Z_R(\theta) = 0 \} \approx 2.3834
                , \quad
                a \defeq \frac{1}{ \bar{\theta}^{\frac{4}{3}} \left( Z_R'(\bar{\theta}) \right)^2 } \approx 0.2910
                , \quad
                b \defeq a \left( 9/2 \right)^{\frac{1}{3}} \approx 0.1763.
        \end{equation}
        Then, the curve $\{ (q, \varphi(q)) \}_{q \in [0,1]}$ is parameterized as
        \begin{equation}
                \{ (q, \varphi(q)) \}_{q \in [0,1]}
                        = \left\{  \left( q_L(\theta), \varphi_L(\theta) \right) \right\}_{\theta \in [0, \infty]} \bigcup \left\{ \left( q_R(\theta), \varphi_R(\theta) \right) \right\}_{\theta \in [0, \bar{\theta}]},
        \end{equation}
        where\footnote{The values of $Z_L(\theta)$ and $Z_R(\theta)$ are not defined when $\theta = 0$. However, the limit points (e.g., $\lim_{\theta \searrow 0} (q_L(\theta), \varphi_L(\theta))$, $\lim_{\theta \nearrow \infty} (q_L(\theta), \varphi_L(\theta))$) do exist, and our parametric representation includes those limit points.}
        \begin{equation} \label{eq:Emden-Fowler-left}
                q_L(\theta) \defeq a \theta^{-\frac{2}{3}}\left[ \left( \theta Z_L'(\theta) + \frac{1}{3} Z_L(\theta) \right)^2 - \theta^2 Z_L^2(\theta) \right]
                , \quad
                \varphi_L(\theta) \defeq b \theta^{\frac{2}{3}} Z_L^2(\theta),
        \end{equation}
        and
        \begin{equation} \label{eq:Emden-Fowler-right}
                q_R(\theta) \defeq a \theta^{-\frac{2}{3}}\left[ \left( \theta Z_R'(\theta) + \frac{1}{3} Z_R(\theta) \right)^2 + \theta^2 Z_R^2(\theta) \right]
                , \quad
                \varphi_R(\theta) \defeq b \theta^{\frac{2}{3}} Z_R^2(\theta).
        \end{equation}
\end{proposition}

\subsection{Optimal Adaptive Liquidation Strategy} \label{ssec:opt-policy}

Let $f^\star(x,q)$ and $g^\star(x,q)$ be, respectively, the minimizer and the maximizer of the optimization terms in the HJB equation \eqref{eq:HJB}, i.e.,
\begin{align}
        \label{eq:f-star}
        f^\star(x,q)
                &\defeq \argmin_{v \in \mathbb{R}} \left\{ \frac{\eta}{2} q v^2 - V_x\left( x, q \right) v \right\}
                &=& \frac{V_x(x,q)}{\eta q}
                = (3/4)^{-\frac{1}{3}} \times \sigma^{\frac{2}{3}} \eta^{-\frac{2}{3}} \times x^{\frac{1}{3}} \times \frac{\varphi(q)}{q},
        \\
        \label{eq:g-star}
        g^\star(x,q)
                &\defeq \argmax_{w \in \mathbb{R}} \left\{ \frac{1}{2} V_{qq}\left( x, q \right) w^2 - \sigma x w \right\}
                &=& \frac{\sigma x}{V_{qq}(x,q)}
                = -(3/4)^{-\frac{2}{3}} \times \sigma^{\frac{1}{3}} \eta^{-\frac{1}{3}} \times x^{-\frac{1}{3}} \times \frac{\varphi^2(q)}{q},
\end{align}
where we define $x^{\frac{1}{3}} = -|x|^{\frac{1}{3}}$ for $x < 0$.
The function $f^\star(x,q)$ specifies the trader's optimal liquidation rate when the current position size is $x$ and the current quantile level is $q$, and similarly, the function $g^\star(x,q)$ specifies the adversary's optimal quantile diffusion rate in that situation.
We can naturally postulate mutually coupled $(X,Q)$-Markov policies $\pi^\star$ and $\gamma^\star$ induced by these functions $f^\star$ and $g^\star$, which are characterizing the equilibrium of the stochastic game.\footnote{
        This does not mean that the policy $\pi^\star$ is optimal against any adversary's policy $\gamma$, nor an $(X,Q)$-Markov policy induced by $f^\star$ and coupled with $\gamma$ is the best response against $\gamma$.
        It merely means that $\pi^\star$ is the best response against $\gamma^\star$ only, and vice versa.
        In order to obtain a best response against an arbitrary adversary's policy $\gamma \in \Gamma(q)$, we may need to characterize the best possible performance against $\gamma$, e.g., $V(x,q; \gamma) \defeq \inf_{\pi \in \Pi(x)} J( \pi, \gamma; x, q )$, and derive and solve the HJB equation associated with it.
        Nevertheless, the policy $\pi^\star$ is the optimal liquidation strategy that minimizes the CVaR value of implementation shortfall.
}
Under policies $\pi^\star$ and $\gamma^\star$, the system is described by the following stochastic differential equations:
\begin{equation} \label{eq:opt-sde}
        dX_t = - f^\star\left( X_t, Q_t \right) dt, \quad dQ_t = g^\star\left( X_t, Q_t \right) dW_t,
\end{equation}
with $X_0 = x$ and $Q_0 = q$.

\begin{figure}[H]
\centering
\begin{subfigure}{.5\textwidth}
  \centering
  \includegraphics[width=\linewidth]{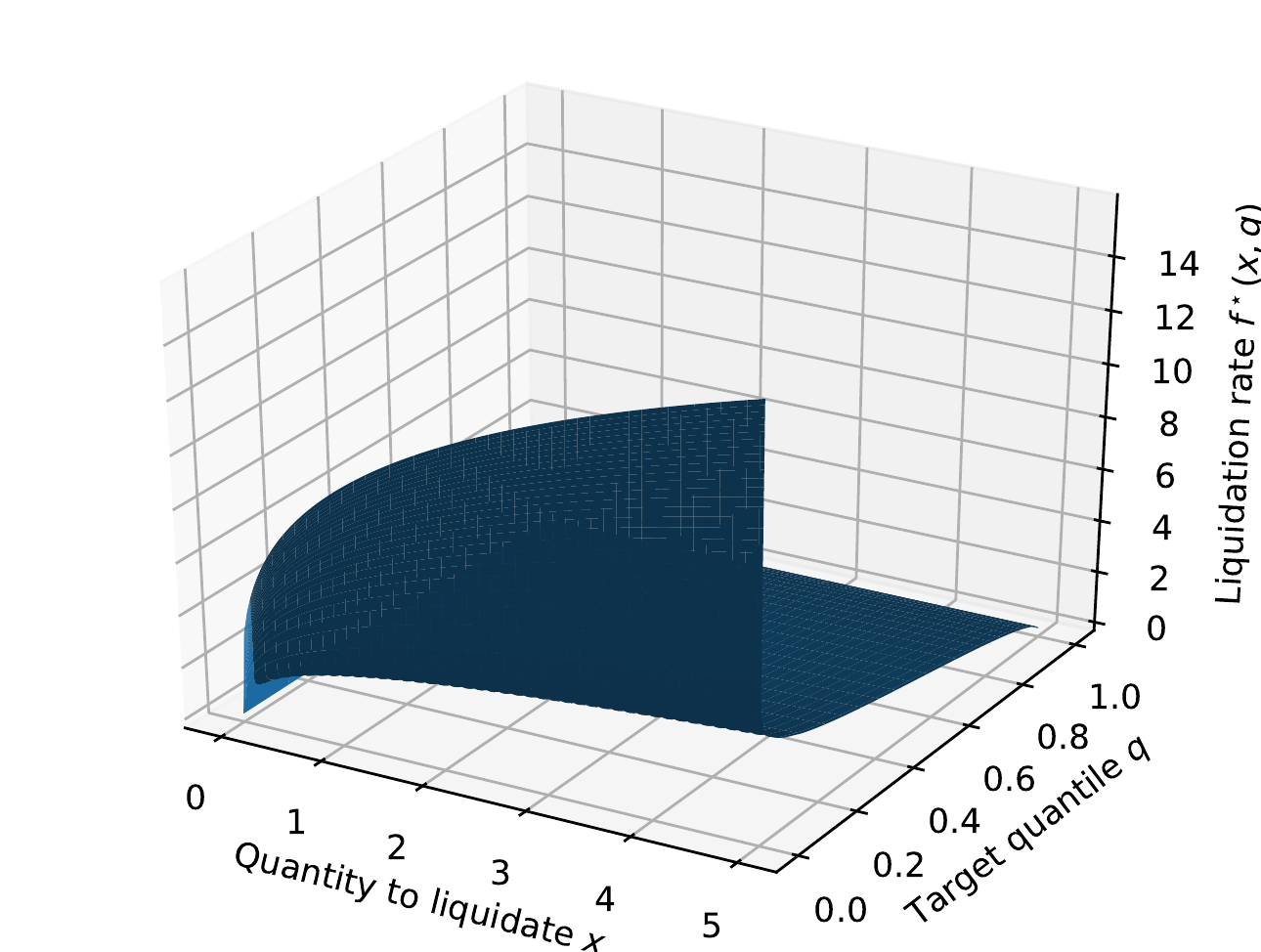}
\end{subfigure}%
\begin{subfigure}{.5\textwidth}
  \centering
  \includegraphics[width=\linewidth]{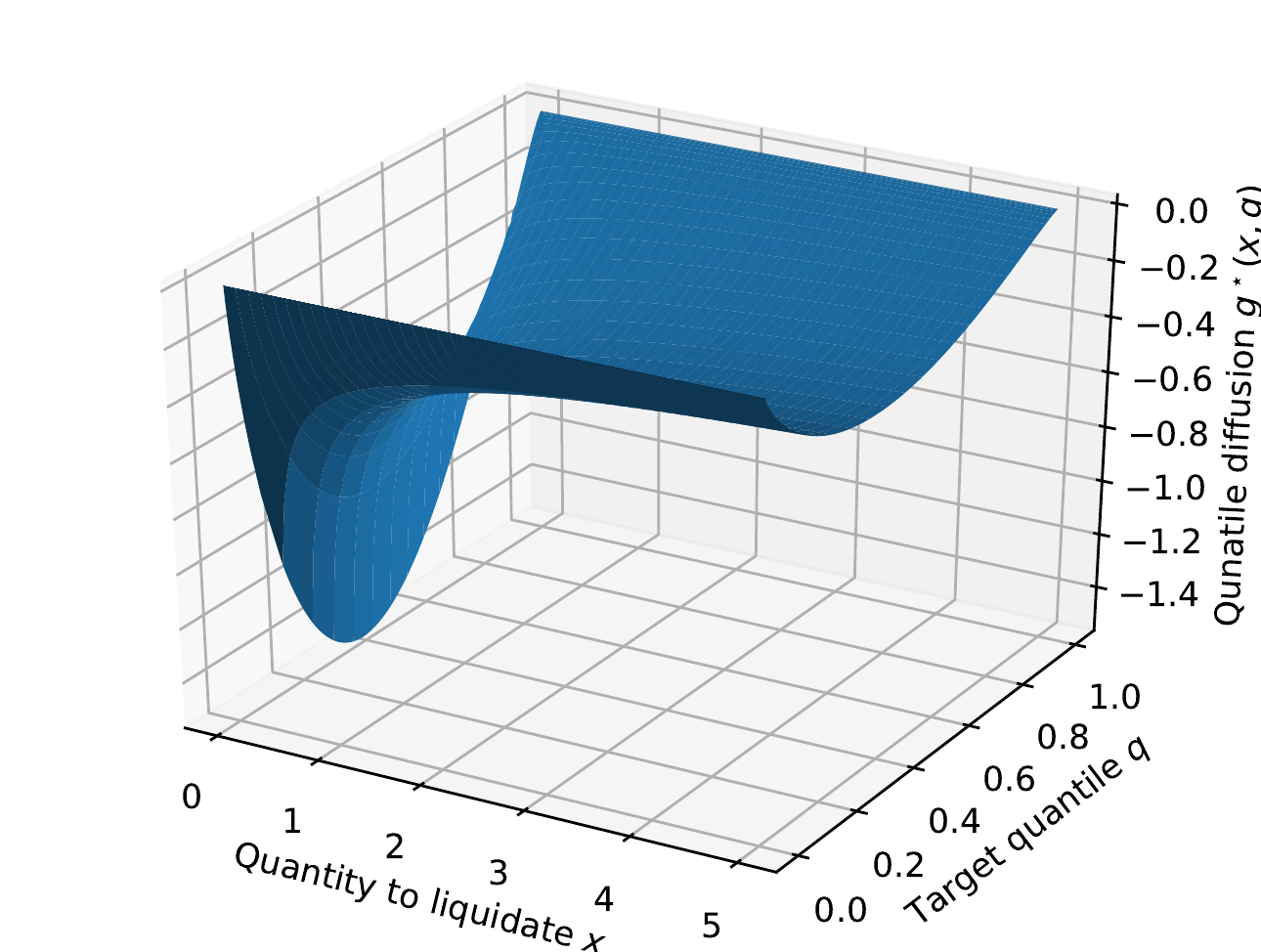}
\end{subfigure}
\caption{ Optimal trading rate function $f^\star(x,q)$ (left) and optimal quantile diffusion rate function $g^\star(x,q)$ (right) given that $\sigma = \eta = 1$.
        The closed-form expressions are provided in \eqref{eq:f-star} and \eqref{eq:g-star}.}
\label{fig:opt-policy}
\end{figure}

However, we cannot directly show that the policies $\pi^\star$ and $\gamma^\star$ satisfy the admissibility conditions introduced in \eqref{eq:Pi} and \eqref{eq:Gamma}, because the functions $f^\star$ and $g^\star$ exhibit extreme behaviors near the boundaries, as shown in Figure \ref{fig:opt-policy}.
For example, if $Q_t \searrow 0$, the liquidation rate $\pi_t$ may increase unboundedly since $\lim_{q \searrow 0} f^\star(x,q) = \infty$, and thus may violate the condition $\E\left[ \left( \int_{t=0}^\infty \pi_t^2 dt \right)^2 \right] < \infty$.
If $Q_t \nearrow 1$, on the other hand, the liquidation rate $\pi_t$ may vanish since $\lim_{q \nearrow 1} f^\star(x,q) = 0$, and thus may violate the condition $\E\left[ \int_{t=0}^\infty |X_t^{x,\pi}|^2 dt \right] < \infty$.

We instead prove that the optimal value function can be achieved ``asymptotically'' by a sequence of admissible policies that approximate $\pi^\star$ and $\gamma^\star$.

\begin{theorem}[Policy optimality] \label{thm:policy-optimality}
        There exists a sequence of function pairs $( f^{(n)}, g^{(n)} )_{n \in \mathbb{N}}$ \newedit{with $f^{(n)}:\mathbb{X} \times [0,1] \rightarrow \mathbb{R}$ and $g^{(n)}:\mathbb{X} \times [0,1] \rightarrow \mathbb{R}$} such that
        \begin{equation}
                \lim_{n \rightarrow \infty} f^{(n)}(x,q) = f^\star(x,q)
                , \quad
                \lim_{n \rightarrow \infty} g^{(n)}(x,q) = g^\star(x,q)
                , \quad
                \forall (x,q) \in \mathbb{R} \setminus \{0\} \times (0, 1],
        \end{equation}
        and it satisfies the following properties for any $x \in \mathbb{X}$ and $q \in [0,1]$:
        \begin{enumerate}[label=(\alph*)]
                \item \label{it:opt-pi-optimality-general}
                        For any given $\gamma \in \Gamma(q)$, let $\pi^{(n),\gamma}$ be an $(X,Q)$-Markov trader's policy induced by $f^{(n)}$ and coupled with $\gamma$.
                        Then, $\pi^{(n),\gamma}$ is admissible and
                        \begin{equation}
                                \limsup_{n \rightarrow \infty} J( \pi^{(n),\gamma}, \gamma; x, q ) \leq V(x,q), \quad \forall \gamma \in \Gamma(q).
                        \end{equation}
                \item \label{it:opt-gamma-optimality-general}
                        For any given $\pi \in \Pi(x)$, let $\gamma^{(n),\pi}$ be an $(X,Q)$-Markov adversary's policy induced by $g^{(n)}$ and coupled with $\pi$.
                        Then, $\gamma^{(n),\pi}$ is admissible and
                        \begin{equation}
                                \liminf_{n \rightarrow \infty} J( \pi, \gamma^{(n), \pi}; x, q) \geq V(x,q), \quad \forall \pi \in \Pi(x).
                        \end{equation}
                \item \label{it:opt-pi-gamma-optimality-general}
                        Let $(\pi^{(n)}, \gamma^{(n)})$ be a mutually coupled $(X,Q)$-Markov policy pair induced by $(f^{(n)}, g^{(n)})$.
                        Then, $\pi^{(n)}$ and $\gamma^{(n)}$ are admissible and
                        \begin{equation}
                                \lim_{n \rightarrow \infty} J( \pi^{(n)}, \gamma^{(n)}; x, q) = V(x,q).
                        \end{equation}
        \end{enumerate}

\end{theorem}

Theorem \ref{thm:policy-optimality} shows that we can construct a sequence of functions $( f^{(n)}, g^{(n)} )_{n \in \mathbb{N}}$ that converges to $(f^\star, g^\star)$ pointwise except at the boundaries, and further induces $(X,Q)$-Markov policies that are admissible and asymptotically optimal.
More precisely, against any adversary's policy $\gamma$, the sequence of functions $( f^{(n)} )_{n \in \mathbb{N}}$ induces a sequence of admissible policies $( \pi^{(n),\gamma} )_{n \in \mathbb{N}}$ for the trader, and in the limit, the trader achieves an outcome that is no worse than the equilibrium outcome.
And vice versa, against any trader's policy $\pi$, the sequence of functions $( g^{(n)} )_{n \in \mathbb{N}}$ induces a sequence of admissible policies  $( \gamma^{(n),\pi} )_{n \in \mathbb{N}}$ for the adversary, and in the limit, the adversary achieves an outcome that is no worse than the equilibrium outcome.
As a combination, the sequence of function pairs $( f^{(n)}, g^{(n)} )_{n \in \mathbb{N}}$ induces a sequence of mutually admissible policy pairs $( \pi^{(n)}, \gamma^{(n)} )_{n \in \mathbb{N}}$ that yields the equilibrium outcome asymptotically.

The construction of such a sequence of function pairs $( f^{(n)}, g^{(n)} )_{n \in \mathbb{N}}$ is straightforward.
We consider a vanishing subset of the augmented state space that contains the boundaries, i.e., $\{ (x,q) | |x| \leq \frac{1}{n}, q \leq \frac{1}{n} \text{ or } q \geq 1- \frac{1}{n} \}$, and obtain $f^{(n)}$ and $g^{(n)}$ by suppressing the extreme behaviors of $f^\star$ and $g^\star$ arising in this subset.
Roughly speaking, the liquidation strategy induced by $( f^{(n)}, g^{(n)} )$ mimics the optimal strategy until it clears almost all positions (i.e., $X_t \leq \frac{1}{n}$) or it becomes almost risk-neutral (i.e., $Q_t \geq 1-\frac{1}{n}$) or extremely risk-averse (i.e., $Q_t \leq \frac{1}{n}$), and then liquidates according to a deterministic schedule thereafter.
We can show that the gap between the outcome of this approximated strategy and the theoretical equilibrium outcome is diminishing as $n$ goes to infinity.
We refer the readers to Appendix \ref{app:proof-opt} for the details.

\seclabel{Aggressiveness-in-the-money.}
Despite that the admissibility of the optimal liquidation policy $\pi^\star$ is not guaranteed, we can still characterize its behaviors by inspecting the stochastic differential equations \eqref{eq:opt-sde}.
Without loss of generality, let us consider the task of liquidation (i.e., $x \geq 0$).

First, we observe that the optimal policy liquidates only (i.e., $f^\star(x,q) \geq 0$), until it completes the execution\footnote{
        We are not sure if the completion time is almost surely finite even though the position will be vanishing eventually (i.e., $X_t \searrow 0$).
        In particular, when $Q_t \approx 1$, the optimal policy trades very slowly ($\pi_t \approx 0$) and the position process may never hit zero.
        We believe that the completion time is finite with a probability of at least $1-q$.
        } (i.e., $f^\star(0,q)=0$).
Note that we have not imposed any constraint on the trading direction.
This formally shows that winding back the position during the liquidation process will never be helpful in reducing the CVaR loss.

Second, when the trader becomes more risk-averse (i.e., $q \searrow 0$), the optimal policy trades more aggressively (i.e., $f^\star(x,q) \nearrow \infty$).
The opposite also holds true.
This is because, by liquidating the position more quickly, it can reduce the risk exposure to the changes in price more effectively.
Even though it will be more costly in terms of market impact, it can make sure that the transactions will be made at a certain level, which is more beneficial to a risk-averse trader than a risk-neutral trader.
We can also understand this behavior based on the alternative interpretation of the problem discussed in Remark \ref{rem:Girsanov}: when the risk-averse liquidation problem is cast as a risk-neutral execution problem that involves an adverse price drift, being more risk-averse is equivalent to facing a more adverse price drift, which encourages the risk-neutral trader to trade more aggressively.

Most interestingly, we observe that when the price moves in a favorable direction toward in-the-money (i.e., $dW_t > 0$), the policy becomes more risk-averse (i.e., $dQ_t < 0$ since $g^\star(x,q) < 0$) and hence it trades more aggressively.
This formally characterizes aggressiveness-in-the-money, which has been observed by \citet{Almgren07, Almgren11, Gatheral11, Forsyth12}.
Intuitively, if the trader has made some ``free'' money due to the price change, he would have an additional incentive to complete the liquidation early so as to secure his current profit, and thus he would be willing to pay an additional deterministic cost for aggressive execution.

This behavior can also be understood in the context of a more general risk-sensitive optimal control problem.
Recall that the optimal adversary's martingale $Q_t^{q, \gamma^\star}$ represents the likelihood that the current sample path leads to one of the worst $q$-fraction of outcomes.
When something favorable happens, it becomes less likely that the sample path is in the worst $q^\text{th}$ quantile, and therefore $Q_t$ decreases.
This means that the trader will need to pay attention to a smaller fraction of adverse scenarios; i.e., he will become more risk-averse.

\seclabel{Threshold behavior.}
While we do not have a formal characterization here, we observe that the optimal strategy exhibits some threshold behavior, particularly near the end of the liquidation.
We observe that the policy trades aggressively when the cumulative cost $C_t$ is below some threshold, and it trades passively when the cumulative cost $C_t$ is above the threshold.
Near the end of the liquidation, such a threshold corresponds to $\var_q[C_\infty]$ (the value-at-risk, i.e., the $q^\text{th}$ quantile of the loss distribution), and the liquidation rate sharply changes around $\var_q[C_\infty]$.
This behavior is related to an alternative representation of CVaR \citep{Rockafellar02}: $\cvar_q[ C_\infty ] = \max_{c \in \mathbb{R}}\{  c + \frac{1}{q} \E\left[ (C_\infty - c)^+ \right] \}$, where the maximizer $c^\star$ is in fact $\var_q[C_\infty]$, and $\cvar_q[ C_\infty ]$ only concerns the cases where $C_\infty > c^\star$.

To better illustrate, suppose that the trader is currently left with a small amount of position to liquidate.
If $C_t < c^\star$, the trader is willing to pay a large transaction cost (up to $c^\star - C_t$) to complete the execution as soon as possible, thereby making sure that the total loss $C_\infty$ won't exceed the threshold $c^\star$.
If $C_t > c^\star$, the trader may believe that the total loss $C_\infty$ will inevitably exceed the threshold $c^\star$, and then tries to minimize the expected future cost by slowing down the liquidation.
One can make a connection with aggressiveness-in-the-money, since having $C_t < c^\star$ implies that the price has moved in a favorable direction.

We believe that the threshold value is a function of remaining position size such that it increases as the position size decreases and converges to $\var_q[C_\infty]$ as the position vanishes.
Moreover, the change in the aggressiveness around the threshold also depends on the remaining position size.
To formalize this behavior, one may adopt an alternative formulation of the problem with an extra state variable representing the cost realized so far (i.e., a Markov policy defined on the augmented state space $(X_t, C_t)$), which is in fact the approach suggested by \cite{Bauerle11} for a general class of control problems with a CVaR objective.
This might be a topic of future research.


\section{Cost Analysis: Adaptive vs. Deterministic Strategy} \label{sec:cost-analysis}

In this section, we provide a comparison between the optimal adaptive strategy derived in \S \ref{sec:opt} and the (optimized) deterministic schedules under which the liquidation is executed according to a deterministic schedule committed at the beginning of the liquidation process.

\subsection{Optimized Deterministic Schedules} \label{sec:deterministic}

First observe that any deterministic schedule will yield a normally distributed implementation
shortfall; i.e., given a deterministic schedule $\pi$, the total implementation shortfall
$C_\infty^{x,\pi} = \int_0^\infty \frac{\eta}{2} \pi_t^2 dt - \int_0^\infty \sigma X_t dW_t$
follows a normal distribution with mean and variance given by
\begin{equation}\label{eq:detmeanvar}
  \E[C_\infty^{x,\pi}] = \int_0^\infty \frac{\eta}{2} \pi_t^2 dt,
  \qquad \text{Var}[C_\infty^{x,\pi}] = \int_0^\infty \sigma^2 X_t^2 dt.
\end{equation}
To understand the performance of deterministic schedules, therefore, it suffices consider the CVaR
value of a normal distribution.

Given a normal random variable $C\sim\Nscr(\mu,\sigma^2)$, it is easy to verify that
\begin{equation}\label{eq:normcvar}
  \cvar_q[C] = \mu + \sigma \frac{\kappa(q)}{q},
\end{equation}
where
\begin{equation} \label{eq:kappa}
  \kappa(q) \defeq \phi( \Phi^{-1}(1-q) ),
\end{equation}
and $\phi(\cdot)$ and $\Phi(\cdot)$ are the p.d.f.\ and the c.d.f.\ of a standard normal
distribution, respectively.
Thus,
\begin{equation}\label{eq:detcvar}
  \cvar_q[ C_\infty^{x,\pi} ] = \int_0^\infty \frac{\eta}{2} \pi_t^2 dt +\frac{\kappa(q)}{q}
  \sqrt{ \int_0^\infty \sigma^2 X_t^2 dt },
\end{equation}
for any deterministic schedule $\pi$.

We first focus on the set of all deterministic schedules and find the optimal one that minimizes the CVaR cost.
The next proposition shows that such an optimal schedule has the form of an exponential schedule under which the trader's position decays exponentially over time (i.e., the liquidation rate is proportional to the current position size).

\begin{proposition}[Optimized deterministic schedule] \label{prop:exp-schedule}
        Given an initial position $x \in \mathbb{R}$ and a target quantile $q \in (0,1)$, the optimal deterministic schedule is given by an exponential schedule $X_t = X_0 e^{-t/\tau^\star}$ where the optimal time constant $\tau^\star$ is given by
        \begin{equation}
                \tau^\star = \left( \frac{\eta x q}{ \sqrt{2} \sigma \kappa(q) } \right)^{\frac{2}{3}}.
        \end{equation}
        Let \textsc{exp} be this optimized exponential schedule.
        Then, its performance $V^{\textsc{exp}}(x,q)$ is given by
        \begin{equation} \label{eq:V-exp}
                V^{\textsc{exp}}(x,q)
                        \defeq \scvar_q\left[ C_\infty^{x,\textsc{exp}} \right]
                        = \frac{3}{2^{\frac{5}{3}}} \times \sigma^{\frac{2}{3}} \eta^{\frac{1}{3}} \times |x|^{\frac{4}{3}} \times q^{\frac{1}{3}} \big( \kappa(q) \big)^{\frac{2}{3}}.
        \end{equation}
\end{proposition}

While the proof is provided in Appendix \ref{app:deterministic}, we remark that the optimality of the exponential schedule can be directly inferred from the result of \citet{Almgren00}: it was shown that, given a finite time-horizon of length $T$, a mean-variance optimization results in a trajectory $X_t = \frac{ \sinh( \kappa (T-t) ) }{ \sinh(\kappa T ) } X_0$ for some constant $\kappa > 0$.
As $T \nearrow \infty$, we can observe that the optimized trajectory converges to an exponential schedule $X_t = e^{ -\kappa t } X_0$.

We next examine the volume-weighted average price (VWAP) schedules, under which the trader
liquidates the asset at a constant rate until completion so that the trader's position decreases
linearly over time.\footnote{
 \newedit{A VWAP schedule typically refers to a liquidation schedule that is proportional to the average market volume profile, aiming to make its average transaction price close to the market volume-weighted average price.
 As we assume a time-stationary market in this paper, the VWAP schedule is equivalent to a constant liquidation rate schedule (also known as a TWAP schedule).
 }
 }  The next proposition identifies the optimized VWAP schedule (see
Appendix \ref{app:deterministic} for the proof).

\begin{proposition}[Optimized VWAP schedule] \label{prop:vwap-schedule}
        Given an initial position $x \in \mathbb{R}$ and a target quantile $q \in (0,1)$, the best VWAP schedule is given by $X_t = X_0 \left( 1 - \frac{t}{T^\star} \right)^+$, where the optimal execution period $T^\star$ is as follows:
        \begin{equation}
                T^\star = \left( \frac{\sqrt{3} \eta x q }{ \sigma \kappa(q) } \right)^{\frac{2}{3}}.
        \end{equation}
        Let \textsc{vwap} be this optimized VWAP schedule.
        Then, its performance $V^{\textsc{vwap}}(x,q)$ is given by
        \begin{equation} \label{eq:V-vwap}
                V^{\textsc{vwap}}(x,q)
                        \defeq \scvar_q\left[ C_\infty^{x,\textsc{vwap}} \right]
                        = \frac{3^{\frac{2}{3}}}{2} \times \sigma^{\frac{2}{3}} \eta^{\frac{1}{3}} \times |x|^{\frac{4}{3}} \times q^{\frac{1}{3}} \big( \kappa(q) \big)^{\frac{2}{3}}.
        \end{equation}
\end{proposition}

\subsection{Cost Analysis}

We now compare three liquidation strategies: the optimal adaptive strategy (\textsc{opt}) derived in \S \ref{sec:opt}, the optimized deterministic strategy (\textsc{exp}), and the optimized VWAP strategy (\textsc{vwap}).
We have derived closed-form expressions \eqref{eq:V-opt}, \eqref{eq:V-exp}, and \eqref{eq:V-vwap} that represent their S-CVaR performance $V$, $V^\textsc{exp}$, and $V^\textsc{vwap}$, respectively.
Given that $V(x,q) \leq V^\textsc{exp}(x,q) \leq V^\textsc{vwap}(x,q)$ by their definitions, we
particularly consider the following ratios that are useful for pairwise comparison:
\begin{equation}\label{eq:ratios}
\begin{split}
        \Upsilon_{\textsc{opt}}^{\textsc{exp}} &\defeq \frac{ V^\textsc{exp}(x,q)}{ V(x,q) } - 1 =
        \frac{ (3/2)^{\frac{1}{3}}  q^{\frac{1}{3}} \big( \kappa(q) \big)^{\frac{2}{3}} }{
          \varphi(q) } - 1, \\
        \Upsilon_{\textsc{opt}}^{\textsc{vwap}} &\defeq \frac{ V^\textsc{vwap}(x,q)}{ V(x,q) } - 1
        = \frac{ 2^{\frac{1}{3}}  q^{\frac{1}{3}} \big( \kappa(q) \big)^{\frac{2}{3}} }{
          \varphi(q) } - 1, \\
        \Upsilon_{\textsc{exp}}^{\textsc{vwap}} &\defeq \frac{ V^\textsc{vwap}(x,q) }{
          V^\textsc{exp}(x,q) } - 1 = (4/3)^\frac{1}{3} -1 \approx 10\%.
\end{split}
\end{equation}
Note that these ratios do not change even if we compare $\cvar$ performance instead of $\scvar$ performance since $\scvar$ is merely a scaled version of $\cvar$ (see Remark \ref{prop:scvar-properties}\ref{it:scvar-cvar-relationship}).

\begin{figure}[!h]
        \centering
        \includegraphics[width=0.5\linewidth]{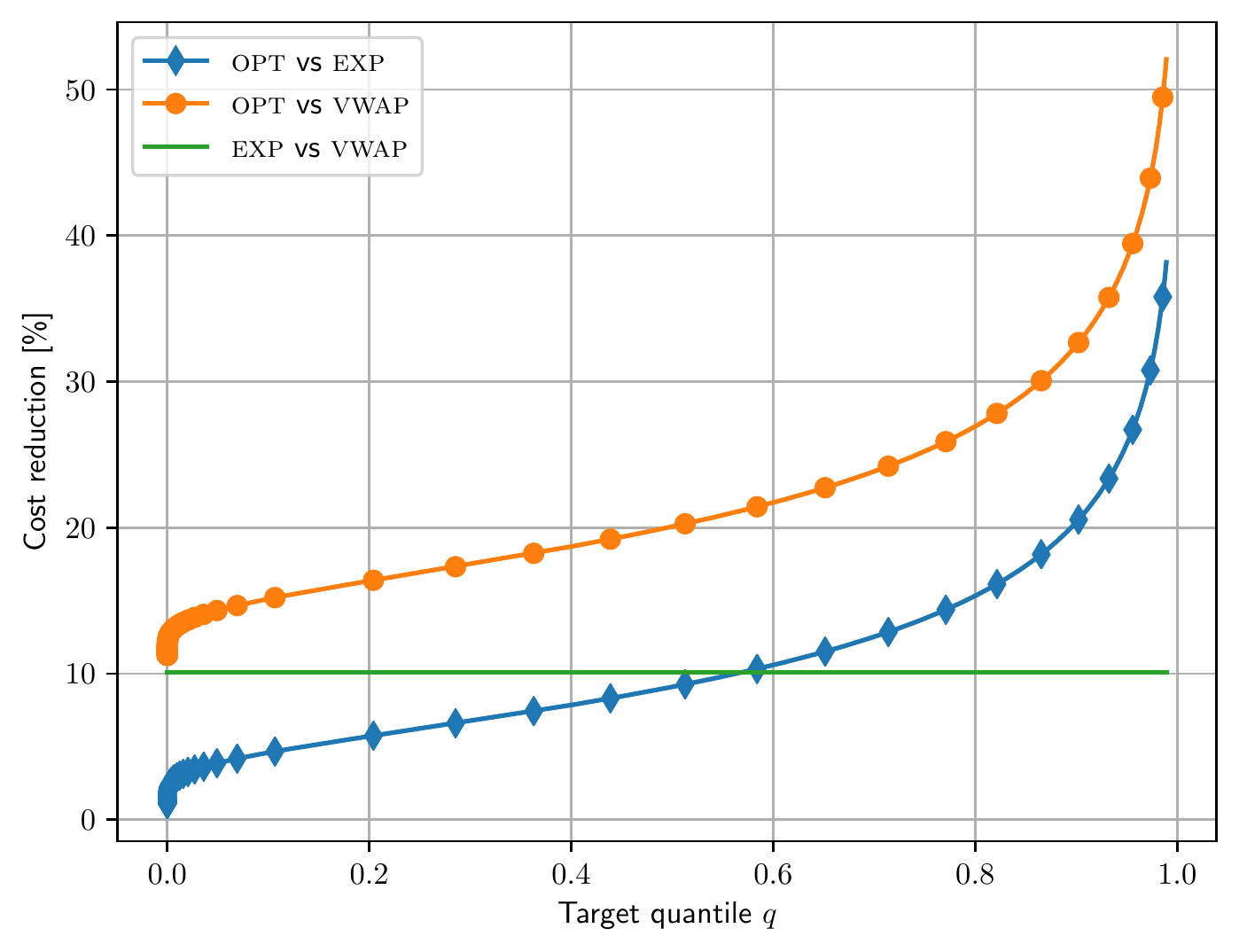}
        \caption{
        The pairwise comparison among three liquidation strategies -- the optimal adaptive strategy, the optimal deterministic strategy (\textsc{exp}), and the optimized VWAP strategy -- measured with the ratios $\Upsilon_{\textsc{opt}}^{\textsc{exp}}$, $\Upsilon_{\textsc{opt}}^{\textsc{vwap}}$, and $\Upsilon_{\textsc{exp}}^{\textsc{vwap}}$ defined in \eqref{eq:ratios}.
        These ratios represent the relative percentage reductions in CVaR cost, and depend on the target quantile $q$ only.
        Each curve shows the value of each ratio when $q$ ranges from zero to one.}
        \label{fig:cost-saving}
\end{figure}

These ratios can be expressed in closed form, and we observe the following.
First, all the ratios depend only on the quantile $q$, but not on the other problem parameters such as the quantity to liquidate $x$, the price volatility $\sigma$, and the market impact factor $\eta$.
Figure \ref{fig:cost-saving} plots these ratios as functions of $q$.
Second, the optimal deterministic schedule always outperforms to the best VWAP schedule by $\approx10.0\%$, irrespective of the value of $q$.
Finally, in a moderate range of $q$, from $0.2$ to $0.8$, the adaptive strategy outperforms the optimal deterministic strategy by $5\%$ to $15\%$, and outperforms the optimized VWAP strategy by $15\%$ to $25\%$.
This gap increases as $q$ increases (i.e., becomes more risk-neutral).
More specifically, it explodes as $q$ approaches one, and vanishes as $q$ approaches zero.\footnote{
        The absolute performance (i.e., the CVaR cost) of all three policies converges to zero as $q \nearrow 1$, and diverges to infinity as $q \searrow 0$.
      }
Note that most traders using optimal
execution algorithms are large investors trading a small portion of their overall portfolio over a
short time horizon. Thus, from the perspective of optimal decision making,  their utility functions are nearly linear, hence the nearly risk neutral regime ($q \approx 1$), where the relative
benefits of dynamic trading are greatest, is also the most practically relevant regime.
See also \S \ref{ssec:numerical-comparison} for a more illustrative comparison between the adaptive strategy and the deterministic strategies.


\section{Numerical Simulations} \label{sec:numerical}

Throughout this section, we denote the optimal adaptive strategy by \textsc{opt}, the optimized deterministic schedule by \textsc{exp}, and the optimized VWAP schedule by \textsc{vwap}.

\subsection{Illustration of Optimal Adaptive Strategy}

We consider a situation where the trader wants to liquidate $x = 1$ unit of an asset given the volatility $\sigma=5.06$ basis points per minute ($1\%$ per day), the market impact factor $\eta = 1.56 \times 10^3$ (liquidating one unit over a day at a constant rate incurs $2$ basis points loss in average), and the target quantile $q$ varying from $0.10$ to $0.99$.

We simulate the optimal policy \textsc{opt} as follows.
We first discretize the time horizon into subintervals of equal length $\Delta t = 10^{-4}$, and generate a sample path of the standard Brownian motion $W_0, W_{\Delta t}, W_{2 \Delta t}, \ldots$.
Starting from $X_0 = x$ and $Q_0 = q$, at each time $t=0, \Delta t, 2 \Delta t, \ldots$, we compute the liquidation rate $\pi_t = f^\star(X_t, Q_t)$ and the quantile diffusion rate $\gamma_t = g^\star(X_t,Q_t)$ and then update the position size $X_{t+\Delta t} = X_t - \pi_t \Delta t$ and the quantile $Q_{t+\Delta t} = Q_t + \gamma_t \Delta W_t$ accordingly, where $\Delta W_t \defeq W_{t+\Delta t} - W_t$.
The expressions for $f^\star$ and $g^\star$ are given in \eqref{eq:f-star} and \eqref{eq:g-star}, and the value of $\varphi(q)$ can be computed using linear interpolation based on its parametric representation derived in Proposition \ref{prop:Emden-Fowler}.
In order to prevent numerical instability, we keep the value of quantile process $Q_t$ between $\epsilon$ and $1-\epsilon$ via truncation (we take $\epsilon = 10^{-5}$); i.e., if $Q_t < \epsilon$ or $Q_t > 1-\epsilon$, it is set to $\epsilon$ or $1-\epsilon$, respectively.
This procedure is repeated until the remaining position size $X_t$ becomes smaller than $10^{-2}$.

Figure \ref{fig:numeric-sample-paths} illustrates the sample paths of the price process $\sigma W_t$, the position process $X_t$, and the quantile process $Q_t$ under \textsc{opt} for different values of target quantile $q \in \{0.1, 0.2, \ldots, 0.99\}$ in the following two scenarios: when the price moves in an adverse direction (left), and when the price moves in a favorable direction (right).
From these results, we confirm the behaviors of the optimal strategy characterized in \S \ref{ssec:opt-policy}.
In every case, the position monotonically decreases over time; i.e., the optimal policy keeps trading in one direction.
Also observe that the policy liquidates the position more aggressively as we take a smaller value for $q$ (i.e., as the policy becomes more risk-averse).
In a comparison between two scenarios (left vs. right), we observe ``aggressiveness-in-the-money''; i.e., the policy trades more aggressively when the price moves in a favorable direction (right).
This behavior can also be observed within each sample path: during the execution process, the quantile process $Q_t$ decreases when the price moves upward and the policy trades more aggressively.
In addition, the quantile process $Q_t$ converges to either zero or one, indicating whether the realized price process is among the worst $q$-fraction of the scenarios.
While not reported here, we observe that $Q_t$ converges to one in the $q$-fraction of simulations and converges to zero in the other $(1-q)$-fraction of simulations (recall that $Q_t$ is a martingale starting at $q$).

\begin{figure}[H]
\centering
\begin{subfigure}{.51\textwidth}
  \centering
  \includegraphics[width=\linewidth]{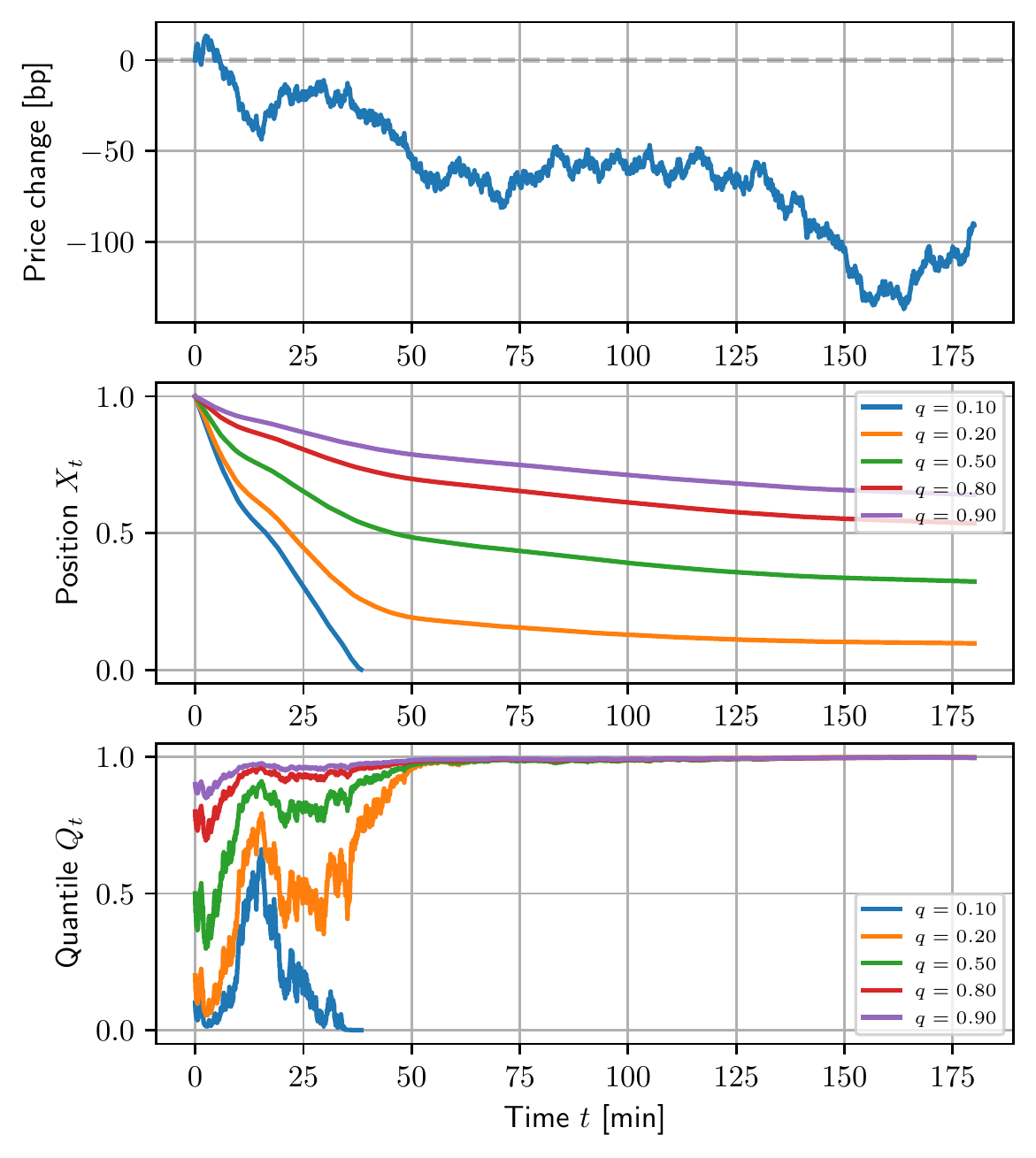}  
\end{subfigure}%
\begin{subfigure}{.49\textwidth}
  \centering
  \includegraphics[width=\linewidth]{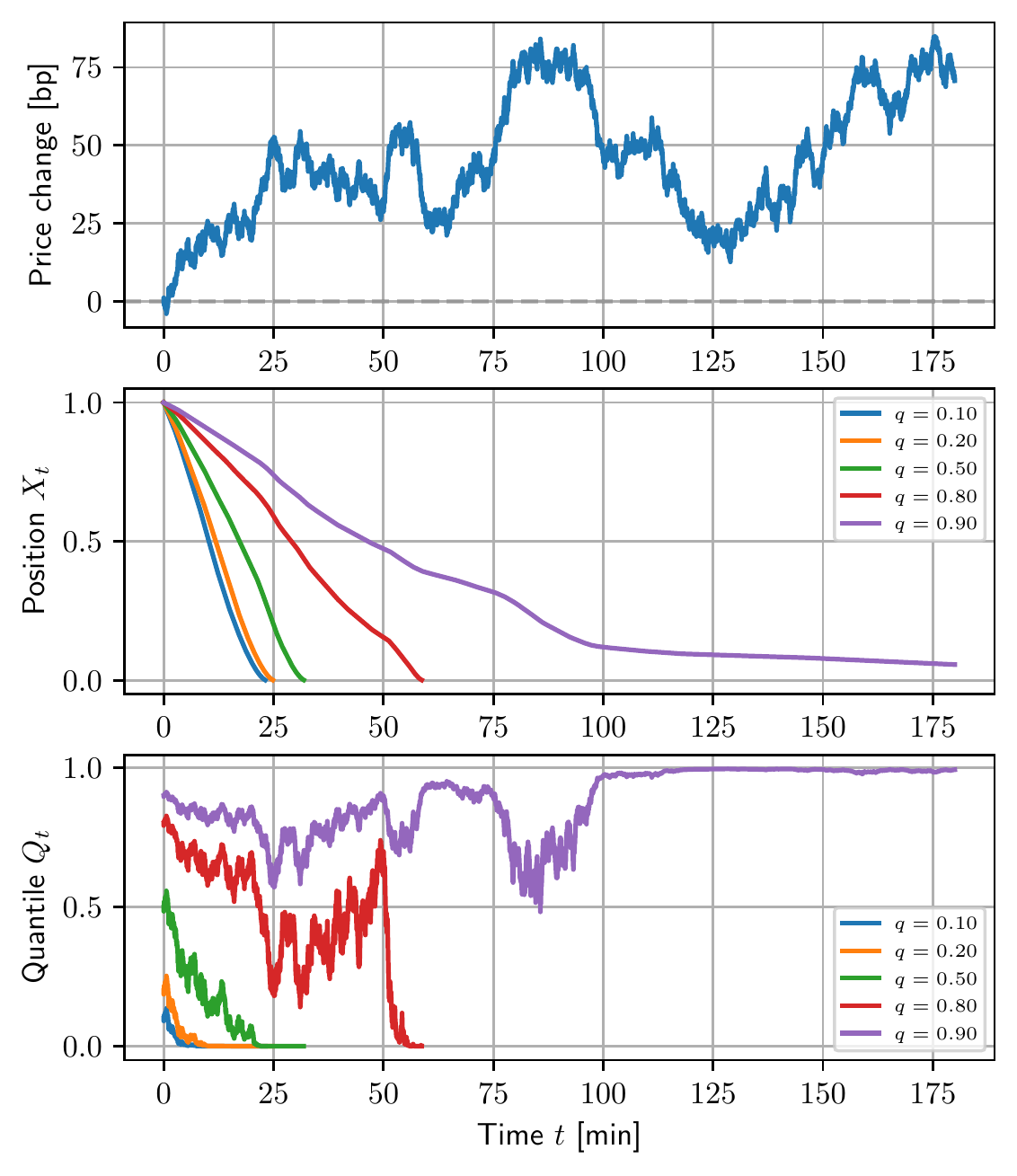}
\end{subfigure}
\caption{
	Illustration of the optimal adaptive liquidation processes with different values for target quantile $q \in \{0.1, 0.2, 0.5, 0.8, 0.9\}$ in two scenarios: when the price moves in an adverse direction (left), and when the price moves in a favorable direction (right).
	The plots in the top row show the realized price process over time, the plots in the middle row show the position processes $X_t$, and the plots in the bottom row show the quantile processes $Q_t$ associated with these strategies.}
\label{fig:numeric-sample-paths}
\end{figure}

\subsection{Comparison with Deterministic Strategies} \label{ssec:numerical-comparison}

We provide the detailed simulation results of the optimal adaptive strategy (\textsc{opt}) in a comparison with those of the deterministic strategies (\textsc{exp}, \textsc{vwap}) introduced in \S \ref{sec:deterministic}.

Figure \ref{fig:numeric-sample-paths2} illustrates the position process trajectories under these three strategies with target quantile $q=0.5$ in the two scenarios as in Figure \ref{fig:numeric-sample-paths}.
One can immediately observe that the deterministic strategies are not adaptive to the price changes.
The optimal adaptive strategy \textsc{opt} liquidates at a similar rate to the exponential schedule \textsc{exp} during the initial periods, but it deviates as soon as it adjusts its aggressiveness adaptively to the price changes.
In particular, it slows down when the price moves in the adverse direction (i.e., the quantile process $Q_t$ moves toward one).
Figure \ref{fig:numeric-average-path} shows the average position trajectories under these three strategies given the target quantile $q \in \{0.1, 0.5, 0.8\}$, aggregated across 100,000 runs of simulations.
Similarly to the above, we observe that \textsc{opt} trades more aggressively than \textsc{exp} when the target quantile is small and less aggressively when the target quantile is large.

\begin{figure}[H]
\centering
\begin{subfigure}{.505\textwidth}
  \centering
  \includegraphics[width=\linewidth]{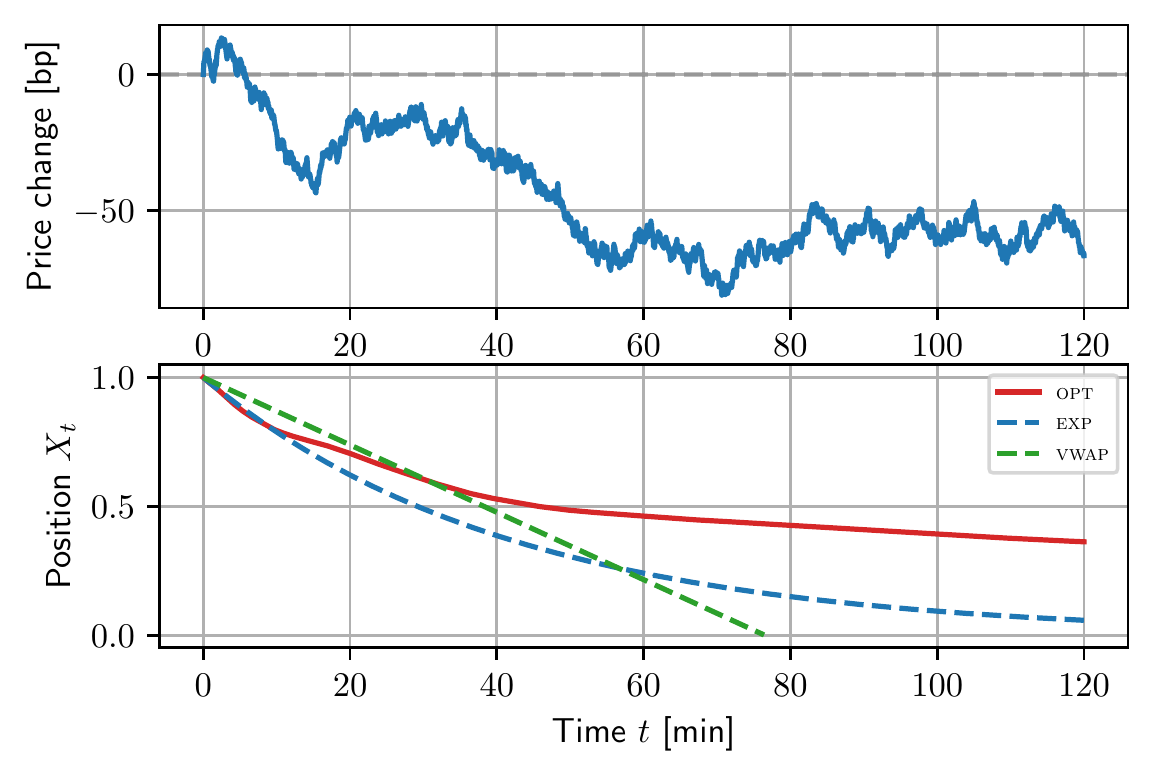}  
\end{subfigure}%
\begin{subfigure}{.495\textwidth}
  \centering
  \includegraphics[width=\linewidth]{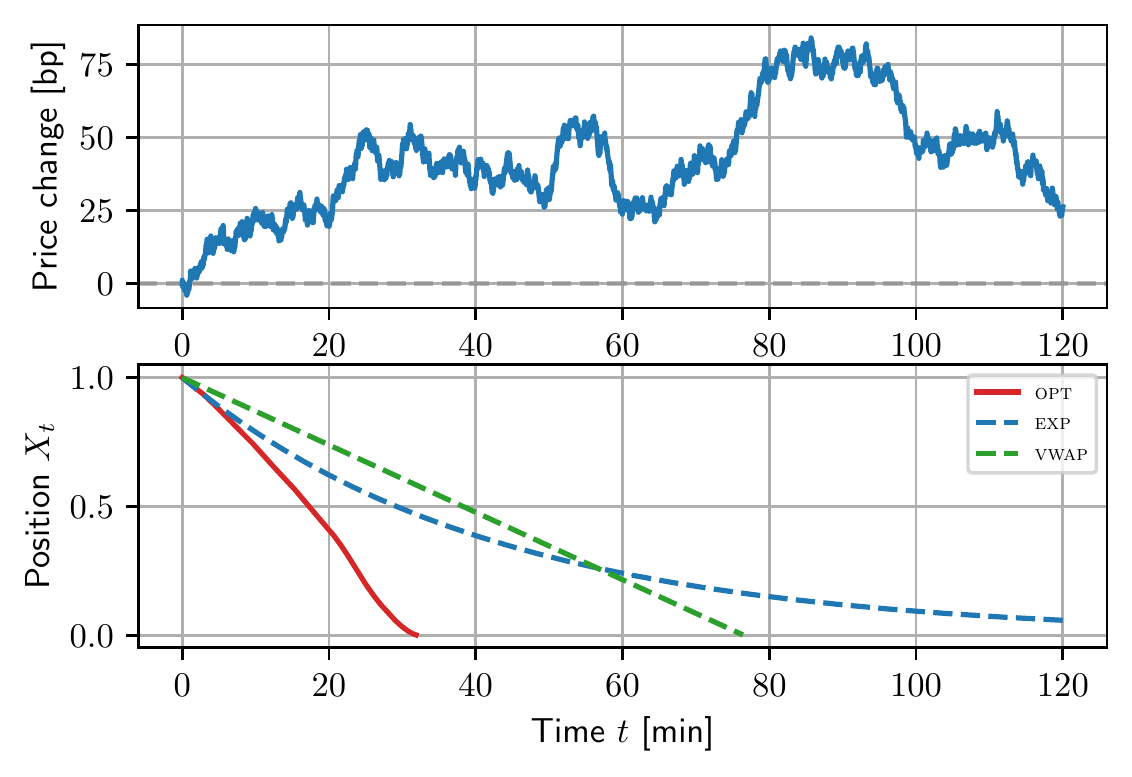}
\end{subfigure}
\caption{
	Illustration of liquidation processes under the optimal adaptive strategy (\textsc{opt}, red solid lines), the optimized deterministic schedule (\textsc{exp}, blue dashed lines), and the optimized VWAP schedule (\textsc{vwap}, green dashed lines) with the target quantile $q=0.5$, and in two scenarios: when the price moves in an adverse direction (left), and when the price moves in a favorable direction (right).
	}
\label{fig:numeric-sample-paths2}
\end{figure}

\begin{figure}[H]
\centering
\begin{subfigure}{.33\textwidth}
  \centering
  \includegraphics[width=\linewidth]{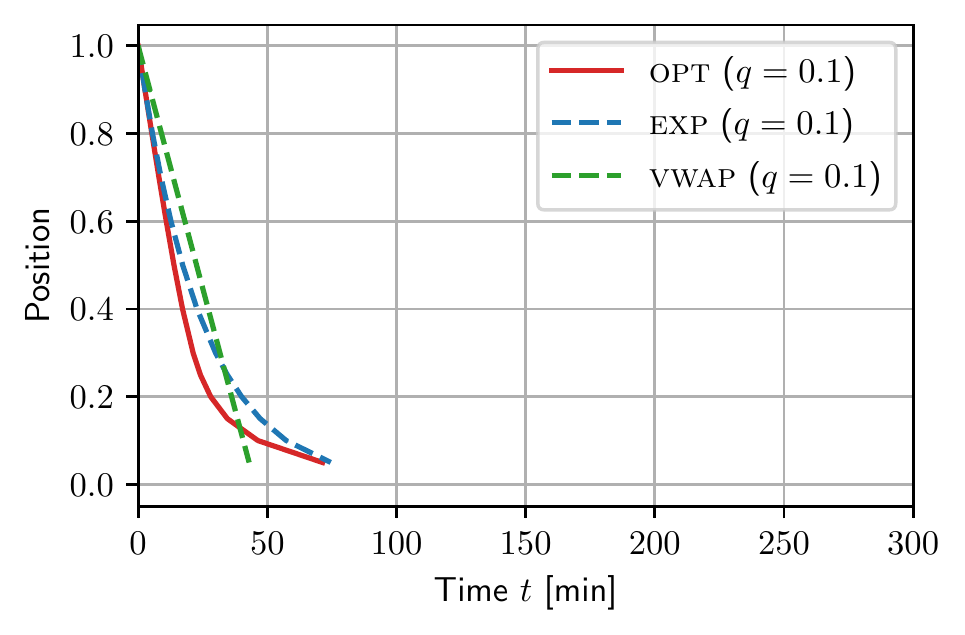}
\end{subfigure}%
\begin{subfigure}{.33\textwidth}
  \centering
  \includegraphics[width=\linewidth]{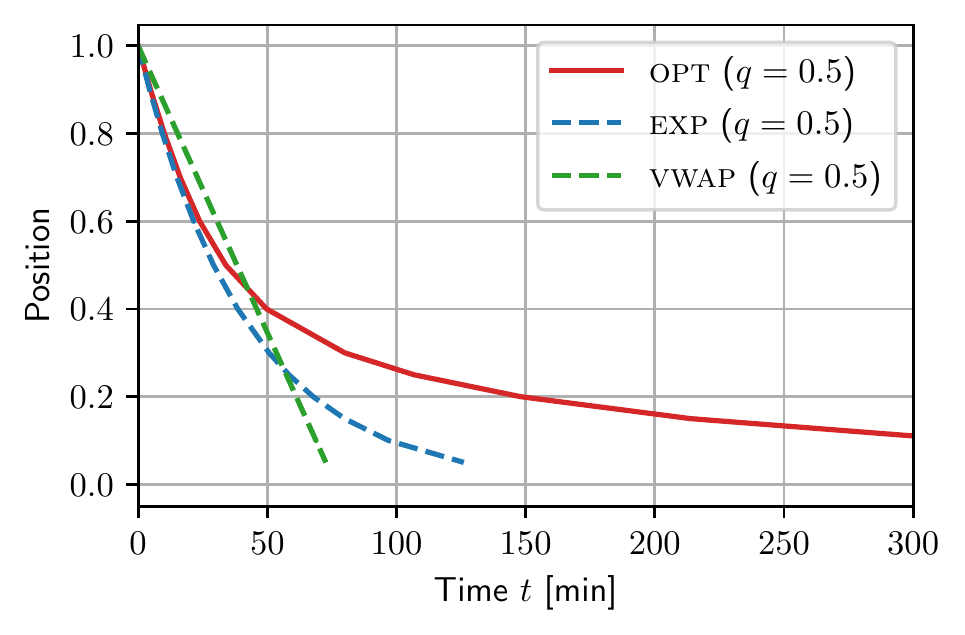}
\end{subfigure}
\begin{subfigure}{.33\textwidth}
  \centering
  \includegraphics[width=\linewidth]{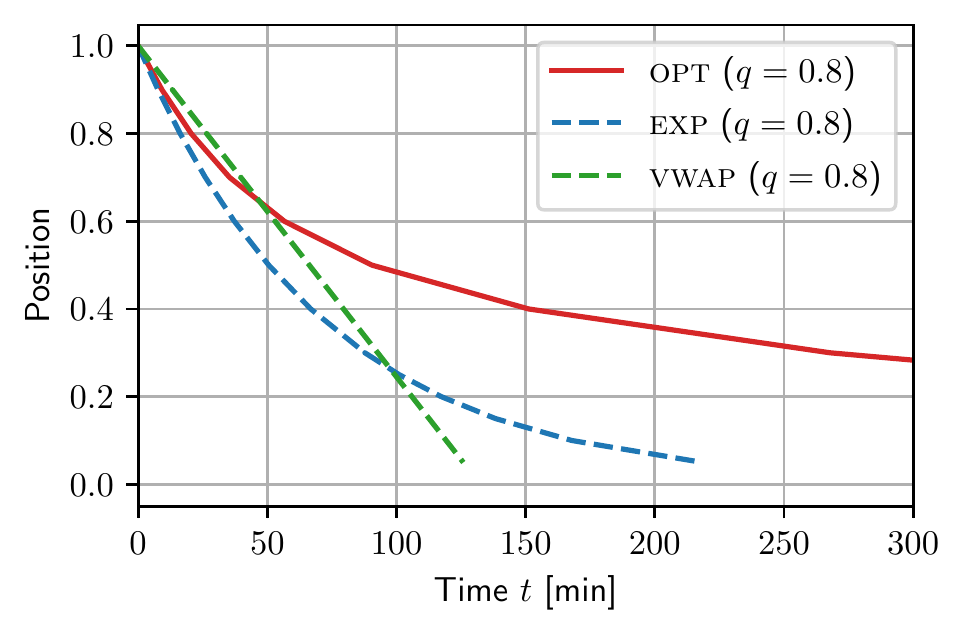}
\end{subfigure}%
\caption{
	Average liquidation processes under the optimal adaptive strategy (\textsc{opt}, red), the optimized deterministic schedule (\textsc{exp}, blue), and the optimized VWAP schedule (\textsc{vwap}, green) given the target quantile $q=0.1$ (left), $q=0.5$ (middle), and $q=0.8$ (right).
	Each curve reports the average time that the position process under each strategy falls behind a certain level, i.e., $\left\{ ( \mathbb{E}[ \min_t\{ X_t^\pi \leq y \} ], y ) \right\}_{y \in \{0.9, 0.8, \ldots, 0.05\}}$.
	}
\label{fig:numeric-average-path}
\end{figure}

Figure \ref{fig:numeric-histogram} shows the implementation shortfall distributions (i.e., the histograms of $C_\infty$) resulting from \textsc{opt} (top) and \textsc{exp} (bottom) with different values of the target quantile $q \in \{0.1, 0.5, 0.8\}$.
These histograms are obtained from $100{,}000$ simulation trials, where all the strategies see the same price process realization per simulation.
The resulting distributions are visually very different: \textsc{exp} yields a normal distribution whereas \textsc{opt} yields a distribution that has a sharp peak at the $q^\text{th}$ quantile.
Such a sharp peak can be explained by the threshold behavior of the optimal adaptive strategy, discussed at the end of \S \ref{ssec:opt-policy}.
Figure \ref{fig:numeric-comparison} visualizes these distributions for a wider range of target quantiles $q \in \{0.1,0.2, \ldots, 0.99\}$.
We observe that the implementation shortfall distribution induced by \textsc{opt} is more concentrated than the ones induced by \textsc{exp} and \textsc{vwap} when the target quantile $q$ is small, and it is the opposite when the target quantile $q$ is large (roughly speaking, it has a longer right tail, visually similar to an exponential distribution).

Table \ref{tab:numeric-stats} reports in detail the summary statistics of those implementation shortfall distributions.
We first confirm that the simulation results are consistent with our theoretic predictions, i.e., the measured CVaR values match with the values calculated from the expressions \eqref{eq:V-opt}, \eqref{eq:V-exp}, and \eqref{eq:V-vwap}.
We also see that the optimal adaptive strategy does not make an improvement over the deterministic strategies in terms of the average nor the variance as it specifically targets to minimize the CVaR value at a given target quantile.
But, it significantly improves the median value particularly in the risk-neutral regime (when $q \approx 1$).
See also a discussion on Figure \ref{fig:numeric-frontiers} below.

\begin{figure}[H]
\centering
\begin{subfigure}{.33\textwidth}
  \centering
  \includegraphics[width=\linewidth]{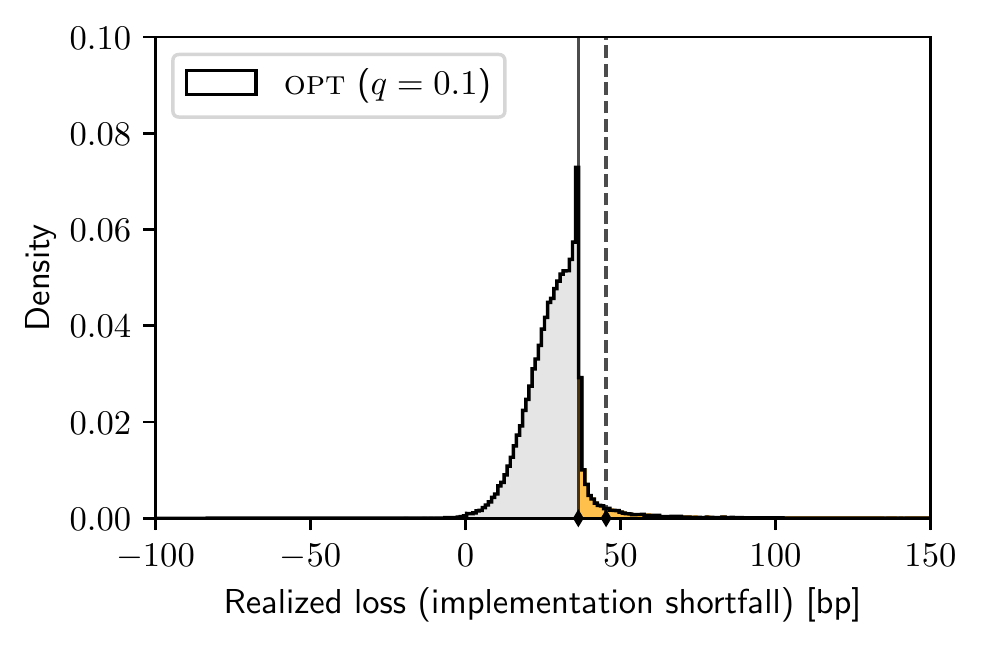}
  \includegraphics[width=\linewidth]{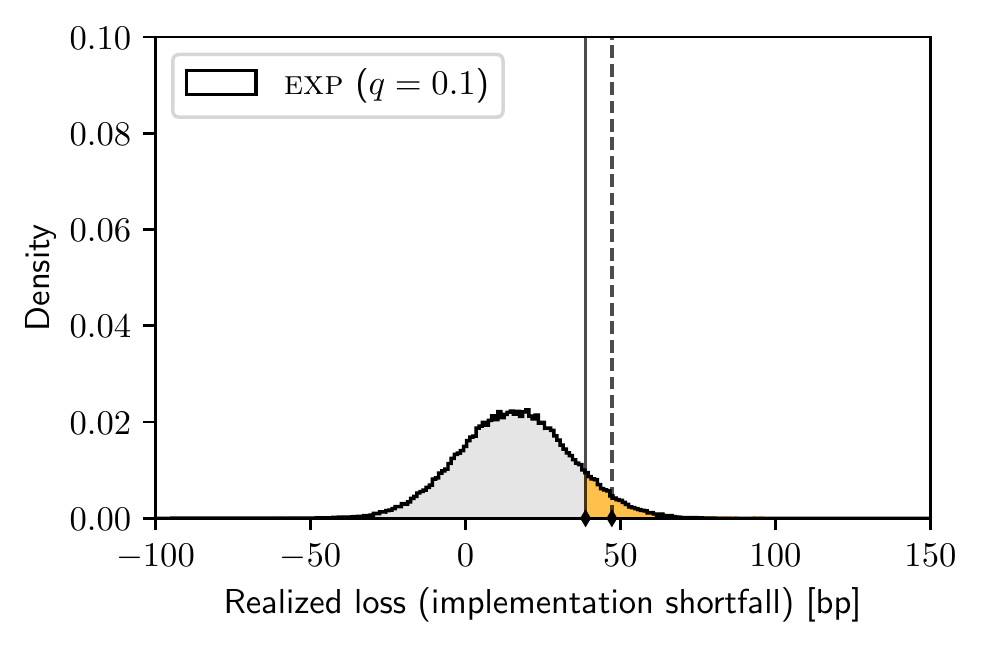}
\end{subfigure}%
\begin{subfigure}{.33\textwidth}
  \centering
  \includegraphics[width=\linewidth]{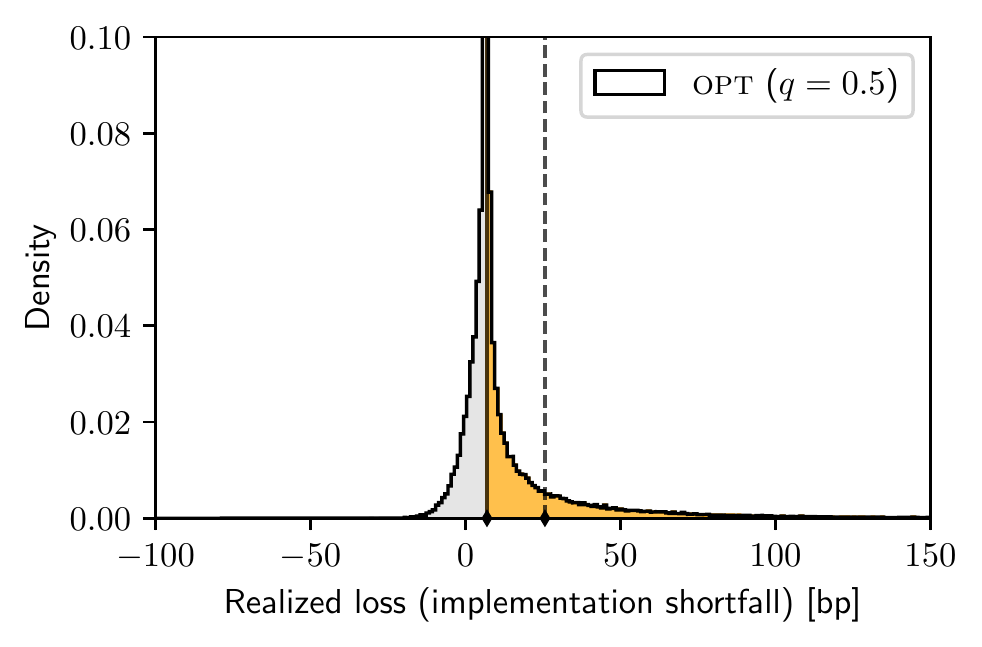}
  \includegraphics[width=\linewidth]{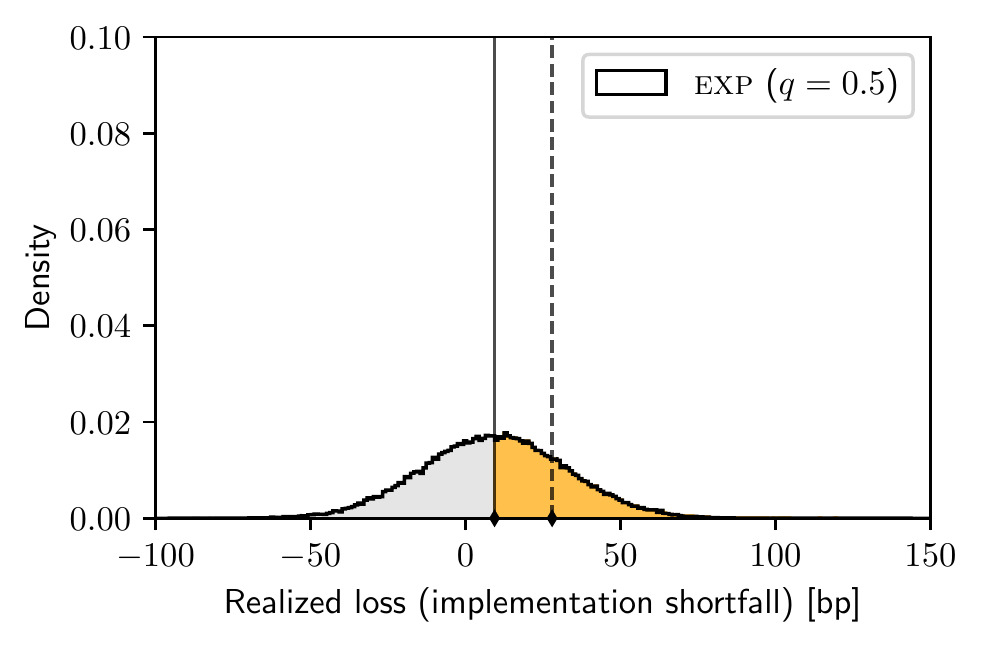}
\end{subfigure}
\begin{subfigure}{.33\textwidth}
  \centering
  \includegraphics[width=\linewidth]{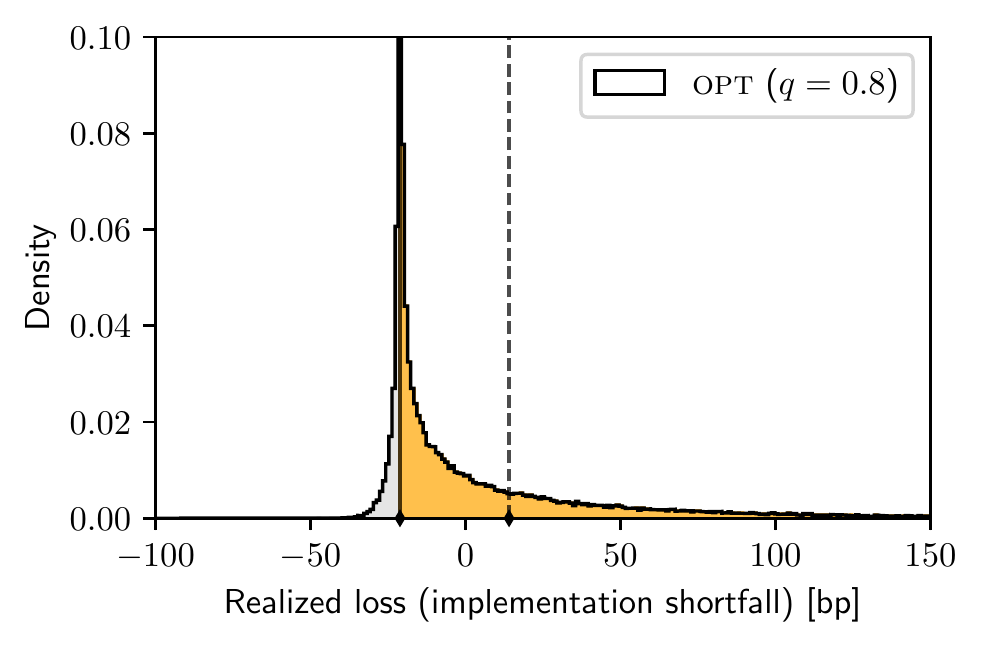}
  \includegraphics[width=\linewidth]{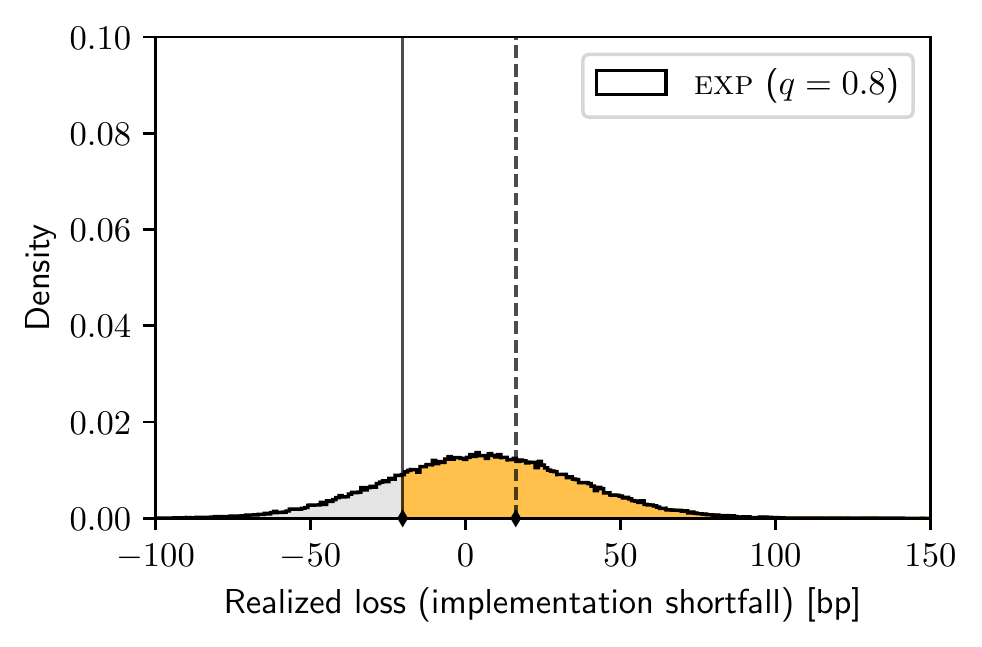}
\end{subfigure}%
\caption{
	The distributions of the implementation shortfall (i.e., the histogram of $C_\infty$) incurred by the optimal adaptive liquidation strategy (top) and the optimized deterministic schedule (bottom) given the target quantile $q \in \{0.1, 0.5, 0.8\}$ (left, middle, and right, respectively).
	In each plot, the solid vertical line represents the $q$-quantile (i.e., the value-at-risk $\var_q[C_\infty]$), and the dashed vertical line represents the tail average beyond the $q$-quantile (i.e., the conditional value-at-risk $\cvar_q[C_\infty]$).
	The highlighted area represents the worst $q$-fraction of the outcomes whose average corresponds to $\cvar_q[C_\infty]$.
	These are obtained from 100,000 runs of simulations. }
\label{fig:numeric-histogram}
\end{figure}

\begin{figure}[H]
	\centering
	\includegraphics[width=0.7\linewidth]{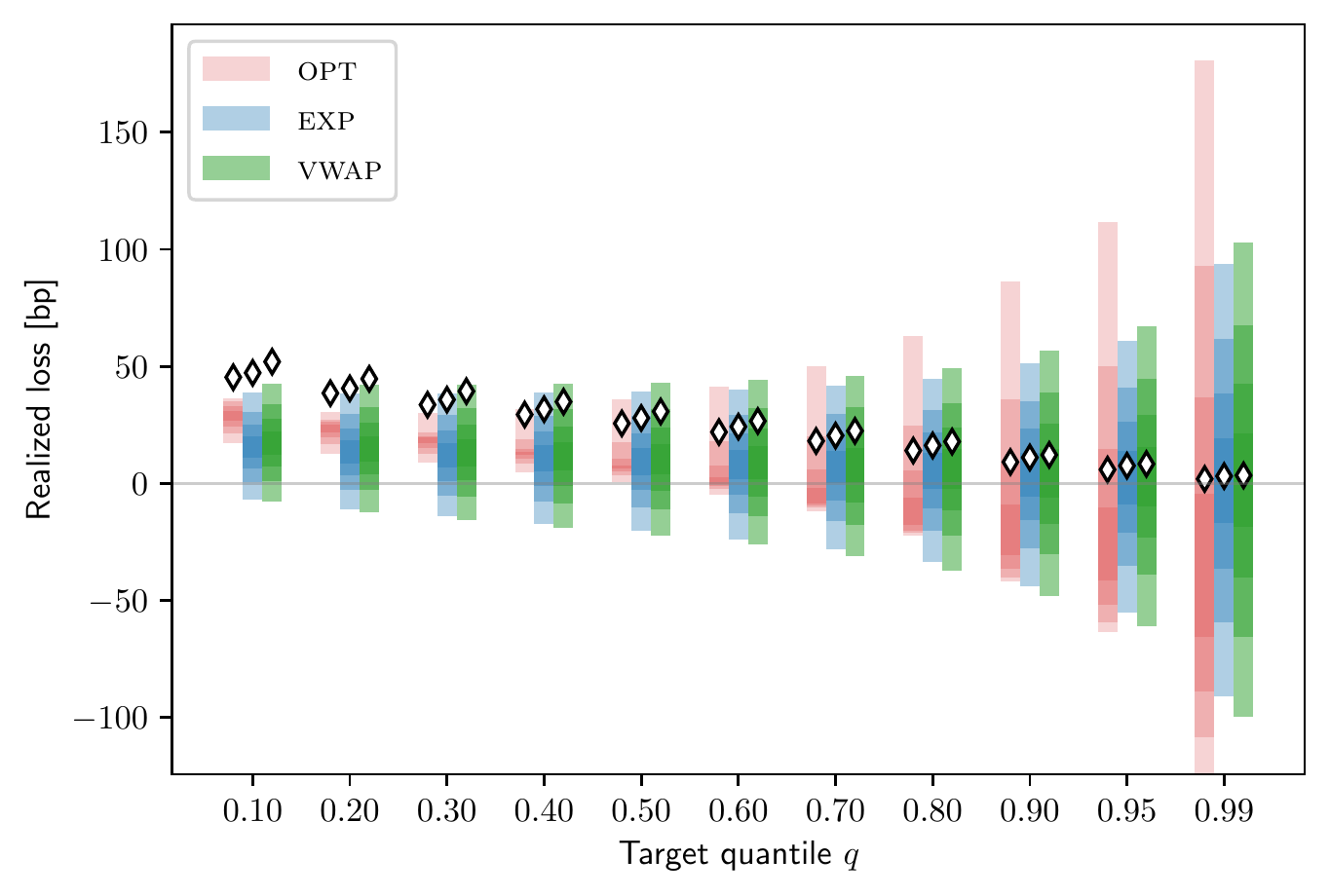} 
	\caption{
	Alternative visualization of the implementation shortfall distributions induced by the optimal adaptive strategy (\textsc{opt}, red), the optimized deterministic schedule (\textsc{exp}, blue), and the optimized VWAP schedule (\textsc{vwap}, green) given the target quantile $q \in \{0.1, 0.2, \ldots, 0.99\}$.
	Each distribution is represented with a colored bar with seven segments whose edges are $0.1, 0.2, \ldots, 0.9$-quantiles, e.g., the brightest segment on the top represents the range of the implementation shortfall between $0.1$-quantile and $0.2$-quantile, and the darkest segment represents the range of implementation shortfall between $0.4$-quantile and $0.6$-quantile.
	The diamond-shaped dots report the CVaR values of these distributions.
	}
	\label{fig:numeric-comparison}
\end{figure}

        \begin{table}[H]
        \centering \small
        \begin{tabular}{*{10}c}
        \toprule
        \thead{\makecell{Target\\quantile}} & \thead{Policy} & \thead{CVaR (theo.)} & \thead{VaR} & \thead{Average} & \thead{Median} & \thead{\makecell{Std.\\dev.}} & \thead{\makecell{50\%\\time}} & \thead{\makecell{95\%\\time}} \\ 
\midrule 
\multirow{3}{*}{$q=0.10$} & \textsc{opt} & $45.37$ ($44.96$)  & $36.42$  & $28.61$  & $29.11$  & $10.34$  & $14$  & $71$  \\ 
 & \textsc{exp} & $47.24$ ($47.02$)  & $38.75$  & $15.79$  & $15.77$  & $17.86$  & $17$  & $75$  \\ 
 & \textsc{vwap} & $52.00$ ($51.75$)  & $42.61$  & $17.37$  & $17.30$  & $19.65$  & $23$  & $43$  \\ 
\midrule 
\multirow{3}{*}{$q=0.20$} & \textsc{opt} & $38.55$ ($38.26$)  & $27.40$  & $23.95$  & $23.62$  & $12.93$  & $18$  & $145$  \\ 
 & \textsc{exp} & $40.62$ ($40.44$)  & $29.72$  & $13.60$  & $13.57$  & $19.26$  & $20$  & $87$  \\ 
 & \textsc{vwap} & $44.69$ ($44.51$)  & $32.74$  & $14.96$  & $14.88$  & $21.18$  & $26$  & $50$  \\ 
\midrule 
\multirow{3}{*}{$q=0.30$} & \textsc{opt} & $33.62$ ($33.40$)  & $20.28$  & $20.33$  & $18.52$  & $16.29$  & $22$  & $241$  \\ 
 & \textsc{exp} & $35.80$ ($35.66$)  & $22.76$  & $12.02$  & $11.99$  & $20.51$  & $23$  & $98$  \\ 
 & \textsc{vwap} & $39.41$ ($39.25$)  & $24.98$  & $13.21$  & $13.15$  & $22.56$  & $30$  & $57$  \\ 
\midrule 
\multirow{3}{*}{$q=0.40$} & \textsc{opt} & $29.44$ ($29.28$)  & $13.67$  & $17.20$  & $13.02$  & $20.23$  & $27$  & $363$  \\ 
 & \textsc{exp} & $31.72$ ($31.58$)  & $16.22$  & $10.66$  & $10.66$  & $21.80$  & $26$  & $111$  \\ 
 & \textsc{vwap} & $34.91$ ($34.75$)  & $17.76$  & $11.72$  & $11.64$  & $23.98$  & $34$  & $64$  \\ 
\midrule 
\multirow{3}{*}{$q=0.50$} & \textsc{opt} & $25.65$ ($25.48$)  & $6.95$  & $14.38$  & $6.95$  & $24.86$  & $34$  & $518$  \\ 
 & \textsc{exp} & $27.96$ ($27.80$)  & $9.39$  & $9.41$  & $9.39$  & $23.24$  & $29$  & $126$  \\ 
 & \textsc{vwap} & $30.74$ ($30.60$)  & $10.30$  & $10.34$  & $10.30$  & $25.56$  & $38$  & $73$  \\ 
\midrule 
\multirow{3}{*}{$q=0.60$} & \textsc{opt} & $21.95$ ($21.79$)  & $-0.42$  & $11.72$  & $0.40$  & $30.46$  & $44$  & $722$  \\ 
 & \textsc{exp} & $24.27$ ($24.10$)  & $1.79$  & $8.19$  & $8.20$  & $24.97$  & $34$  & $145$  \\ 
 & \textsc{vwap} & $26.67$ ($26.52$)  & $1.94$  & $8.98$  & $9.00$  & $27.46$  & $44$  & $84$  \\ 
\midrule 
\multirow{3}{*}{$q=0.70$} & \textsc{opt} & $18.17$ ($18.01$)  & $-9.23$  & $9.17$  & $-6.21$  & $37.56$  & $60$  & $998$  \\ 
 & \textsc{exp} & $20.44$ ($20.27$)  & $-7.46$  & $6.92$  & $6.95$  & $27.24$  & $40$  & $173$  \\ 
 & \textsc{vwap} & $22.48$ ($22.31$)  & $-8.21$  & $7.59$  & $7.62$  & $29.97$  & $52$  & $100$  \\ 
\midrule 
\multirow{3}{*}{$q=0.80$} & \textsc{opt} & $14.07$ ($13.92$)  & $-21.12$  & $6.62$  & $-13.14$  & $47.48$  & $90$  & $1404$  \\ 
 & \textsc{exp} & $16.24$ ($16.05$)  & $-20.25$  & $5.53$  & $5.51$  & $30.62$  & $51$  & $218$  \\ 
 & \textsc{vwap} & $17.89$ ($17.66$)  & $-22.28$  & $6.09$  & $6.07$  & $33.71$  & $66$  & $126$  \\ 
\midrule 
\multirow{3}{*}{$q=0.90$} & \textsc{opt} & $9.16$ ($9.03$)  & $-41.76$  & $3.92$  & $-21.65$  & $64.62$  & $169$  & $2125$  \\ 
 & \textsc{exp} & $11.08$ ($10.87$)  & $-43.93$  & $3.83$  & $3.78$  & $37.22$  & $75$  & $323$  \\ 
 & \textsc{vwap} & $12.22$ ($11.96$)  & $-48.20$  & $4.23$  & $4.20$  & $41.02$  & $98$  & $186$  \\ 
\midrule 
\multirow{3}{*}{$q=0.95$} & \textsc{opt} & $5.92$ ($5.85$)  & $-64.35$  & $2.35$  & $-28.15$  & $82.52$  & $294$  & $2874$  \\ 
 & \textsc{exp} & $7.59$ ($7.35$)  & $-71.53$  & $2.68$  & $2.59$  & $45.23$  & $110$  & $477$  \\ 
 & \textsc{vwap} & $8.34$ ($8.09$)  & $-78.99$  & $2.93$  & $2.93$  & $49.82$  & $145$  & $275$  \\ 
\midrule 
\multirow{3}{*}{$q=0.99$} & \textsc{opt} & $1.94$ ($2.09$)  & $-132.60$  & $0.59$  & $-37.62$  & $129.72$  & $873$  & $4740$  \\ 
 & \textsc{exp} & $3.16$ ($2.90$)  & $-167.00$  & $1.20$  & $1.12$  & $71.96$  & $279$  & $1207$  \\ 
 & \textsc{vwap} & $3.47$ ($3.19$)  & $-183.10$  & $1.33$  & $1.49$  & $79.16$  & $366$  & $696$  \\ 
        \bottomrule
        \end{tabular}
        \caption{
        		Summary statistics of implementation shortfall $C_\infty$ incurred by the optimal adaptive strategy (\textsc{opt}), the optimized deterministic schedule (\textsc{exp}), and the optimized VWAP schedule (\textsc{vwap}), given $x=1$, $\sigma=5.06$, $\eta=1.56 \times 10^3$, and the target quantile $q \in \{0.1, 0.2, \ldots, 0.99\}$.
		For each combination of a policy and a target quantile, it reports the CVaR value, the VaR value, the average, the median, and the standard deviation of implementation shortfall $C_\infty$, represented in basis points.
		It additionally reports the average time to complete 50\% of the execution and the average time to complete 95\% of the execution, represented in minutes.
		These statistics are measured from 100,000 runs of simulations.
		The numbers in parentheses in the third column report the theoretically predicted CVaR values, computed with the expressions \eqref{eq:V-opt}, \eqref{eq:V-exp}, and \eqref{eq:V-vwap}.
                }
        \label{tab:numeric-stats}
        \end{table}

Figure \ref{fig:numeric-frontiers} demonstrates some additional advantages of the adaptive strategy other than minimizing the CVaR value.
Here, the target quantile $q$ is considered as a control parameter to the strategies on the behalf of a trader who may not be particularly interested in minimizing the CVaR value, and we compare three families of strategies (\textsc{opt}, \textsc{exp}, and \textsc{vwap} with the target quantile ranging from $0.1$ to $0.99$) in terms of the mean, the median, and the tail probability of the resulting implementation shortfall distributions.
It is shown in the left plot that, by implementing the adaptive strategy \textsc{opt} with a suitably chosen target quantile, it can achieve a smaller median value than any of deterministic strategies that yield the same average loss value.
Similarly, it is shown in the right plot that a smaller tail probability can be achieved by \textsc{opt} for a given target threshold level: for example, when wanting to avoid the event that the implementation shortfall exceeds 25 basis points, the trader can implement \textsc{opt} with the target quantile $q=0.4$ so that such event takes place with probability 12.9\%, whereas the probability is 25.3\% under the best deterministic schedule (\textsc{exp} with $q=0.6$).

\begin{figure}[H]
\centering
\begin{subfigure}{.49\textwidth}
  \centering
  \includegraphics[width=\linewidth]{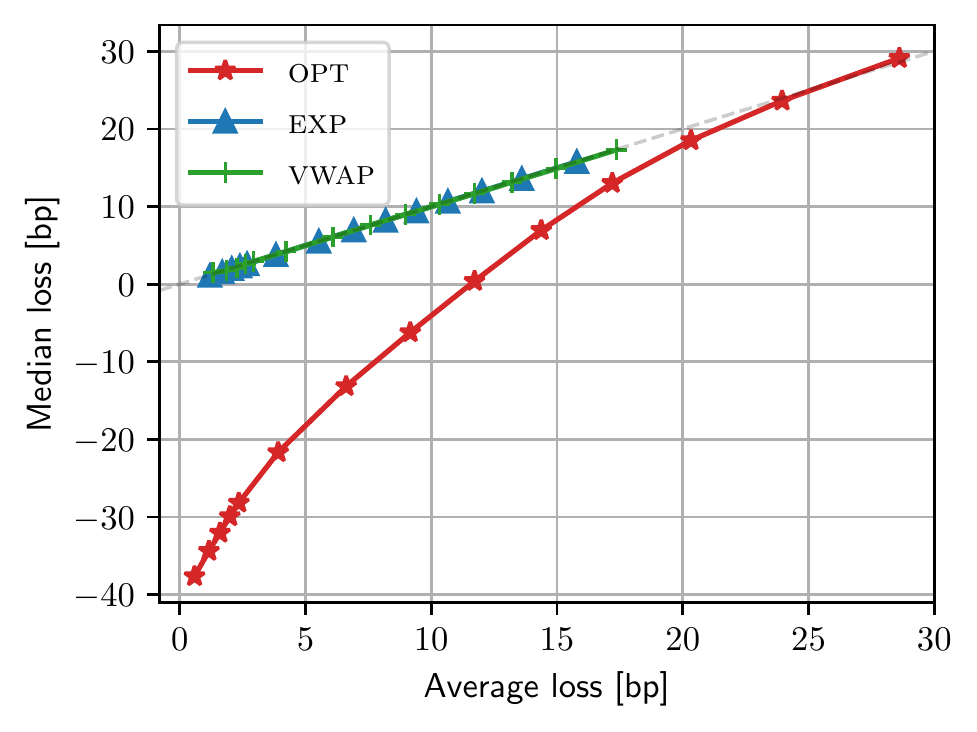}
\end{subfigure}
\begin{subfigure}{.49\textwidth}
  \centering
  \includegraphics[width=\linewidth]{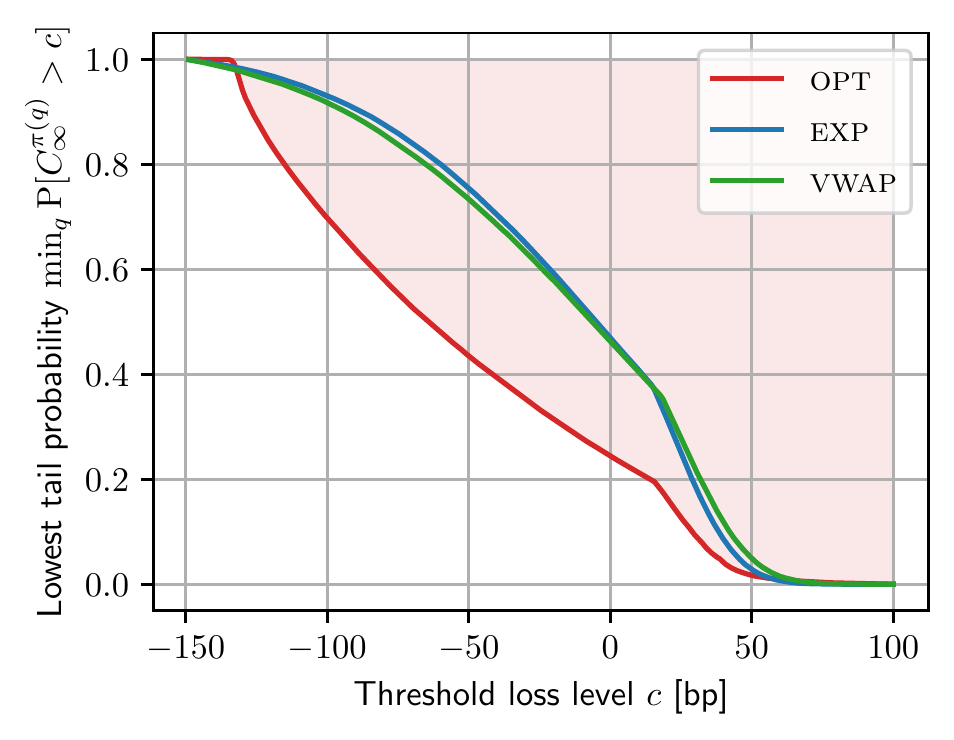}
\end{subfigure}%
\caption{
	Plots that characterize three families of implementation shortfall distributions produced by the optimal adaptive strategy (\textsc{opt}, red), the optimized deterministic schedule (\textsc{exp}, blue), and the optimized VWAP schedule (\textsc{vwap}, green).
	The left plot shows the average and the median loss values that can be obtained by a strategy (either \textsc{opt}, \textsc{exp}, or \textsc{vwap}) when $q$ varies, i.e., $\{ ( \mathbb{E}[ C_\infty^{\pi_q}], \text{Median}[ C_\infty^{\pi_q} ] ) \}_{q \in \{0.1, \ldots, 0.99\}}$ where $\pi_q$ denotes the strategy with a target quantile $q$.
	The right plot reports the pointwise minimum of the complementary cumulative distribution functions induced by a strategy with the different values of target quantiles, i.e., $\min_{q \in \{0.1, \ldots, 0.99\}}\mathbb{P}[ C_\infty^{\pi_q} > c ]$, the lowest tail probability that the implementation shortfall exceeds a given threshold level $c$ when the target quantile $q$ for the strategy can be arbitrarily chosen.
	}
\label{fig:numeric-frontiers}
\end{figure}


{\small
  \singlespacing
\bibliography{cvar-exec}

\begin{thebibliography}{29}
\providecommand{\natexlab}[1]{#1}
\providecommand{\url}[1]{\texttt{#1}}
\expandafter\ifx\csname urlstyle\endcsname\relax
  \providecommand{\doi}[1]{doi: #1}\else
  \providecommand{\doi}{doi: \begingroup \urlstyle{rm}\Url}\fi

\bibitem[Almgren and Chriss(2000)]{Almgren00}
Robert Almgren and Neil Chriss.
\newblock Optimal execution of portfolio transactions.
\newblock 2000.

\bibitem[Artzner et~al.(1999)Artzner, Delbaen, Eber, and Heath]{Artzner99}
Philippe Artzner, Freddy Delbaen, Jean-Mark Eber, and David Heath.
\newblock Coherent measures of risk.
\newblock \emph{Mathematical Finance}, 9\penalty0 (3):\penalty0 203--228, 1999.

\bibitem[Artzner et~al.(2007)Artzner, Delbaen, Eber, Heath, and Ku]{Artzner07}
Philippe Artzner, Freddy Delbaen, Jean-Marc Eber, David Heath, and Hyejin Ku.
\newblock Coherent multiperiod risk adjusted values and {B}ellman's principle.
\newblock \emph{Annals of Operations Research}, 152:\penalty0 5--22, 2007.

\bibitem[Backhoff-Veraguas and Tangpi(2020)]{Backhoff20}
Julio Backhoff-Veraguas and Ludovic Tangpi.
\newblock On the dynamic representation of some time-inconsistent risk measures
  in a {B}rownian filtration.
\newblock \emph{Mathematics and Financial Economics}, 14:\penalty0 433--460,
  2020.

\bibitem[Barbu and Precupanu(2012)]{Barbu}
Viorel Barbu and Teodor Precupanu.
\newblock \emph{Convexity and Optimization in Banach Spaces}.
\newblock Springer, 4th edition, 2012.

\bibitem[B{\"a}uerle and Ott(2011)]{Bauerle11}
Nicole B{\"a}uerle and Jonathan Ott.
\newblock Markov decision processes with average-value-at-risk criteria.
\newblock \emph{Mathematical Methods of Operational Research}, 74\penalty0
  (4):\penalty0 361--379, 2011.

\bibitem[Chapman et~al.(2018)Chapman, Lacotte, Smith, Yang, Han, Pavone, and
  Tomlin]{Chapman18}
Margaret~P. Chapman, Jonathan~P. Lacotte, Kevin~M. Smith, Insoon Yang, Yuxi
  Han, Marco Pavone, and Claire~J. Tomlin.
\newblock Risk-sensitive safety specifications for stochastic system using
  conditional value-at-risk.
\newblock 2018.

\bibitem[Chow et~al.(2015)Chow, Tamar, Mannor, and Pavone]{Chow15}
Yinlam Chow, Aviv Tamar, Shie Mannor, and Marco Pavone.
\newblock Risk-sensitive and robust decision-making: a {CVaR} optimization
  approach.
\newblock \emph{Advances in Neural Information Processing Systems 28 (NIPS
  2015)}, 2015.

\bibitem[Chow et~al.(2018)Chow, Ghavamzadeh, Janson, and Pavone]{Chow18}
Yinlam Chow, Mohammad Ghavamzadeh, Lucas Janson, and Marco Pavone.
\newblock Risk-constrained reinforcement learning with percentile risk
  criteria.
\newblock \emph{Journal of Machine Learning Research}, 18:\penalty0 1--51,
  2018.

\bibitem[Courant and Hilbert(1953)]{CourantHilbert53}
R.~Courant and D.~Hilbert.
\newblock \emph{Methods of Mathematical Physics}, volume~I.
\newblock New York and London (Interscience Publishers), 1953.

\bibitem[Durrett(2010)]{Durrett}
Rick Durrett.
\newblock \emph{Probability: Theory and Examples}.
\newblock Cambridge University Press, 4th edition, 2010.

\bibitem[Forsyth(2011)]{Forsyth11}
Peter~A. Forsyth.
\newblock A {H}amilton--{J}acobi--{B}ellman approach to optimal trade
  execution.
\newblock \emph{Applied Numerical Mathematics}, 61:\penalty0 241--265, 2011.

\bibitem[Forsyth et~al.(2012)Forsyth, Kennedy, Tse, and Windcliff]{Forsyth12}
Peter~A. Forsyth, Shannon Kennedy, Shu~Tong Tse, and Heath Windcliff.
\newblock Optimal trade execution: {A} mean quadratic variation approach.
\newblock \emph{Journal of Economic Dynamics \& Control}, 36:\penalty0
  1971--1991, 2012.

\bibitem[Gatheral and Schied(2011)]{Gatheral11}
Jim Gatheral and Alexander Schied.
\newblock Optimal trade execution under geometric {B}rownian motion in the
  {A}lmgren and {C}hriss framework.
\newblock \emph{International Journal of Theoretical and Applied Finance},
  14\penalty0 (3):\penalty0 353--368, 2011.

\bibitem[Glasserman and Xu(2013)]{Glasserman13}
Paul Glasserman and Xingbo Xu.
\newblock Robust portfolio control with stochastic factor dynamics.
\newblock \emph{Operations Research}, 61\penalty0 (4):\penalty0 874--893, 2013.

\bibitem[Huang and Guo(2016)]{Huang16}
Yonghui Huang and Xianping Guo.
\newblock Minimum average value-at-risk for finite horizon semi-{M}arkov
  decision processes in continuous time.
\newblock \emph{SIAM Journal on Optimization}, 26\penalty0 (1):\penalty0 1--28,
  2016.

\bibitem[Kissell and Malamut(2005)]{Kissell05}
Robert Kissell and Roberto Malamut.
\newblock Understanding the profit and loss distribution of trading algorithms.
\newblock In B.~R. Bruce, editor, \emph{Algorithmic Trading}, pages 41--49.
  Institutional Investor, 2005.

\bibitem[Li et~al.(2020)Li, Zhong, and Brandeau]{Li20}
Xiaocheng Li, Huaiyang Zhong, and Margaret~L. Brandeau.
\newblock Quantile {M}arkov decision processes.
\newblock 2020.

\bibitem[Lin et~al.(2015)Lin, Chen, and Pe{\~n}a]{Lin15}
Qihang Lin, Xi~Chen, and Javier Pe{\~n}a.
\newblock A trade execution model under a composite dynamic coherent risk
  measure.
\newblock \emph{Operations Research Letters}, 43:\penalty0 52--58, 2015.

\bibitem[Lorenz and Almgren(2007)]{Almgren07}
Julian Lorenz and Robert Almgren.
\newblock Adaptive arrival price.
\newblock In B.~R. Bruce, editor, \emph{Algorithmic Trading III}, pages 59--66.
  Institutional Investor, 2007.

\bibitem[Lorenz and Almgren(2011)]{Almgren11}
Julian Lorenz and Robert Almgren.
\newblock Mean-variance optimal adaptive execution.
\newblock \emph{Applied Mathematical Finance}, 18\penalty0 (5):\penalty0
  395--422, 2011.

\bibitem[Miller and Yang(2017)]{Miller17}
Christopher~W. Miller and Insoon Yang.
\newblock Optimal control of conditional value-at-risk in continuous time.
\newblock \emph{SIAM Journal on Control and Optimization}, 55\penalty0
  (2):\penalty0 856--884, 2017.

\bibitem[Pflug and Pichler(2016)]{Pflug16}
Georg~Ch. Pflug and Alois Pichler.
\newblock Time-inconsistent multistage stochastic programs: Martingale bounds.
\newblock \emph{European Journal of Operational Research}, 249:\penalty0
  155--163, 2016.

\bibitem[Polyanin and Zaitsev(2003)]{ODEhandbook}
Andrei~D. Polyanin and Valentin~F. Zaitsev.
\newblock \emph{Handbook of Exact Solutions for Ordinary Differential
  Equations}.
\newblock Chapman \& Hall/CRC Press, 2003.

\bibitem[Protter(2015)]{Protter}
Philip Protter.
\newblock \emph{Stochastic Integration and Differential Equations}.
\newblock Springer, 2015.

\bibitem[Rockafellar and Uryasev(2002)]{Rockafellar02}
R.~Tyrrell Rockafellar and Stanislav Uryasev.
\newblock Conditional value-at-risk for general loss distributions.
\newblock \emph{Journal of Banking {\&} Finance}, 26:\penalty0 1443--1471,
  2002.

\bibitem[Schied and Sch\"{o}neborn(2009)]{Schied09}
Alexander Schied and Torsten Sch\"{o}neborn.
\newblock Risk aversion and the dynamics of optimal liquidation strategies in
  illiquid markets.
\newblock \emph{Finance and Stochastics}, 13\penalty0 (2):\penalty0 181--204,
  2009.

\bibitem[Shapiro(2009)]{Shapiro09}
Alexander Shapiro.
\newblock On a time consistency concept in risk averse multistage stochastic
  programming.
\newblock \emph{Operations Research Letters}, 37:\penalty0 143--147, 2009.

\bibitem[Sion(1958)]{Sion58}
Maurice Sion.
\newblock On general minimax theorems.
\newblock \emph{Pacific Journal of Mathematics}, 8\penalty0 (1):\penalty0
  171--176, 1958.

\end{thebibliography}
}

\newpage

\appendix

\newpage
\newpage
\seclabel{Organization of appendix.}
The appendix is organized as follows.
In Appendix~\ref{app:deterministic}, we identify the optimal deterministic strategy and its performance for \S \ref{sec:deterministic}.
The CVaR performance of the optimal deterministic strategy is utilized as an upper bound on the CVaR performance of the optimal adaptive strategy.
In Appendix~\ref{app:proof-problem}, we provide the basic characterizations of S-CVaR measure introduced in \S \ref{sec:problem}.
In Appendix~\ref{app:proof-cvar-dp}, we provide the preliminary characterizations of the value function, and by applying Sion's minimax theorem, we prove Theorem~\ref{thm:cvar-minimax} and Theorem~\ref{thm:cvar-dp} stated in \S \ref{sec:cvar-dp}.
The main challenge here is to verify the conditions of Sion's minimax theorem.
In Appendix~\ref{app:proof-opt}, we provide proofs for \S \ref{sec:opt}.
We first state and prove Theorem~\ref{thm:optimality-detail} from which Theorem~\ref{thm:verification} and Theorem~\ref{thm:policy-optimality} follow almost immediately.
Proposition~\ref{prop:cvar-dp-dynkin}, Proposition~\ref{prop:Emden-Fowler} and Theorem~\ref{thm:value-function} are proven separately.

\section{Optimal Deterministic Schedules} \label{app:deterministic}

\begin{lemma}
        For any $a, b > 0$,
        \begin{equation}
                \min_{x \in \mathbb{R}_+}\left\{ \frac{a}{x} + b \sqrt{x} \right\}
                        = \left. \frac{a}{x} + b \sqrt{x} ~ \right|_{x = \left( \frac{2a}{b} \right)^{\frac{2}{3}} }
                        = \frac{3 a^{\frac{1}{3}} b^{\frac{2}{3}} }{ 2^{\frac{2}{3}} }.
        \end{equation}
\end{lemma}

\begin{proof}
        Let $f(x) \defeq \frac{a}{x} + b \sqrt{x}$.
        Since $f'(x) = -\frac{a}{x^2} + \frac{b}{2\sqrt{x}}$, the equation $f'(x)=0$ has a unique solution at $x = \left( \frac{2a}{b} \right)^{\frac{2}{3}}$.
\end{proof}

\begin{proof}[\proofnamest{Proof of Proposition~\ref{prop:exp-schedule}}]
        We prove the optimality of exponential schedules and identify the optimal decaying rate.

        First we consider a mean-variance optimization problem:
        \begin{equation} \label{eq:mv-lagrangian}
                \minimize_{\pi \in \mathcal{D}(x)} ~~ \E[ C_\infty^{x,\pi} ] + \lambda \text{Var}[ C_\infty^{x,\pi} ],
        \end{equation}
        where $\mathcal{D}(x)$ is the set of all deterministic policies and
        $\lambda \in (0,\infty)$ is a penalty for variance term.  Applying \eqref{eq:normcvar},
        this is equivalent to an optimization over the deterministic trajectories of
        $(X_t)_{t \geq 0}$:
        \begin{equation}
                \inf_{X:X_0=x} \int_{t=0}^\infty \left( \frac{\eta}{2} \dot{X}_t^2 + \lambda \sigma^2 X_t^2 \right)dt,
        \end{equation}
        where $\dot{X}_t \defeq dX_t/dt$.  By applying standard calculus of variations arguments
        \citep{CourantHilbert53}, we deduce that the optimal schedule $X^\star$ has to satisfy the
        Euler-Lagrange equation
        $ 2 \lambda \sigma^2 X_t^\star - \eta \ddot{X}_t^\star = 0, $ at each time $t$, with
        boundary conditions $X_0^\star = x$ and $\lim_{t \rightarrow \infty} X_t^\star = 0$.  The
        solution is uniquely given by an exponential schedule
        \begin{equation}
                X_t^\star = x \exp\left( - t /  \rho_\lambda \right),
        \end{equation}
        with the decay rate $\rho_\lambda \defeq \sqrt{ \frac{ \eta }{ 2 \lambda \sigma^2 } } \in (0, \infty)$, and such a schedule yields
        \begin{equation} \label{eq:mv-exponential}
                \E[ C_\infty^{x,\pi^\star} ] = \frac{\eta x^2}{4 \rho_\lambda}
                , \quad
                \text{Var}[ C_\infty^{x,\pi^\star} ] = \frac{\sigma^2 x^2 \rho_\lambda}{2}
                .
        \end{equation}
        In other words, the efficient frontier of the range of mean and variance achievable by deterministic schedules, $\big\{ \big( \E[ C_\infty^{x,\pi} ], \text{Var}[ C_\infty^{x,\pi} ] \big) \big\}_{\pi \in \mathcal{D}(x)}$, is given by $\big\{ \big( \frac{\eta x^2}{4 \rho}, \frac{\sigma^2 x^2 \rho}{2} \big) \big\}_{\rho \in (0,\infty)}$ and is attained by exponential schedules.

        Let us now consider the achievable range of mean and standard deviation, $\big\{ \big( \E[ C_\infty^{x,\pi} ], \sqrt{ \text{Var}[ C_\infty^{x,\pi} ] } \big) \big\}_{\pi \in \mathcal{D}(x)}$.
        Observe that its efficient frontier is still characterized by $\big\{ \big( \frac{\eta x^2}{4 \rho}, \sqrt{ \frac{\sigma^2 x^2 \rho}{2} } \big) \big\}_{\rho \in (0,\infty)}$.
        Therefore, the optimal solution of the following mean-standard deviation optimization problem
        \begin{equation} \label{eq:ms-lagrangian}
                \minimize_{\pi \in \mathcal{D}(x)} ~~ \E[ C_\infty^{x,\pi} ] + \theta \sqrt{ \text{Var}[ C_\infty^{x,\pi} ] },
        \end{equation}
        is also given by an exponential schedule, for any given $\theta \in (0,\infty)$.

        Finally, observe that for any deterministic schedule $\pi \in \mathcal{D}(x)$ the
        resulting cost $C_\infty^{x,\pi}$ is normally distributed, thus
        $\cvar_q\left[ C_\infty^{x,\pi} \right]$ is minimized by an exponential schedule.
        Applying \eqref{eq:normcvar}, for any $\pi \in \mathcal{D}(x)$ and $q \in (0,1)$, we have
        \begin{equation}
                \cvar_q\left[ C_\infty^{x,\pi} \right] = \E\left[ C_\infty^{x,\pi} \right] + \frac{ \kappa(q) }{q} \sqrt{ \text{Var}\left[ C_\infty^{x,\pi} \right] },
        \end{equation}
        Using \eqref{eq:mv-exponential}, the optimal time constant can be determined as
        \begin{equation}
                \tau^\star
                        \defeq \argmin_\rho\left\{ \frac{\eta x^2}{4 \rho } + \frac{ \kappa(q) }{ q } \sqrt{ \frac{\sigma^2 x^2 \rho}{2} } \right\}
                        = \left( \frac{\eta x q }{ \sqrt{2} \sigma \kappa(q) } \right)^{\frac{2}{3}}.
        \end{equation}
        This concludes the proof.
\end{proof}

\begin{proof}[\proofnamest{Proof of Proposition~\ref{prop:vwap-schedule}}]
        With some calculation, it can be easily shown that a VWAP schedule $X_t = x \left( 1 - \frac{t}{T} \right)^+$ yields
        \begin{equation} \label{eq:mv-twap}
                \E[ C_\infty^{x,\pi} ]  = \frac{\eta x^2}{2 T}
                , \quad
                \text{Var}[ C_\infty^{x,\pi} ] =  \frac{\sigma^2 x^2 T}{3}
                .
        \end{equation}
        Therefore, the optimal execution horizon $T^\star$ is given by
        \begin{equation}
                T^\star
                        \defeq \argmin_T\left\{ \frac{\eta x^2}{2 T} + \frac{\kappa(q)}{q} \sqrt{\frac{\sigma^2 x^2 T}{3}} \right\}
                        = \left( \frac{\sqrt{3} \eta x q }{\sigma \kappa(q)} \right)^{\frac{2}{3}}.
        \end{equation}
\end{proof}

We state the following lemma that identifies the boundary values of $\scvar_q[ C_\infty^{x,\textsc{exp}} ]$, which is useful to characterize the optimal value function.

\begin{lemma} \label{lem:kappa-limit}
        The function $\kappa(q)$ given in \eqref{eq:kappa} satisfies
        \begin{equation}
                \sup_{q \in [0,1]} \kappa(q) < \infty
                , \quad
                \lim_{n \rightarrow \infty} \sqrt{n} \kappa(1/n) = \lim_{n \rightarrow \epsilon} \sqrt{n} \kappa(1-1/n) = 0.
        \end{equation}
\end{lemma}

\begin{proof}
        Recall that $\kappa(q) \defeq \phi\left( \Phi^{-1}(1-q) \right)$.
        Since $\phi(z) = \frac{1}{\sqrt{2\pi}} e^{-x^2/2} \leq \frac{1}{\sqrt{2 \pi}}$ for any $z \in \mathbb{R}$, we have $\sup_{q \in [0,1]} \kappa(q) \leq \frac{1}{\sqrt{2\pi}} < \infty$.
        Also note that $\kappa(q) = \kappa(1-q)$ since $\Phi^{-1}(q) = - \Phi^{-1}(1-q)$ and $\phi(z) = \phi(-z)$.
        Therefore, it suffices to show that $\lim_{n \rightarrow \infty} \sqrt{n} \kappa(1/n) = 0$.

        We have the following tail bounds of standard normal distribution \citep[Theorem~1.2.3]{Durrett}: for any $z > 0$,
        \begin{equation}
                \left( z^{-1} - z^{-3} \right) \phi(z) \leq 1-\Phi(z) \leq z^{-1} \phi(z).
        \end{equation}
        Define $z_n \defeq \Phi^{-1}(1-1/n) > 0$, and then we have
        \begin{equation}
                \frac{\phi( z_n )}{2z_n} \leq \frac{1}{n} \leq \frac{\phi( z_n )}{z_n},
        \end{equation}
        for large enough $n$ (such that $z_n^{-3} \leq \frac{1}{2} z_n^{-1}$) since $\lim_{n \rightarrow \infty} z_n = \infty$.
        Observe that $1/n \leq \phi(z_n)/z_n \leq \phi(z_n) = \exp(-z_n^2/2)/\sqrt{2\pi} \leq \exp( - z_n^2/2 )$ and thus $z_n \leq \sqrt{2 \log n}$.
        We further deduce that, since $\kappa(1/n) = \phi( \Phi^{-1}(1-1/n) ) = \phi(z_n)$,
        \begin{equation}
                \sqrt{n} \kappa(1/n)
                        = \sqrt{n} \phi(z_n)
                        = 2 \sqrt{n} z_n \times \frac{\phi(z_n)}{2 z_n}
                        \leq 2 \sqrt{n} z_n \times \frac{1}{n}
                        \leq \frac{ 2 z_n }{ \sqrt{n} }
                        \leq 2 \sqrt{ \frac{ 2 \log n }{n} }.
        \end{equation}
        Therefore, $\lim_{n \rightarrow \infty} \sqrt{n} \kappa(1/n) = 0$.
\end{proof}

\newpage
\section{Preliminary Characterizations of $\scvar$} \label{app:proof-problem}

We begin with a technical lemma:
\begin{lemma} \label{lem:weakly-star-compact}
        The risk envelope $\Qscr(q)$ is a non-empty, convex, and weak-* compact subset of $\Lscr^\infty$.
\end{lemma}

\begin{proof}
        It is non-empty because $Q(\omega) = q$ is always feasible.

        Consider $Q_1, Q_2 \in \Qscr(q)$ and $Q_\lambda \defeq \lambda Q_1 + (1-\lambda)Q_2$ for some $\lambda \in [0,1]$.
        Since $Q_1(\omega), Q_2(\omega) \in [0,1]$, we have $Q_\lambda(\omega) \in [0,1]$, and by the linearity of expectation, $\E[Q_\lambda] = \lambda \E[Q_1] + (1-\lambda)\E[Q_2] = q$.
        Therefore, $Q_\lambda \in \Qscr(q)$ and thus $\Qscr(q)$ is convex.

        Finally, note that $\Qscr(q)$ is a closed subset of the unit ball in $\Lscr^\infty(\Omega, \Fscr, \mathbb{P})$.
        Given that $\Lscr^\infty$ is the dual space of $\Lscr^1$, it is weak-* compact by Banach-Alaoglu theorem.
\end{proof}

\begin{proof}[\proofnamest{Proof of Proposition~\ref{prop:scvar-properties}}]

\subproof{Proof of claim \ref{it:scvar-cvar-relationship} and \ref{it:scvar-interpretation}.}
        Claim \ref{it:scvar-cvar-relationship} immediately follows from our definition of CVaR.
        Claim \ref{it:scvar-interpretation} follows from the following identity \citep[Theorem 6.2]{Shapiro09}:
        \begin{equation}
                \cvar_q[C]
                        = \E\left[ C \left| C \geq F_C^{-1}(1-q) \right. \right]
                        = \frac{ \E\left[ C \, \I{ C \geq F_C^{-1}(1-q) } \right] }{ \PR[ C \geq F_C^{-1}(1-q) ] }
                        = \frac{ \E\left[ C \, \I{ C \geq F_C^{-1}(1-q) } \right] }{ q }.
        \end{equation}

\subproof{Proof of claim \ref{it:scvar-boundary}.}
        When $q=0$, the risk envelope $\Qscr(q)$ has a single element $Q(\omega) = 0$, and hence, $\sup_{Q \in \Qscr(0)}\E[ CQ ] = 0$.
        When $q=1$, the risk envelope $\Qscr(q)$ also has a single element $Q(\omega) = 1$, and hence, $\sup_{Q \in \Qscr(0)}\E[ CQ ] = \E C$.

\subproof{Proof of claim \ref{it:scvar-bounds}.}
        For any $Q \in \Qscr(q)$, we have $\left| \E[ C Q ] \right| \leq \E\left[ \left| C Q \right| \right] \leq \E|C|$ since $|Q| \leq 1$ almost surely, and therefore, $\left| \scvar_q[C] \right| \leq \E|C|$.
        Furthermore, $\Qscr(q)$ contains $Q(\omega) = q$, and therefore, $\cvar_q[C] \geq \E[ q C ] = q \E C$.

\subproof{Proof of claim \ref{it:scvar-concavity}.}
        Consider $q_1, q_2 \in [0,1]$ and $q_\lambda \defeq \lambda q_1 + (1-\lambda) q_2$ for some $\lambda \in [0,1]$.
        Let $Q_1^* \in \argmax_{Q \in \Qscr(q_1)} \E[ CQ ]$ and $Q_2^* \in \argmax_{Q \in \Qscr(q_2)} \E[ CQ ]$ (the maximum is attained since $\Qscr(q)$ is weak-* compact).
        Let $Q_\lambda \defeq \lambda Q_1^* + (1-\lambda) Q_2^*$.
        Observe that $Q_\lambda \in [0,1]$ a.s. and $\E[ Q_\lambda ] = \lambda \E[ Q_1^* ] + (1-\lambda) \E[Q_2^*] = \lambda q_1 + (1-\lambda) q_2 = q_\lambda$.
        Therefore, $Q_\lambda \in \Qscr(q_\lambda)$.
        Trivially, $\E[ CQ_\lambda ] = \E[ \lambda CQ_1^* + (1-\lambda) C Q_2^* ] = \lambda \E[ C Q_1^* ] + (1-\lambda) \E[ C Q_2^* ]$, and thus
        \begin{align}
                \scvar_{q_\lambda}[ C ]
                        = \sup_{Q \in \Qscr(q_\lambda)} \E[ C Q ]
                        \geq \E[ C Q_\lambda ]
                        &= \lambda \E[ C Q_1^* ] + (1-\lambda) \E[ C Q_2^* ]
                        \\&= \lambda \scvar_{q_1}[ C ] + (1-\lambda) \scvar_{q_2}[ C ].
        \end{align}
        Continuity follows from concavity since the function is bounded on its domain from
        \ref{it:scvar-bounds}.
\end{proof}

\newpage

\section{Proofs for \S \ref{sec:cvar-dp}} \label{app:proof-cvar-dp}

From now on, we characterize $\scvar_q[ \, \cdot \, ]$ as a mapping from $\Lscr^2$ to $\mathbb{R}$.
This can be done without loss generality since we have $C_\infty^{x,\pi} \in \Lscr^2$ for any feasible policy $\pi \in \Pi(x)$.
This is to utilize the fact that $\Lscr^2$ is reflexive so that its weak-* topology coincides with its weak topology.

\begin{lemma} \label{lem:weakly-compact}
        The risk envelope $\Qscr(q)$ is a weakly compact subset of $\Lscr^2$.
\end{lemma}

\begin{proof}
        As stated in the proof of Lemma~\ref{lem:weakly-star-compact}, $\Qscr(q)$ is a closed subset of the unit ball in $\Lscr^\infty$, which is a subset of $\Lscr^2$.
        By Banach--Alaoglu theorem, it is weak-* compact in $\Lscr^2$ and hence weakly compact since $\Lscr^2$ is reflexive.
\end{proof}

\begin{proof}[\proofnamest{Proof of Proposition~\ref{prop:Q-properties}}]

\subproof{Proof of claim \ref{it:Q-martingale}.}
        Note that $(Q_t^{q,\gamma})_{t \geq 0}$ is a continuous local martingale since it is a stochastic integral of a progressively measurable process with respect to Brownian motion \citep[Theorem 33 in Chap. III]{Protter}.
        Since $Q_t^{q,\gamma} \in [0,1]$ for any $t \in [0, \infty)$ by the definition of $\Gamma(q)$, it is a bounded local martingale, which is indeed a martingale \citep[Thm. 51 in Chap. I]{Protter}.

\subproof{Proof of claim \ref{it:Q-limit}.}
        The limit $\lim_{t \rightarrow \infty} Q_t^{q,\gamma}$ exists due to martingale convergence theorem \citep[Theorem 10 in Chap. I]{Protter}.

\subproof{Proof of claim \ref{it:Q-absorption}.}
        Recall that $\big( Q_t^{q,\gamma} \big)_{t \geq 0}$ is a martingale taking values in $[0,1]$.
        Define a stopping time $\tau \defeq \inf_{t \geq 0}\left\{ Q_t^{q,\gamma} = 0 \right\}$.
        Then, for any $\tau' \geq \tau$, we have $\E[ Q_{\tau'}^{q,\gamma} | \Fscr_\tau ] = Q_{\tau}^{q,\gamma} = 0$ and therefore $Q_{\tau'}^{q,\gamma} = 0$ almost surely (otherwise, we would have $\E[ Q_{\tau'}^{q,\gamma} | \Fscr_\tau ] > 0 $).
        In words, once $\big( Q_t^{q,\gamma} \big)_{t \geq 0}$ hits zero, it never escapes
        thereafter, and the same argument holds for the other boundary.
\end{proof}

\subsection{Proof of Theorem~\ref{thm:cvar-minimax}} \label{app:cvar-minimax-proof}

Within this subsection, we characterize the trader's policy with its position process $(X_t)_{t
  \geq 0}$ rather than dealing with the liquidation rate process $(\pi_t)_{t \geq 0}$: the set of
admissible policies is represented as
\begin{equation}
  \Xscr(x) \defeq \left\{ X : \mathbb{T} \times \Omega \rightarrow \mathbb{R} \left|
      X \in \Prog,
      X_0 = x,
      \E\left[ \left( \int_{t=0}^\infty \dot{X}_t^2 dt \right)^2 \right] < \infty,
      \E\left[ \int_{t=0}^\infty X_t^2 dt \right] < \infty,
      \sup_{t \geq 0} \left| X_t \right| \leq M
    \right. \right\}
  ,
\end{equation}
where we require that each sample path $X \in \Xscr(x)$ is differentiable so that we can define
$\dot{X}_t \defeq dX_t/dt$.  Then, $\Pi(x) = \{ \pi = \dot{X} | X \in \Xscr(x) \}$.
Accordingly, we represent the loss process as
\begin{equation}
        C_t^X \defeq \int_{s=0}^t \frac{1}{2} \eta |\dot{X}_s|^2 ds - \int_{s=0}^t \sigma X_s dW_s,
\end{equation}
so that we have $C_t^X = C_t^{x,\pi}$ if $X \in \Xscr(x)$ and $\pi = \dot{X}$.

We aim to prove Theorem~\ref{thm:cvar-minimax} by utilizing Sion's minimax theorem:
\begin{lemma}[Sion's minimax theorem \citep{Sion58}] \label{lem:sion-minimax}
        Let $\Xscr$ be a convex subset of a linear topological space and $\Yscr$ a compact convex subset of a linear topological space.
        If $f$ is a real-valued function on $\Xscr \times \Yscr$ with
        $f(x,\cdot)$ upper semicontinuous and quasi-concave on $Y$, for each $x \in \Xscr$, and
        $f(\cdot,y)$ lower semicontinuous and quasi-convex on $X$, for each $y \in \Yscr$,
        then,
        \begin{equation}
                \inf_{x \in \Xscr} \max_{y \in \Yscr} f(x,y) = \max_{y \in \Yscr} \inf_{x \in \Xscr} f(x,y),
        \end{equation}
        where an optimal solution $y\in\Yscr$ must exist for each maximization.
\end{lemma}

\begin{lemma} \label{lem:Pi-convexity}
        $\Xscr(x)$ is a convex subset of a linear space endowed with a norm $\| \cdot \|_{\Xscr}$:
        \begin{equation}
                \| X \|_{\Xscr} \defeq \sqrt{ \E\left[ \int_{t=0}^\infty \dot{X}_t^2 dt \right] }
                                                + \sqrt{ \E\left[  \int_{t=0}^\infty X_t^2 dt \right] }.
        \end{equation}
\end{lemma}

\begin{proof}
        It can be easily verified that $\| \cdot \|_{\Xscr}$ is a valid norm (as a norm of a Sobolev space).
        Also note that $\| X \|_{\Xscr} < \infty$ for any $X \in \Xscr(x)$.
        Now consider $X^{(1)}, X^{(2)} \in \Xscr(x)$ and $X^{(\lambda)} \defeq \lambda X^{(1)} + (1-\lambda) X^{(2)}$ for some $\lambda \in [0,1]$.
        Observe that (i) $X^{(\lambda)}_0 = x$, (ii) $\sqrt{ \int_{t=0}^\infty \big( \dot{X}_t^{(\lambda)} \big)^2 dt } \leq \lambda \sqrt{ \int_{t=0}^\infty \big( \dot{X}_t^{(1)} \big)^2 dt } + (1-\lambda) \sqrt{ \int_{t=0}^\infty \big( \dot{X}_t^{(2)} \big)^2 dt } \in \Lscr^4$, and thus $\E\left[ \left( \int_{t=0}^\infty \big( \dot{X}_t^{(\lambda)} \big)^2 dt \right)^2 \right] < \infty$, (iii) similarly, $\E\left[ \int_{t=0}^\infty \big( X_t^{(\lambda)} \big)^2 dt \right] < \infty$, and (iv) $\sup_{t \geq 0} | X_t^{(\lambda)}| \leq \lambda \sup_{t \geq 0} | X_t^{(1)}| + (1-\lambda) \sup_{t \geq 0} | X_t^{(2)}| \leq M$.
        Therefore, $X^{(\lambda)} \in \Xscr(x)$, and hence $\Xscr(x)$ is a convex set.
\end{proof}

\begin{proof}[\proofnamest{Proof of Theorem~\ref{thm:cvar-minimax}}]
        By Lemma~\ref{lem:isomorphism}, it is equivalent to show that
        \begin{equation}
                \inf_{X \in \Xscr(x)} \max_{Q \in \Qscr(q)} \E\left[ C_\infty^{X} Q \right] = \max_{Q \in \Qscr(q)} \inf_{X \in \Xscr(x)} \E\left[ C_\infty^{X} Q \right].
        \end{equation}
        We aim to prove this by utilizing Sion's minimax theorem (Lemma~\ref{lem:sion-minimax}), so it suffices to check the conditions of Lemma~\ref{lem:sion-minimax}.
        We have shown that $\Xscr(x)$ is a convex subset of a linear space endowed with a norm $\| \cdot \|_{\Xscr}$ (Lemma~\ref{lem:Pi-convexity}), and $\Qscr(q)$ is a compact convex subset of $\Lscr^2( \Omega, \Fscr, \mathbb{P} )$ endowed with the weak topology (Lemma~\ref{lem:weakly-compact}).
        In the rest of the proof, we verify the followings:
        \begin{enumerate}[label=(\roman*)]
                \item \label{it:cvar-minimax-X-mapping} The mapping $X \mapsto \E[ C_\infty^{X} Q ]$ is convex and continuous in norm $\| \cdot \|_{\Xscr}$ for any given $Q \in \Qscr(q)$.
                \item \label{it:cvar-minimax-Q-mapping} The mapping $Q \mapsto \E[ C_\infty^{X} Q ]$ is concave and weakly continuous on $\Qscr(q)$ for any given $X \in \Xscr(x)$.
        \end{enumerate}

\subproof{Proof of claim \ref{it:cvar-minimax-X-mapping}.}
        The convexity immediately follows from the fact that the mapping $X \mapsto C_\infty^{X}$
        is quadratic and convex.
        We now focus on the continuity.

        Fix $X \in \Xscr(x)$ and consider a sequence $\big( X^{(n)} \in \Xscr(x) \big)_{n \in
          \mathbb{N}}$ such that $\lim_{n \rightarrow \infty} \| X - X^{(n)} \|_{\Xscr} =
        0$. That is,
        \begin{equation} \label{eq:cvar-minimax-eq1}
                \lim_{n \rightarrow \infty} \E\left[ \int_{t=0}^\infty \left( \dot{X}_t - \dot{X}_t^{(n)} \right)^2 dt \right] = 0
                , \quad
                \lim_{n \rightarrow \infty} \E\left[ \int_{t=0}^\infty \left( X_t - X_t^{(n)} \right)^2 dt \right] = 0.
        \end{equation}
        Let
        \begin{equation}
                A \defeq \sqrt{ \int_{t=0}^\infty \dot{X}_t^2 dt },
                \quad
                A_n \defeq \sqrt{ \int_{t=0}^\infty ( \dot{X}_t^{(n)} )^2 dt },
                \quad
                \Delta_n \defeq \sqrt{ \int_{t=0}^\infty \left( \dot{X}_t - \dot{X}_t^{(n)} \right)^2 dt }.
        \end{equation}
        We then have $A \in \Lscr^2$, $A_n \in \Lscr^2$, and $| A - A_n | \leq \Delta_n$ (triangle inequality) where $\Delta_n \stackrel{ \Lscr^2 }{ \rightarrow} 0$ due to the above condition \eqref{eq:cvar-minimax-eq1}.
        Further observe that
        \begin{align}
                \left| C_\infty^{X} Q - C_\infty^{X^{(n)}} Q \right|
                        &\leq \left| C_\infty^{X} - C_\infty^{X^{(n)}} \right|
                        \\&\leq \frac{\eta}{2}  \left| \int_{t=0}^\infty \dot{X}_t^2 dt - \int_{t=0}^\infty ( \dot{X}_t^{(n)} )^2 dt \right|
                                + \sigma \left| \int_{t=0}^\infty X_t dW_t - \int_{t=0}^\infty X_t^{(n)} dW_t \right|
                        \\&= \frac{\eta}{2}  \left| A^2 - A_n^2  \right|
                                + \sigma \left| \int_{t=0}^\infty \left( X_t - X_t^{(n)} \right) dW_t \right| ,
                                \label{eq:cvar-minimiax-eq2}
        \end{align}
        where the first inequality holds since $|Q| \leq 1$.
        For the first term of \eqref{eq:cvar-minimiax-eq2}, we have
        \begin{align}
                \E\left[ \left| A_n^2 - A^2 \right| \right]
                        &= \E\left[ \left|  (A_n - A)^2 + 2 A ( A_n -  A) \right| \right]
                        \\&\leq \E\left[ (A_n - A)^2 \right] + 2 \E\left[ | A | \cdot | A_n - A | \right]
                        \\&\leq \E\left[ (A_n - A)^2 \right] + 2 \sqrt{ \E\left[ A^2 \right] }
          \sqrt{ \E\left[ | A_n - A |^2 \right] } \label{eq:cvar-cs}
                        \\&= \E\left[ \Delta_n^2 \right] + 2 \sqrt{ \E\left[ A^2 \right] } \sqrt{ \E\left[ \Delta_n^2 \right] },
        \end{align}
        where for \eqref{eq:cvar-cs} we apply the Cauchy-Schwartz inequality.
        Therefore, $\E\left[ \left| A_n^2 - A^2 \right| \right] \rightarrow 0$ since $\Delta_n
        \stackrel{ \Lscr^2 }{ \rightarrow} 0$, as $n \rightarrow \infty$.
        For the second term of \eqref{eq:cvar-minimiax-eq2}, using the It\^{o} isometry,
        \begin{equation}
                \E\left[ \left( \int_{t=0}^\infty \left( X_t - X_t^{(n)} \right) dW_t \right)^2  \right]
                        = \E\left[ \int_{t=0}^\infty \left( X_t - X_t^{(n)} \right)^2 dt  \right],
        \end{equation}
        which vanishes as $n \rightarrow \infty$ due to condition \eqref{eq:cvar-minimax-eq1}.
        Therefore,
        \begin{equation}
                \lim_{n \rightarrow \infty} \E\left[ \left| C_\infty^{X} Q - C_\infty^{X^{(n)}} Q \right| \right] = 0,
        \end{equation}
        which implies that $\lim_{n \rightarrow \infty} \E\left[ C_\infty^{X^{(n)}} Q \right] = \E\left[ C_\infty^{X} Q \right]$.
        This concludes the proof.

        \subproof{Proof of claim \ref{it:cvar-minimax-Q-mapping}.}  The concavity immediately
        follows from the fact that the mapping $Q \mapsto C_\infty^X Q$ is linear.  Fix
        $Q \in \Qscr(q)$, consider a sequence $\big( Q_n \in \Qscr(q) \big)_{n \in \mathbb{N}}$
        converging to $Q$ in $\Lscr^2$-weak topology: i.e.,
        $\lim_{n \rightarrow \infty} \E[ Z Q_n ] = \E[ Z Q ]$ for any $Z \in \Lscr^2$.  Since
        $C_\infty^X \in \Lscr^2$ for any $X \in \Xscr(x)$, the continuity of the mapping
        $Q \mapsto \E[ C_\infty^X Q ]$ immediately follows.
\end{proof}

\newpage
\subsection{Preliminary Characterizations of Value Function}

We begin with a preliminary characterization of the value function \eqref{eq:V} that will be used
in later proofs in this section.

\begin{proposition} \label{prop:V-properties}
        The value function \newedit{$V: \mathbb{X} \times [0,1] \rightarrow \mathbb{R}$} satisfies the followings:
        \begin{enumerate}[label=(\roman*)]
                \item \label{it:V-bounds} $0 \leq V(x,q) \leq \sigma^{\frac{2}{3}} \eta^{\frac{1}{3}} \times |x|^{\frac{4}{3}} \times q^{\frac{1}{3}} \big( \kappa(q) \big)^{\frac{2}{3}}$.
                \item \label{it:V-convexity} $V(x,q)$ is convex in $x$ on \newedit{$\mathbb{X}$} for each $q \in [0,1]$.
                \item \label{it:V-concavity} $V(x,q)$ is concave in $q$ on $[0,1]$ for each \newedit{$x \in \mathbb{X}$}.
                \item \label{it:V-boundary} $V(x,0) = V(x,1) = V(0,q) = 0$ for each \newedit{$x \in \mathbb{X}$} and $q \in [0,1]$.
        \end{enumerate}
\end{proposition}

\begin{proof}
\subproof{Proof of claim \ref{it:V-bounds}.}
        Note that $\E[ C_\infty^{x,\pi} ] \geq 0$ for any $\pi \in \Pi(x)$.
        By Proposition~\ref{prop:scvar-properties}\ref{it:scvar-bounds}, we have $\scvar_q[ C_\infty^{x,\pi} ] \geq 0$ for any $\pi \in \Pi(x)$, and hence $V(x,q) \geq 0$.
        The upper bound follows from Proposition~\ref{prop:exp-schedule} given that $\frac{3}{2^{5/3}} < 1$.

\subproof{Proof of claim \ref{it:V-convexity}.}
        Fix $q \in [0,1]$ and define $\varphi_\gamma : \mathbb{R} \mapsto \mathbb{R}$ as follows:
        \begin{equation}
                \varphi_\gamma(x) \defeq \inf_{\pi \in \Pi(x)} J( \pi, \gamma; x, q ),
        \end{equation}
        so that we have $V(x,q) = \max_{\gamma \in \Gamma(q)} \varphi_\gamma(x)$.
        Since the pointwise maximum of a set of convex functions is convex, it suffices to show that $\varphi_\gamma(\cdot)$ is convex for each $\gamma \in \Gamma(q)$.

        Fix $\gamma \in \Gamma(q)$ and consider any \newedit{$x_1, x_2 \in \mathbb{X}$}.
        For any $\epsilon > 0$, by definition of the infimum, there exist $\pi^{1,\epsilon} \in \Pi(x_1)$ and $\pi^{2,\epsilon} \in \Pi(x_2)$ such that
        \begin{equation}
                J(\pi^{1,\epsilon}, \gamma; x_1, q) \leq \varphi_\gamma(x_1) + \epsilon,
                \quad
                J(\pi^{2,\epsilon}, \gamma; x_2, q) \leq \varphi_\gamma(x_2) + \epsilon.
        \end{equation}
        Given some $\lambda \in [0,1]$, let $x_\lambda \defeq \lambda x_1 + (1-\lambda) x_2$ and $\pi^{\lambda, \epsilon} \defeq \lambda \pi^{1,\epsilon} + (1-\lambda) \pi^{2,\epsilon}$. Note that
        \begin{equation}
                X_t^{x_\lambda, \pi^{\lambda,\epsilon}} = x_\lambda - \int_{s=0}^t \pi_s^{\lambda, \epsilon} ds
                        =  \lambda X_t^{x_1, \pi^{1,\epsilon}} + (1-\lambda) X_t^{x_2, \pi^{2,\epsilon}}.
        \end{equation}
        Similarly to the proof of Lemma \ref{lem:Pi-convexity}, it can be shown that $\E\left[ \left( \int_{t=0}^\infty \big| \pi_t^{\lambda,\epsilon} \big|^2 dt \right)^2 \right] < \infty$ and $\E\left[ \int_{t=0}^\infty \big| X_t^{\lambda, \pi^{\lambda,\epsilon}} \big|^2 dt  \right] < \infty$, and therefore $\pi^{\lambda, \epsilon} \in \Pi(x_\lambda)$.
        Consequently,
        \begin{equation}
                C_\infty^{x_\lambda, \pi^{\lambda, \epsilon}}
                        = \int_{s=0}^t \frac{1}{2} \eta \left| \pi_s^{\lambda,\epsilon} \right|^2 ds
                                - \int_{s=0}^t \sigma X_s^{x_\lambda, \pi^{\lambda,\epsilon}} dW_s
                        \leq \lambda C_\infty^{x_1, \pi^{1, \epsilon}} + (1-\lambda) C_\infty^{x_2, \pi^{2, \epsilon}},
        \end{equation}
        and therefore,
        \begin{align}
                J( \pi^{\lambda,\epsilon}, \gamma; x_\lambda, q )
                        &= \E\left[ C_\infty^{x_\lambda, \pi^{\lambda, \epsilon}} Q_\infty^{q, \gamma} \right]
                        \leq \lambda \E\left[ C_\infty^{x_1, \pi^{1, \epsilon}} Q_\infty^{q, \gamma} \right]
                                + (1-\lambda) \E\left[ C_\infty^{x_2, \pi^{2, \epsilon}} Q_\infty^{q, \gamma} \right]
                        \\&= \lambda J( \pi^{1,\epsilon}, \gamma; x_1,q ) + (1-\lambda) J( \pi^{2,\epsilon}, \gamma; x_2,q )
                        \\&\leq \lambda \varphi_\gamma(x_1) + (1-\lambda) \varphi_\gamma(x_2) + \epsilon.
        \end{align}
        As a result, we have
        \begin{equation}
                \varphi_\gamma( \lambda x_1 + (1-\lambda) x_2 )
                        \leq J( \pi^{\lambda,\epsilon}, \gamma; x_\lambda,q )
                        \leq \lambda \varphi_\gamma(x_1) + (1-\lambda) \varphi_\gamma(x_2) + \epsilon.
        \end{equation}
        Since $x_1, x_2, \lambda, \epsilon$ were arbitrarily chosen, $\varphi_\gamma(\cdot)$ is convex on $\mathbb{R}$.


        \subproof{Proof of claim \ref{it:V-concavity}.}
        Fix \newedit{$x \in \mathbb{X}$} and define $\varphi_x : [0,1] \mapsto \mathbb{R}$ as follows:
        \begin{equation}
                \varphi_\pi(q) \defeq \sup_{\gamma \in \Gamma(q)} J( \pi, \gamma; x, q ),
        \end{equation}
        so that we have $V(x,q) = \inf_{\pi \in \Pi(x)} \varphi_\pi(q)$.
        Since the point-wise infimum of a set of concave functions is concave, it suffices to show that $\varphi_\pi(\cdot)$ is concave on $[0,1]$ for each $\pi \in \Pi(x)$.

        Fix $\pi \in \Pi(x)$ and consider any $q_1, q_2 \in [0,1]$.
          By Lemma~\ref{lem:isomorphism} and Lemma~\ref{lem:weakly-compact}, there exist $\gamma^1 \in \Gamma(q_1)$ and $\gamma^2 \in \Gamma(q_2)$ such that
        \begin{equation}
                J( \pi, \gamma^1; x, q_1 ) = \varphi_\pi(q_1)
                , \quad
                J( \pi, \gamma^2; x, q_2 ) = \varphi_\pi(q_2).
        \end{equation}
        Given some $\lambda \in [0,1]$, let $q_\lambda \defeq \lambda q_1 + (1-\lambda) q_2$ and $\gamma^\lambda \defeq \lambda \gamma^1 + (1-\lambda) \gamma^2$ (i.e., $\gamma_t^\lambda(\omega) = \lambda \gamma_t^1(\omega) + (1-\lambda) \gamma_t^2(\omega)$ for $\forall t, \omega$).
        Note that
        \begin{equation}
                Q_t^{q_\lambda, \gamma^\lambda}
                        = q_\lambda + \int_{s=0}^t \gamma_s^\lambda dW_s
                        = \lambda Q_t^{q_1, \gamma^1} + (1-\lambda) Q_t^{q_2, \gamma^2},
        \end{equation}
        and hence $Q_t^{q_\lambda, \gamma^\lambda} \in [0,1]$ almost surely for any $t$, i.e., $\gamma^\lambda \in \Gamma(q_\lambda)$.
        Therefore,
        \begin{align}
                J( \pi, \gamma^\lambda; x, q_\lambda )
                        &= \E\left[ C_\infty^{x, \pi} Q_\infty^{q_\lambda, \gamma^\lambda} \right]
                        = \lambda \E\left[ C_\infty^{x, \pi} Q_\infty^{q_1, \gamma^1} \right]
                                + (1-\lambda) \E\left[ C_\infty^{x, \pi} Q_\infty^{q_2, \gamma^2} \right]
                        \\&= \lambda J( \pi, \gamma^1; x, q_1 ) + (1-\lambda) J( \pi, \gamma^2; x, q_2 )
                        \\&= \lambda \varphi_\pi(q_1) + (1-\lambda) \varphi_\pi(q_2).
        \end{align}
        As a result, we have
        \begin{equation}
                \varphi_\pi( \lambda q_1 + (1-\lambda) q_2 )
                        \geq J( \pi, \gamma^\lambda; x, q_\lambda )
                        = \lambda \varphi_\pi(q_1) + (1-\lambda) \varphi_\pi(q_2).
        \end{equation}
        Since $q_1, q_2, \lambda$ were arbitrarily chosen, $\varphi_\pi(\cdot)$ is concave on $[0,1]$.

\subproof{Proof of claim \ref{it:V-boundary}.}
        The claim immediately follows from claim \ref{it:V-bounds}.
\end{proof}


\newpage
\subsection{Proof of CVaR Dynamic Programming Principle} \label{ssec:proof-cvar-dp}

We first state a proposition that is useful for proving upper-semicontinuity of a mapping with respect to the weak topology.

\begin{lemma}[\citet{Barbu}, Proposition 2.10. Restated for upper-semicontinuity] \label{lem:concave-upper-semicontinuity}
        Let $\Yscr$ be a locally convex space.
        A proper concave function $f : \Yscr \rightarrow [-\infty, \infty)$ is upper-semicontinuous on $\Yscr$ if and only if it is upper-semicontinuous with respect to the weak topology on $\Yscr$.
\end{lemma}

We next prove a minimax theorem that includes the value function and a stopping time, which extends Theorem~\ref{thm:cvar-minimax}.

\begin{proposition} \label{prop:cvar-dp-minimax}
        For any stopping time $\tau:\Omega \rightarrow \mathbb{T}$ \newedit{with $\E[ \tau ] < \infty$}, we have
        \begin{equation}
                \inf_{\pi \in \Pi(x)} \sup_{\gamma \in \Gamma(q)} \E\left[ C_\tau^{x,\pi} Q_\tau^{q,\gamma} + V( X_\tau^{x,\pi}, Q_\tau^{q,\gamma} ) \right]
                        = \sup_{\gamma \in \Gamma(q)}  \inf_{\pi \in \Pi(x)} \E\left[ C_\tau^{x,\pi} Q_\tau^{q,\gamma} + V( X_\tau^{x,\pi}, Q_\tau^{q,\gamma} ) \right].
                      \end{equation}
       Moreover, an optimal solution $\gamma \in \Gamma(q)$ must exist for each maximization, i.e.,
        \begin{equation}
                \inf_{\pi \in \Pi(x)} \max_{\gamma \in \Gamma(q)} \E\left[ C_\tau^{x,\pi} Q_\tau^{q,\gamma} + V( X_\tau^{x,\pi}, Q_\tau^{q,\gamma} ) \right]
                        = \max_{\gamma \in \Gamma(q)}  \inf_{\pi \in \Pi(x)} \E\left[ C_\tau^{x,\pi} Q_\tau^{q,\gamma} + V( X_\tau^{x,\pi}, Q_\tau^{q,\gamma} ) \right].
       \end{equation}
\end{proposition}
\begin{proof}
        Let us first define
        \begin{equation}
                \Qscr_\tau(q) \defeq \left\{ \E( Q | \Fscr_\tau ) : Q \in \Qscr(q) \right\}.
        \end{equation}
        As analogous to Lemma~\ref{lem:isomorphism}, it can be easily shown that $Q_\tau(q) = \left\{ Q_\tau^{q,\gamma} | \gamma \in \Gamma(q) \right\}$.
        With the notation introduced in the proof of Theorem~\ref{thm:cvar-minimax} (\S \ref{app:cvar-minimax-proof}), it suffices to show that
        \begin{equation}
                \inf_{X \in \Xscr(x)} \max_{Q \in \Qscr_\tau(q)} \E\left[ C_\tau^{X} Q + V( X_\tau, Q ) \right]
                        = \max_{Q \in \Qscr_\tau(q)} \inf_{X \in \Xscr(x)} \E\left[ C_\tau^{X} Q + V( X_\tau, Q ) \right],
        \end{equation}
        for which we need to verify the followings:
        \begin{enumerate}[label=(\roman*)]
                \item \label{it:cvar-dp-minimax-X-mapping} The mapping $X \mapsto \E\left[ C_\tau^X Q + V( X_\tau, Q ) \right]$ is convex and continuous on $\Xscr(x)$ with respect to the norm $\| \cdot \|_{\Xscr}$ for any given $Q \in \Qscr_\tau(q)$.
                \item \label{it:cvar-dp-minimax-weak-compact} $\Qscr_\tau(q)$ is a non-empty, convex, and weakly compact subset of $\Lscr^2( \Omega, \Fscr_\tau, \mathbb{P} )$.
                \item \label{it:cvar-dp-minimax-Q-mapping1} The mapping $Q \mapsto \E\left[ V( X_\tau, Q ) \right]$ is continuous on $\Qscr_\tau(q)$ endowed with $\Lscr^2$-norm\footnote{
                                \newedit{Note that what we ultimately want to show is that the mapping $Q \mapsto \E\left[ V( X_\tau, Q ) \right]$ is upper-semicontinuous on $\Qscr_\tau(q)$ endowed with $\Lscr^2$-weak topology.
                                We first show its continuity with respect to $\Lscr^2$-norm topology and then utilize Lemma~\ref{lem:concave-upper-semicontinuity} to achieve this goal.}
                } for any given $X \in \Xscr(x)$.
                \item \label{it:cvar-dp-minimax-Q-mapping2} The mapping $Q \mapsto \E\left[ C_\tau^X Q + V( X_\tau, Q ) \right]$ is concave and upper-semicontinuous on $\Qscr_\tau(q)$ endowed with $\Lscr^2$-weak topology for any given $X \in \Xscr(x)$.
        \end{enumerate}
        Together with the convexity of $\Xscr(x)$ (Lemma~\ref{lem:Pi-convexity}), we obtain the desired claim by utilizing Sion's minimax theorem (Lemma~\ref{lem:sion-minimax}).

        \subproof{Proof of claim \ref{it:cvar-dp-minimax-X-mapping}.}
        The convexity of the mapping $X \mapsto \E\left[ C_\tau^X Q + V( X_\tau, Q ) \right]$
        immediately follows from the arguments in the proof of Theorem~\ref{thm:cvar-minimax}, and
        the convexity of $V(\cdot, Q)$, which was shown in Proposition~\ref{prop:V-properties}\ref{it:V-convexity}.
        Now fix $X \in \Xscr(x)$ and consider a sequence $( X^{(n)} \in \Xscr(x) )_{n \in \mathbb{N}}$ such that $\| X - X^{(n)} \|_{\Xscr} \rightarrow 0$ as $n \rightarrow \infty$.
        Following the arguments in the proof of Theorem~\ref{thm:cvar-minimax}, we can show that
        $C_\tau^{X^{(n)}} Q \stackrel{ \Lscr^1 }{\rightarrow} C_\tau^X Q$ as $n \rightarrow
        \infty$.
        We also have that
        \begin{align}
                        \E\left[ \left| X_\tau^{(n)} - X_\tau  \right| \right]
                        &= \E\left[ \left| \int_{t=0}^\tau ( \dot{X}_t^{(n)} - \dot{X}_t ) dt \right| \right]
                        \\&\leq \E\left[  \int_{t=0}^\tau \left| \dot{X}_t^{(n)} - \dot{X}_t \right| dt \right]
                        \\&\leq \left( \E\left[ \int_{t=0}^\tau \left| \dot{X}_t^{(n)} - \dot{X}_t \right|^2 dt \right] \right)^{\frac{1}{2}} \times \left( \E\left[ \int_{t=0}^\tau 1^2 dt \right] \right)^{\frac{1}{2}}
                        \\&\leq \| X^{(n)} - X \|_\mathcal{X} \times \sqrt{ \E[\tau] },
        \end{align}
        where the second inequality holds due to H\"{o}lder's inequality.
        Since $\E[ \tau ] < \infty$ and $\| X^{(n)} - X \|_\mathcal{X} \searrow 0$ as $n \rightarrow \infty$, we deduce that $X_\tau^{(n)} \stackrel{ \Lscr^1 }{\rightarrow} X_\tau$.
        By the continuous mapping theorem, we further obtain that $V( X_\tau^{(n)}, Q) \stackrel{ p }{ \rightarrow } V( X_\tau, Q )$ as $n \rightarrow \infty$.
        Given that $\sup_{t \geq 0} |X_t^{(n)}| \leq M$ for all $n \in \mathbb{N}$, the sequence $\big( | V( X_\tau^{(n)}, Q ) | \big)_{n \in \mathbb{N}}$ is uniformly bounded by $\sup_{q' \in [0,1]} \big\{ \sigma^{\frac{2}{3}} \eta^{\frac{1}{3}} \times M^{\frac{4}{3}} \times {q'}^{\frac{1}{3}} \kappa(q')^{\frac{2}{3}} \big\} < \infty$, and hence uniformly integrable.
        Therefore, $V( X_\tau^{(n)}, Q ) \stackrel{ \Lscr^1 }{ \rightarrow } V( X_\tau, Q )$.
        Combining these results, we obtain $\lim_{n \rightarrow \infty} \E\left[ C_\tau^{X^{(n)}} Q + V( X_\tau^{(n)}, Q ) \right] = \E\left[ C_\tau^X Q + V(X_\tau, Q) \right]$, which concludes the proof.

\subproof{Proof of claim \ref{it:cvar-dp-minimax-weak-compact}.}
        The proof is identical to that of Lemma~\ref{lem:weakly-star-compact} and \ref{lem:weakly-compact} except that $\Lscr^\infty(\Omega, \Fscr, \mathbb{P})$ and $\Lscr^2(\Omega, \Fscr, \mathbb{P})$ are replaced with  $\Lscr^\infty(\Omega, \Fscr_\tau, \mathbb{P})$ and $\Lscr^2(\Omega, \Fscr_\tau, \mathbb{P})$, respectively.

\subproof{Proof of claim \ref{it:cvar-dp-minimax-Q-mapping1}.}
        Fix $Q \in \Qscr_\tau(q)$ and consider a sequence $( Q^{(n)} \in \Qscr_\tau(q) )_{n \in \mathbb{N}}$ such that $Q^{(n)} \stackrel{ \Lscr^2 }{\rightarrow} Q$ as $n \rightarrow \infty$.
        By the continuous mapping theorem, we have $V( X_\tau, Q^{(n)} ) \stackrel{p}{\rightarrow} V(X_\tau, Q)$.
        Since the sequence $\big( | V( X_\tau, Q^{(n)} ) | \big)_{n \in \mathbb{N}}$ is uniformly integrable (uniformly bounded by $\sup_{q' \in [0,1]} \big\{ \sigma^{\frac{2}{3}} \eta^{\frac{1}{3}} \times |M|^{\frac{4}{3}} \times {q'}^{\frac{1}{3}} \kappa(q')^{\frac{2}{3}} \big\} < \infty$), we further have $V( X_\tau, Q^{(n)} ) \stackrel{ \Lscr^1 }{\rightarrow} V(X_\tau, Q)$, which concludes the proof.

\subproof{Proof of claim \ref{it:cvar-dp-minimax-Q-mapping2}.}
        The concavity of the mapping $Q \mapsto \E\left[ C_\tau^X Q + V( X_\tau, Q ) \right]$ immediately follows from the concavity of $V(X_\tau, \cdot )$, which was shown in Proposition~\ref{prop:V-properties}\ref{it:V-concavity}.
        On the other hand, by combining the result of claim \ref{it:cvar-dp-minimax-Q-mapping1}
        with Lemma~\ref{lem:concave-upper-semicontinuity}, we deduce that the mapping $Q
        \mapsto \E\left[ V( X_\tau, Q ) \right]$ is upper-semicontinuous on $\Qscr_\tau(q)$
        endowed with $\Lscr^2$-weak topology.
        Therefore, it suffices to show that the mapping $Q \mapsto \E\left[ C_\tau^X Q \right]$ is upper-semicontinuous with respect to $\Lscr^2$-weak topology.

        Fix $Q \in \Qscr_\tau(q)$ and consider a sequence $( Q^{(n)} \in \Qscr_\tau(q) )_{n \in \mathbb{N}}$ such that $\E[ Z Q^{(n)} ] \rightarrow \E[ Z Q ]$ as $n \rightarrow \infty$ for any $Z \in \Lscr^2$. Since $C_\tau^X \in \Lscr^2$, we have $\lim_{n \rightarrow \infty} \E[ C_\tau^X Q^{(n)} ] = \E[ C_\tau^X Q ]$, which implies the continuity of the mapping $Q \mapsto \E\left[ C_\tau^X Q \right]$.
\end{proof}

Before proving Theorem~\ref{thm:cvar-dp}, we introduce additional notation to describe the Markovian structure of problem.
        We denote by $(X_s^{t,x,\pi})_{s \geq t}$ the trader's position process under control $\pi$ that begins from the value $x$ at time $t$, i.e.,
        \begin{equation}
                X_s^{t,x,\pi} \defeq x - \int_{u=t}^s \pi_u du.
        \end{equation}
        We define the adversary's martingale process $(Q_s^{t,q,\gamma})_{s \geq t}$ and the loss process $(C_s^{t,x,\pi})_{s \geq t}$ analogously,
        \begin{equation}
                Q_s^{t,q,\gamma} \defeq q + \int_{u=t}^s \gamma_u dW_u
                , \quad
                C_s^{t,x,\pi} \defeq \int_{u=t}^s \frac{1}{2} \eta \pi_u^2 du
                                - \int_{u=t}^s \sigma X_u^{t,x,\pi} dW_u.
        \end{equation}
        With this notation, we can describe the aforementioned processes in a recursive way,
        \begin{equation}
                X_s^{0,x,\pi} = X_s^{t,X_t^{0,x,\pi},\pi}
                , \quad
                Q_s^{0,q,\gamma} = Q_s^{t,Q_t^{0,q,\gamma},\gamma}
                , \quad
                C_s^{0,q,\gamma} = C_t^{0,x,\pi} + C_s^{t,X_t^{0,x,\pi},\pi}.
        \end{equation}
        The policy spaces are defined analogously as well,
        \begin{equation}
                \Pi_t(x) \defeq \left\{ \pi \in \Pi(x) \left|
                        ~ \pi_s = 0, \forall s < t
                        \right. \right\}
                , \quad
                \Gamma_t(q) \defeq \left\{ \gamma \in \Gamma(q) \left|
                        ~ \gamma_s = 0, \forall s < t
                \right. \right\}.
        \end{equation}

We prove Theorem~\ref{thm:cvar-dp} by utilizing Proposition~\ref{prop:cvar-dp-minimax} and Proposition~\ref{prop:time-decomposition}.

\begin{proof}[\proofnamest{Proof of Theorem~\ref{thm:cvar-dp}}]
        Define
        \begin{equation}
                U(x,q) \defeq \inf_{\pi \in \Pi(x)} \sup_{\gamma \in \Gamma(q)} \E\left[ C_\tau^{0,x,\pi} Q_\tau^{0,q,\gamma} + V\left( X_\tau^{0,x,\pi}, Q_\tau^{0,q,\gamma} \right) \right].
        \end{equation}
        By Lemma~\ref{lem:isomorphism} and Lemma~\ref{lem:weakly-compact}, we can equivalently write
        \begin{equation}
                U(x,q) = \inf_{\pi \in \Pi(x)} \max_{\gamma \in \Gamma(q)} \E\left[ C_\tau^{0,x,\pi} Q_\tau^{0,q,\gamma} + V\left( X_\tau^{0,x,\pi}, Q_\tau^{0,q,\gamma} \right) \right].
        \end{equation}
        We aim to prove $V(x,q) = U(x,q)$.

\subproof{Proof of ``$V(x,q) \leq U(x,q)$''.}
        Fix $\epsilon > 0$.
        By definition of $U(x,q)$, there exists $\pi^\circ \in \Pi(x)$ such that
        \begin{align}
                U(x,q) + \epsilon
                        &\geq \max_{\gamma \in \Gamma(q)} \E\left[ C_\tau^{0,x,\pi^\circ} Q_\tau^{0,q,\gamma} + V\left( X_\tau^{0,x,\pi^\circ}, Q_\tau^{0,q,\gamma} \right) \right].
        \end{align}
        On the other hand, we have that for any $t \geq 0$, $\hat{x} \in \mathbb{X}$, and $\hat{q} \in [0,1]$,
        \begin{equation}
                V(\hat{x}, \hat{q}) = \inf_{\hat{\pi} \in \Pi_t(\hat{x})} \max_{\hat{\gamma} \in \Gamma_t(\hat{q})} \E\left[ C_\infty^{t,\hat{x},\hat{\pi}} Q_\infty^{t,\hat{q},\hat{\gamma}} \right],
        \end{equation}
        by time homogeneity of the problem.
        Consequently, for each $\omega \in \Omega$ and $\gamma \in \Gamma(q)$, there exists $\hat{\pi}^{\omega,\gamma} \in \Pi_{\tau(\omega)}( X_\tau^{0,x,\pi^\circ}(\omega) )$ such that
        \begin{align}
                V\left( X_\tau^{0,x,\pi^\circ}(\omega), Q_\tau^{0,q,\gamma}(\omega) \right) + \epsilon
                        &\geq \max_{\hat{\gamma} \in \Gamma_{\tau(\omega)}( Q_\tau^{0,q,\gamma}(\omega))}
                                \E\left[ \left.
                                        C_\infty^{\tau, X_\tau^{0,x,\pi^\circ}, \hat{\pi}^{\omega,\gamma}}Q_\infty^{\tau,Q_\tau^{0,q,\gamma},\hat{\gamma}}
                                \right| \Fscr_\tau \right](\omega)
                        \\&\geq
                                \E\left[ \left.
                                        C_\infty^{\tau, X_\tau^{0,x,\pi^\circ}, \hat{\pi}^{\omega,\gamma}}Q_\infty^{\tau,Q_\tau^{0,q,\gamma},\gamma}
                                \right| \Fscr_\tau \right](\omega),
        \end{align}
        where the second inequality follows from the fact that the given $\gamma$ may not be optimal for the period $t \geq \tau$.

        For each $\gamma \in \Gamma(q)$, consider a trader's policy $\pi^\gamma: \mathbb{T} \times \Omega \rightarrow \mathbb{R}$ constructed as follows:
        \begin{equation}
                \pi_t^\gamma(\omega) \defeq \left\{ \begin{array}{ll}
                                \pi_t^\circ(\omega) & \text{if } t < \tau(\omega), \\
                                \hat{\pi}_t^{\omega,\gamma}(\omega) & \text{if } t \geq \tau(\omega),
                        \end{array} \right.
        \end{equation}
        i.e., it implements an $\epsilon$-optimal solution to the Bellman equation before $\tau$, and then implements another $\epsilon$-optimal solution for the remaining horizon.
        We observe that $\pi^\gamma \in \Pi(x)$ since $\pi_t^\circ(\omega) \in \Pi(x)$ and $\hat{\pi}_t^{\omega,\gamma} \in \Pi(x)$.
        Combining these results, we obtain
        \begin{align}
                U(x,q)
                        &\geq \max_{\gamma \in \Gamma(q)} \E\left[ C_\tau^{0,x,\pi^\circ} Q_\tau^{0,q,\gamma} + V\left( X_\tau^{0,x,\pi^\circ}, Q_\tau^{0,q,\gamma} \right) \right] - \epsilon
                        \\&\geq \max_{\gamma \in \Gamma(q)} \E\left[
                                C_\tau^{0,x,\pi^\circ} Q_\tau^{0,q,\gamma}
                                + \E\left[ \left. C_\infty^{\tau, X_\tau^{0,x,\pi^\circ}, \hat{\pi}^{\omega,\gamma}}Q_\infty^{\tau,Q_\tau^{0,q,\gamma},\gamma} \right| \Fscr_\tau \right]
                        \right] - 2 \epsilon
                        \\&= \max_{\gamma \in \Gamma(q)} \E\left[
                                        C_\tau^{0,x,\pi^\gamma} Q_\tau^{0,q,\gamma}
                                        + C_\infty^{\tau, X_\tau^{0,x,\pi^\gamma}, \pi^\gamma}Q_\infty^{\tau,Q_\tau^{0,q,\gamma},\gamma}
                        \right] - 2 \epsilon
                        \\&= \max_{\gamma \in \Gamma(q)} \E\left[
                                        C_\infty^{0, x, \pi^\gamma}Q_\infty^{0,q,\gamma}
                         \right] - 2 \epsilon
                         \\&\geq \max_{\gamma \in \Gamma(q)} \inf_{\pi \in \Pi(x)}  \E\left[
                                        C_\infty^{0, x, \pi}Q_\infty^{0,q,\gamma}
                         \right] - 2 \epsilon
                         \\&= V(x,q) - 2 \epsilon.
        \end{align}
        Since the choice of $\epsilon$ was arbitrary, we deduce that $U(x,q) \geq V(x,q)$.

\subproof{Proof of ``$V(x,q) \geq U(x,q)$''.}
        By minimax equality result for $U(x,q)$ (Proposition~\ref{prop:cvar-dp-minimax}), there
        exists $\gamma^\circ \in \Gamma(q)$ such that
        \begin{align}
                U(x,q)
                        &= \inf_{\pi \in \Pi(x)} \E\left[ C_\tau^{0,x,\pi} Q_\tau^{0,q,\gamma^\circ} + V\left( X_\tau^{0,x,\pi}, Q_\tau^{0,q,\gamma^\circ} \right) \right].
        \end{align}
        By minimax equality result for $V(x,q)$ (Theorem~\ref{thm:cvar-minimax}) and by time
        homogeneity of the problem, we further have that for any $t \geq 0$, $\hat{x} \in
        \mathbb{X}$, and $\hat{q} \in [0,1]$,
        \begin{equation}
                V(\hat{x}, \hat{q}) = \max_{\hat{\gamma} \in \Gamma_t(\hat{q})} \inf_{\hat{\pi} \in \Pi_t(\hat{x})}  \E\left[ C_\infty^{t,\hat{x},\hat{\pi}} Q_\infty^{t,\hat{q},\hat{\gamma}} \right].
        \end{equation}
        Therefore, for each $\omega \in \Omega$ and $\pi \in \Pi(x)$, there exists $\hat{\gamma}^{\omega, \pi} \in \Gamma_{\tau(\omega)}( Q_\tau^{0,q,\gamma^\circ}(\omega) )$ such that
        \begin{align}
                V\left( X_\tau^{0,x,\pi}(\omega), Q_\tau^{0,q,\gamma^\circ}(\omega) \right)
                        &= \inf_{\hat{\pi} \in \Pi_{\tau(\omega)}( X_\tau^{0,x,\pi}(\omega))}
                                \E\left[ \left.
                                        C_\infty^{\tau, X_\tau^{0,x,\pi}, \hat{\pi}}Q_\infty^{\tau,Q_\tau^{0,q,\gamma^\circ},\hat{\gamma}^{\omega,\pi}}
                                \right| \Fscr_\tau \right](\omega)
                        \\&\leq
                                \E\left[ \left.
                                        C_\infty^{\tau, X_\tau^{0,x,\pi}, \pi}Q_\infty^{\tau,Q_\tau^{0,q,\gamma^\circ},\hat{\gamma}^{\omega,\pi}}
                                \right| \Fscr_\tau \right](\omega).
        \end{align}

        For each $\pi \in \Pi(x)$, consider an adversary's policy $\gamma^\pi: \mathbb{T} \times \Omega \rightarrow \mathbb{R}$ constructed as
        \begin{equation}
                \gamma_t^\pi(\omega) \defeq \left\{ \begin{array}{ll}
                                \gamma_t^\circ(\omega) & \text{if } t < \tau(\omega), \\
                                \hat{\gamma}_t^{\omega,\pi}(\omega) & \text{if } t \geq \tau(\omega).
                        \end{array} \right.
        \end{equation}
        We observe that $\gamma^\pi \in \Gamma(q)$.
        Combining these results, we obtain
        \begin{align}
                U(x,q)
                        &= \inf_{\pi \in \Pi(x)} \E\left[ C_\tau^{0,x,\pi} Q_\tau^{0,q,\gamma^\circ} + V\left( X_\tau^{0,x,\pi}, Q_\tau^{0,q,\gamma^\circ} \right) \right]
                        \\&\leq \inf_{\pi \in \Pi(x)} \E\left[
                                C_\tau^{0,x,\pi} Q_\tau^{0,q,\gamma^\circ} + \E\left[ \left. C_\infty^{\tau, X_\tau^{0,x,\pi}, \pi}Q_\infty^{\tau,Q_\tau^{0,q,\gamma^\circ},\hat{\gamma}^{\omega,\pi}} \right| \Fscr_\tau \right]
                        \right]
                        \\&= \inf_{\pi \in \Pi(x)} \E\left[
                                C_\tau^{0,x,\pi} Q_\tau^{0,q,\gamma^\pi} + C_\infty^{\tau, X_\tau^{0,x,\pi}, \pi}Q_\infty^{\tau,Q_\tau^{0,q,\gamma^\pi},\gamma^\pi}
                        \right]
                        \\&= \inf_{\pi \in \Pi(x)} \E\left[
                                C_\infty^{0, X_0^{0,x,\pi}, \pi} Q_\infty^{0,q,\gamma^\pi}
                        \right]
                        \\&\leq \inf_{\pi \in \Pi(x)} \max_{\gamma \in \Gamma(q)} \E\left[
                                C_\infty^{0, X_0^{0,x,\pi}, \pi} Q_\infty^{0,q,\gamma}
                        \right]
                        = V(x,q).
        \end{align}
        This concludes the proof.
\end{proof}

\newpage
\section{Proofs for \S \ref{sec:opt}} \label{app:proof-opt}

\subsection{Proofs of Theorem~\ref{thm:verification} and Theorem~\ref{thm:policy-optimality}}

 In this section, we prove Theorem~\ref{thm:verification} and
  Theorem~\ref{thm:policy-optimality} together in Theorem~\ref{thm:optimality-detail} below. We do
  this by constructing the sequence of function pairs $( f^{(n)}, g^{(n)} )_{n \in \mathbb{N}}$
  from a function $V^\star$ satisfying the conditions of Theorem~\ref{thm:verification} and by
  characterizing the outcome of the game when the trader or the adversary employs a policy induced
  by $f^{(n)}$ or $g^{(n)}$.  Before making a formal statement of
  Theorem~\ref{thm:optimality-detail}, we clarify what we aim to show, address some technical
  difficulty, and illustrate how to detour such a difficulty.

Consider a function $V^\star(x,q)$ that satisfies the conditions of Theorem~\ref{thm:verification}, and define
\begin{align*}
        f^\star(x,q) \defeq \frac{V^\star_x(x,q)}{\eta q}
        , \quad
        g^\star(x,q) \defeq \frac{\sigma x}{V^\star_{qq}(x,q)},
\end{align*}
which solve the trader's and the adversary's optimization problems in the HJB equation:
\begin{equation}
        f^\star(x,q) = \argmin_{v \in \mathbb{R}} \left\{ \frac{\eta}{2} q v^2 - V^\star_x\left( x, q \right) v \right\}
        , \quad
        g^\star(x,q) = \argmax_{w \in \mathbb{R}} \left\{ \frac{1}{2} V^\star_{qq}\left( x, q \right) w^2 - \sigma x w \right\},
      \end{equation}
for any \newedit{$x \in \mathbb{X}$} and $q \in (0,1)$.
We would argue that they specify the trader's optimal liquidation rate, $\pi_t^\star = f^\star(X_t, Q_t)$, and the adversary's optimal quantile diffusion rate, $\gamma_t^\star = g^\star(X_t, Q_t)$, and this combination of policies, $(\pi^\star, \gamma^\star)$, achieves the equilibrium of this game at which the outcome equals $V^\star(x,q)$.

First note that we aim to characterize the equilibrium of this game.  We will not argue
  that the trader's policy induced by $f^\star$ is optimal against any adversary's policy
  $\gamma$.  Instead, we will argue is that such a trader's policy yields an outcome that is no
  worse than $V^\star(x,q)$ against any adversary's policy $\gamma$, i.e.,
  $J(\pi^{\star,\gamma}, \gamma) \leq V^\star(x,q)$ for any adversary's policy $\gamma$, where
  $\pi^{\star,\gamma}$ is the trader's policy induced by $f^\star$ and coupled with $\gamma$.  Similarly,
  we will also argue that the adversary's policy induced by $g^\star$ yields an outcome that is no
  better than $V^\star(x,q)$ against any trader's policy $\pi$, i.e.,
  $J(\pi, \gamma^{\star,\pi}) \geq V^\star(x,q)$ for any trader's policy $\pi$ where
  $\gamma^{\star,\pi}$ is the adversary's policy induced by $g^\star$ and coupled with $\pi$.  By
  combining two results, we would be able to argue that a mutually coupled $(X,Q)$-Markov policy
  pair $(\pi^\star, \gamma^\star)$ induced by $(f^\star, g^\star)$ is the best response to each
  other and yields an outcome that equals $V^\star(x,q)$ as a saddle point.

To complete the above arguments, however, we need to show that the policies $\pi^\star$ and $\gamma^\star$ are admissible (i.e., they satisfy the feasibility conditions that we have introduced in $\Pi(x)$ and $\Gamma(q)$).
Unfortunately, their feasibility is uncertain due to the pathological behavior of $f^\star$ and $g^\star$ near the boundaries when $x \approx 0$, $q \approx 0$, or $q \approx 1$: for instance, we may have $\E[ \int_{t=0}^\infty \pi_t^2 dt ] = \infty$ since $\lim_{q \searrow 0} f^\star(x,q) = \infty$;
 $\E[ \int_{t=0}^\infty |X_t|^2 dt ] = \infty$ since $\lim_{q \nearrow 1} f^\star(x,q) = 0$;
 or $\E[ \int_{t=0}^\infty \gamma_t^2 dt ] = \infty$ since $\lim_{x \searrow 0} g^\star(x,q) = -\infty$, which contradicts to the fact that $Q_t$ is a bounded martingale.

To avoid this feasibility issue, we construct a sequence of function pairs $( f^{(n)}, g^{(n)} )_{n \in \mathbb{N}}$ that approximates $( f^\star, g^\star )$ while inducing feasible policies.
More specifically, we define a sequence of function pairs $( f^{(n)}, g^{(n)} )_{n \in
  \mathbb{N}}$:
\begin{align}\label{eq:fngn}
        f^{(n)}(x, q) \defeq \left\{ \begin{array}{ll}
                f^\star(x,q) & \text{if } (x,q) \in \Ascr_n, \\
                x/n & \text{if } (x,q) \notin \Ascr_n,
                \end{array} \right.
        \quad
        g^{(n)}(x,q) \defeq \left\{ \begin{array}{ll}
                g^\star(x,q) & \text{if } (x,q) \in \Ascr_n, \\
                0 & \text{if } (x,q) \notin \Ascr_n,
                \end{array} \right.
\end{align}
where
\begin{equation}
        \Ascr_n \defeq \left\{ (x, q) \in \newedit{\mathbb{X}} \times [0,1] \left| \, |x| > \frac{1}{n}, \, \frac{1}{n} < q < 1-\frac{1}{n} \right. \right\}.
\end{equation}
It can be easily seen that $f^{(n)}$ and $g^{(n)}$ converge to $f^\star$ and $g^\star$ at every interior point $(x,q) \in \R_+ \times (0,1)$ as $n \rightarrow \infty$.
Moreover, these functions suppress the extreme behavior of $f^\star$ and $g^\star$ near the
boundaries, i.e., outside the set $\Ascr_n$. The policies induced by $f^{(n)}$ and $g^{(n)}$ can
be shown to be feasible. This is formally established in Theorem~\ref{thm:optimality-detail}~\ref{it:opt-pi-feasibility}--\ref{it:opt-pi-gamma-feasibility}.
Finally, we can complete the equilibrium argument made above by showing that this sequence of function pairs $( f^{(n)}, g^{(n)} )_{n \in \mathbb{N}}$ achieves the equilibrium asymptotically (Theorem~\ref{thm:optimality-detail} \ref{it:opt-pi-optimality}--\ref{it:opt-pi-gamma-optimality}).
In the midst of the proof, we confirm that the equilibrium outcome, $V(x,q)$, coincides the solution of HJB equation, $V^\star(x,q)$ (Theorem~\ref{thm:optimality-detail} \ref{it:opt-V-star-validity-pi} and \ref{it:opt-V-star-validity-gamma}), by showing that the gap between them can be tightened arbitrarily small as $n \rightarrow \infty$.

\begin{theorem} \label{thm:optimality-detail}
        Fix \newedit{$x \in (0, M]$} and $q \in (0,1)$, and consider functions $V^\star$, $f^\star$, $g^\star$, $f^{(n)}$, and $g^{(n)}$ introduced above.
        For any $n > \max\{ \frac{1}{x}, \frac{1}{q}, \frac{1}{1-q} \}$, we have:
        \begin{enumerate}[label=(\roman*)]
                \item \label{it:opt-pi-feasibility}
                        For any adversary's policy $\gamma \in \Gamma(q)$, a $(X,Q)$-Markov trader's policy $\pi^{(n), \gamma}$ induced by $f^{(n)}$ and coupled with $\gamma$ is admissible, i.e., $\pi^{(n), \gamma} \in \Pi(x)$.
                \item \label{it:opt-gamma-feasibility}
                        For any trader's policy $\pi \in \Pi(x)$, a $(X,Q)$-Markov adversary's policy $\gamma^{(n),\pi}$ induced by $g^{(n)}$ and coupled with $\pi$ is admissible, i.e., $\gamma^{(n),\pi} \in \Gamma(q)$.
                \item \label{it:opt-pi-gamma-feasibility}
                        Mutually coupled $(X,Q)$-Markov policy pair $(\pi^{(n)}, \gamma^{(n)})$ induced by $(f^{(n)}, g^{(n)})$ is admissible, i.e., $\pi^{(n)} \in \Pi(x)$ and $\gamma^{(n)} \in \Gamma(q)$.
                \item \label{it:opt-V-star-validity-pi}
                        $V(x,q) \leq V^\star(x,q)$.
                \item \label{it:opt-V-star-validity-gamma}
                        $V(x,q) \geq V^\star(x,q)$.
                \item \label{it:opt-pi-optimality}
                        For any $\gamma \in \Gamma(q)$, $(\pi^{(n),\gamma})_{n \in \mathbb{N}}$ defined in \ref{it:opt-pi-feasibility} satisfies $\limsup_{n \rightarrow \infty} J( \pi^{(n), \gamma}, \gamma; x, q ) \leq V(x,q)$.
                \item \label{it:opt-gamma-optimality}
                        For any $\pi \in \Pi(x)$, $(\gamma^{(n),\pi})_{n \in \mathbb{N}}$ defined in \ref{it:opt-gamma-feasibility} satisfies $\liminf_{n \rightarrow \infty} J( \pi, \gamma^{(n),\pi}; x, q ) \geq V(x,q)$.
                \item \label{it:opt-pi-gamma-optimality}
                        $(\pi^{(n)}, \gamma^{(n)})_{n \in \mathbb{N}}$ defined in \ref{it:opt-pi-gamma-feasibility} satisfies $\lim_{n \rightarrow \infty} J( \pi^{(n)}, \gamma^{(n)}; x, q ) = V(x,q)$.
        \end{enumerate}
\end{theorem}

\begin{proof}
        Throughout the proof, we define a sequence of hitting times $( \tau_n )_{n \in \mathbb{N}}$ as follows:
        \begin{equation}\label{eq:taun}
                \tau_n \defeq \inf_{t \geq 0}\{ (X_t, Q_t) \notin \Ascr_n \}.
        \end{equation}
        Since $\Ascr_n^c$ is a closed set and $(X_t)_{t \geq 0}$ and $(Q_t)_{t \geq 0}$ are continuous, $\tau_n$ is a stopping time for each $n$.
        Also note that if the trader's policy is induced by $f^{(n)}$ or the adversary's policy is induced by $g^{(n)}$, then the process $(X_t,Q_t)$ never returns to $\Ascr_n$ once it leaves $\Ascr_n$, since $|X_t|$ is monotonically decreasing or $Q_t$ remains unchanged after $\tau_n$.

\subproof{Proof of claim \ref{it:opt-pi-feasibility}.} Under the trader's policy $\pi^{(n),\gamma}$, the position process $(X_t)_{t \geq 0}$ is described by
        \begin{equation} \label{eq:xn-sde}
                X_t = \left\{ \begin{array}{ll}
                        x - \int_{s=0}^t f^\star\left( X_s, Q_s \right) ds, & \forall t \leq \tau_n, \\
                        X_{\tau_n} e^{-(t-\tau_n)/n}, & \forall t \geq \tau_n.
                \end{array} \right.
        \end{equation}
        Observe that $X_t \in [0, x]$ for any $t \geq 0$, and in particular, $X_t$ is monotonically decreasing over time.
        Given that $(Q_t)_{t \geq 0}$ has a continuous sample path and $n$ is large enough such that $q \in \left[ \frac{1}{n}, 1 - \frac{1}{n} \right]$, we also have $Q_t \in \left[ \frac{1}{n}, 1 - \frac{1}{n} \right]$ for any $t \in [0, \tau_n]$.
        Then, the sample path of $(X_t)_{t \geq 0}$ for each $\omega$ is uniquely determined (i.e., the associated SDE has a unique strong solution): the uniqueness of sample path on $[0, \tau_n)$ follows from the fact that $f^\star(\cdot,\cdot)$ is Lipschitz continuous on $\left( \frac{1}{n}, x \right] \times \left( \frac{1}{n}, 1-\frac{1}{n} \right) \subset \Ascr_n$ due to condition \ref{it:verification-f-monotone}, and the uniqueness on $[\tau_n, \infty)$ immediately follows from the fact that $X_t = X_{\tau_n} e^{-(t-\tau_n)/n}$ for $t \geq \tau_n$.

        We next show that $\tau_n$ is almost surely bounded.
        For any $t < \tau_n$ (so that $(X_t,Q_t) \in \Ascr_n$), by condition \ref{it:verification-f-monotone}, we have
        \begin{equation}
                f^{(n)}(X_t, Q_t) \geq f^\star( X_t, 1 - 1/n ) \geq f^\star(1/n, 1-1/n) := \alpha_n > 0.
        \end{equation}
        Then, for any $t < \tau_n$, we have $dX_t/dt = - f^{(n)}(X_t, Q_t) \leq - \alpha_n$, and hence $X_{\tau_n} \leq x - \alpha_n \tau_n$.
        Together with the fact that $X_{\tau_n} \geq 0$, we deduce that
        \begin{equation} \label{eq:pi-tau-finite-mean}
                \tau_n \leq \frac{x}{ \alpha_n },
        \end{equation}
        on any sample path.

        We further have that, by condition \ref{it:verification-f-monotone}, for any $t < \tau_n$,
        \begin{equation}
                f^{(n)}(X_t, Q_t) \leq f^\star( X_t, 1/n ) \leq f^\star( x, 1/n ) := \beta_n < \infty.
        \end{equation}
        Then,
        \begin{equation}
                \int_{t=0}^{\tau_n} \big( \pi_t^{(n),\gamma} \big)^2 dt
                        = \int_{t=0}^{\tau_n} \big( f^\star(X_t, Q_t ) \big)^2 dt
                        \leq \tau_n \beta_n^2
                        \leq \frac{x \beta_n^2}{\alpha_n}
                , \quad
                \int_{t=0}^{\tau_n} X_t^2 dt \leq \tau_n x^2 \leq \frac{x^3}{\alpha_n},
        \end{equation}
        on any sample path.
        Further, since $\pi_t^{(n),\gamma} = X_t/n$ for any $t \geq \tau_n$,
        \begin{equation}
                \int_{t=\tau_n}^\infty \big( \pi_t^{(n),\gamma} \big)^2 dt = \frac{1}{n^2} \int_{t=\tau_n}^\infty X_t^2 dt = \frac{1}{n^2} \int_{t=\tau_n}^\infty X_{\tau_n}^2 e^{-2(t-\tau_n)/n} dt \leq \frac{x^2}{2n}
                , \quad
                \int_{t=\tau_n}^\infty X_t^2 dt \leq \frac{nx^2}{2}.
        \end{equation}
        As a result,
        \begin{equation}
                \E\left[ \left( \int_{t=0}^\infty \big( \pi_t^{(n),\gamma} \big)^2 dt \right)^2 \right] \leq \left( \frac{x \beta_n^2}{\alpha_n} + \frac{x^2}{2n} \right)^2 < \infty
                , \quad
                \E\left[ \int_{t=0}^\infty X_t^2 dt \right] \leq \frac{x^3}{\alpha_n} + \frac{nx^2}{2} < \infty.
        \end{equation}
        Also note that $\sup_{t \geq 0} |X_t| \leq x \leq M$.
        Therefore, $\pi^{(n), \gamma}$ is admissible, i.e., $\pi^{(n), \gamma} \in \Pi(x)$.

\subproof{Proof of claim \ref{it:opt-gamma-feasibility}.}
        Under the adversary's policy $\gamma^{(n),\pi}$, the martingale process $(Q_t)_{t \geq 0}$ is described by
        \begin{equation} \label{eq:qn-sde}
                Q_t = \left\{ \begin{array}{ll}
                        q + \int_{s=0}^t g^\star\left( X_s, Q_s \right) dW_s, & \forall t \leq \tau_n, \\
                        Q_{\tau_n}, & \forall t \geq \tau_n.
                \end{array} \right.
        \end{equation}
        From condition \ref{it:verification-g-monotone}, we can show that $g^\star(\cdot,\cdot)$ is Lipschitz continuous and bounded on $\left( \frac{1}{n}, \infty \right) \times \left( \frac{1}{n}, 1-\frac{1}{n} \right) \subset \Ascr_n$.
        Therefore, for each $\omega$, the sample path of $(Q_t)_{t \geq 0}$ is unique and continuous (given that $n$ is large enough so that $q \in \left[\frac{1}{n}, 1 - \frac{1}{n}\right]$, and $(X_t)_{t \geq 0}$ is continuous), and hence $Q_t \in \left[\frac{1}{n}, 1 - \frac{1}{n}\right] \subset [0,1]$ almost surely for any $t \in [0,\infty)$, and hence $\gamma^{(n),\pi}$ is admissible.

        For later use, we show that $\E[\tau_n] < \infty$ under the adversary's policy $\gamma^{(n),\pi}$.
        Observe that, for any $t < \tau_n$,
        \begin{equation}
                \left| g^{(n)}(X_t, Q_t ) \right| = \left| g^\star(X_t, Q_t) \right| \geq \left| g^\star(M, Q_t) \right| \geq \min_{q \in [1/n,1-1/n]}\left\{ g^\star(M, q)  \right\} := \xi_n.
        \end{equation}
        In other words, the diffusion rate of the quantile process $\big( Q_t \big)_{t \geq 0}$ is lower bounded by $\xi_n$ (up to time $\tau_n$).
        Therefore, its quadratic variation process $\big( \left<Q\right>_t \big)_{t \geq 0}$ satisfies
        \begin{equation}
                \left<Q\right>_t = \int_{s=0}^{t \wedge \tau_n} \big( g^\star(X_s, Q_s) \big)^2 ds \geq \xi_n^2 (t \wedge \tau_n),
        \end{equation}
        for any $t \in \mathbb{R}$.
       On the other hand, by Burkholder-Davis-Gundy inequality, there exists $c$ such that $c \E[ \left<Q\right>_{t \wedge \tau_n} ] \leq \E[ ( \sup_{s \leq t \wedge \tau_n} Q_s )^2  ]$ for any $t$, and thus $\xi_n^2 \E[ t \wedge \tau_n ] \leq \E[ \left<Q\right>_{t \wedge \tau_n} ] \leq \E[ ( \sup_{s \leq t \wedge \tau_n} Q_s )^2  ]/c \leq 1 / c$.
        By monotone convergence theorem, we deduce that $\E[ \tau_n ] = \lim_{t \rightarrow \infty} \E[ t \wedge \tau_n ] \leq 1/(c \xi_n^2) < \infty$.

\subproof{Proof of claim \ref{it:opt-pi-gamma-feasibility}.}
        Under the mutually coupled policy pair $(\pi^{(n)}, \gamma^{(n)})$, the process pair $(X_t, Q_t)_{t \geq 0}$ satisfies \eqref{eq:xn-sde} and \eqref{eq:qn-sde} simultaneously.
        By the same argument above (claim \ref{it:opt-gamma-feasibility}), we can show that $\gamma^{(n)}$ is admissible, and also that $\pi^{(n)}$ is admissible by claim \ref{it:opt-pi-feasibility}.

\subproof{Proof of claim \ref{it:opt-V-star-validity-pi}.}
        Fix $\gamma \in \Gamma(q)$ and $n \in \mathbb{N}$.
        For simplicity, we let $\pi$ denote $\pi^{(n),\gamma} \in \Pi(x)$.
        For any $t < \tau_n$, since $f^{(n)}(X_t,Q_t) = f^\star(X_t, Q_t)$, by condition \ref{it:verification-HJB}, we have
        \begin{align}
                 & \left( \frac{\eta}{2} Q_t \pi_t^2 - V_x^\star\left( X_t, Q_t \right) \pi_t \right)
                        +  \left( \frac{1}{2} V_{qq}^\star\left( X_t, Q_t \right) \gamma_t^2 - \sigma X_t \gamma_t \right)
                 \\&= \left( \frac{\eta}{2} Q_t \left( f^\star\left( X_t, Q_t \right) \right)^2 - V_x^\star\left( X_t, Q_t \right) f^\star\left( X_t, Q_t \right) \right)
                        +  \left( \frac{1}{2} V_{qq}^\star\left( X_t, Q_t \right) \gamma_t^2 - \sigma X_t \gamma_t \right)
                 \\&= \min_{v \in \mathbb{R}} \left\{ \frac{\eta}{2} Q_t v^2 - V_x^\star\left( X_t, Q_t \right) v \right\}
                        + \left( \frac{1}{2} V_{qq}^\star\left( X_t, Q_t \right) \gamma_t^2 - \sigma X_t \gamma_t \right)
                 \\&\leq \min_{v \in \mathbb{R}} \left\{ \frac{\eta}{2} Q_t v^2 - V_x^\star\left( X_t, Q_t \right) v \right\}
                        + \max_{w \in \mathbb{R}} \left\{ \frac{1}{2} V_{qq}^\star\left( X_t, Q_t \right) w^2 - \sigma X_t w \right\}
                \\&= 0.
        \end{align}
        By plugging this result into Proposition~\ref{prop:cvar-dp-dynkin}, we obtain
        \begin{align} \label{eq:pi-optimality-sub-1}
                & \E\left[ C_{\tau_n} Q_{\tau_n} \right]
                \\&\leq \E\left[ C_{\tau_n} Q_{\tau_n} + V^\star\left( X_{\tau_n}, Q_{\tau_n} \right) \right]
                \\&= V^\star\left( x, q \right)
                        + \E\left[ \int_{t=0}^{\tau_n}\left\{
                                \left( \frac{\eta}{2} Q_t \pi_t^2 - V_x^\star\left( X_t, Q_t \right) \pi_t \right)
                                +  \left( \frac{1}{2} V_{qq}^\star\left( X_t, Q_t \right) \gamma_t^2 - \sigma X_t \gamma_t \right)
                         \right\} dt \right]
                \\&\leq V^\star\left( x, q \right),
        \end{align}
        where the first inequality follows from the non-negativity of $V^\star(\cdot, \cdot)$.
        On the other hand, by the definition of $\tau_n$ and the continuity of sample paths, we have either $X_{\tau_n} = \frac{1}{n}$ or $Q_{\tau_n} \in \{ \frac{1}{n}, 1 - \frac{1}{n} \}$, and therefore,
        \begin{align}
                \E\left[ V\left( X_{\tau_n}, Q_{\tau_n} \right) \right]
                        &\leq \E\left[ \max\left\{ V( 1/n, Q_{\tau_n} ),~ V( X_{\tau_n}, 1/n ),~ V( X_{\tau_n}, 1 - 1/ n ) \right\} \right]
                        \\&\leq \sup_{q' \in [0,1]}\{ V( 1/n, q' ) \} + V( x, 1/n ) + V( x, 1-1/n ) := A_n(x),
                         \label{eq:pi-optimality-sub-2}
        \end{align}
        where the second inequality follows from the fact that $X_\tau \leq x$ under policy $\pi^{(n),\gamma}$ and $V(\cdot, q)$ is increasing on $[0, \infty)$.
        It can be shown easily that $\lim_{n \rightarrow \infty} A_n(x) = 0$ for any $x$ due to Proposition~\ref{prop:V-properties}.

        Since $\gamma$ was chosen to be arbitrary, from \eqref{eq:pi-optimality-sub-1}, we deduce that
        \begin{equation} \label{eq:pi-optimality-sub-3}
                \sup_{\gamma \in \Gamma(q)} \E\left[ C_{\tau_n}^{x,\pi^{(n),\gamma}} Q_{\tau_n}^{q,\gamma} \right] \leq V^\star( x, q ).
        \end{equation}
        Utilizing Theorem~\ref{thm:cvar-dp} (we have shown that $\E[\tau_n]<\infty$ in \eqref{eq:pi-tau-finite-mean}) and Proposition~\ref{prop:cvar-dp-minimax} with \eqref{eq:pi-optimality-sub-2} and \eqref{eq:pi-optimality-sub-3}, we further obtain
        \begin{align}
                V(x,q)
                        &\stackrel{ \text{Thm}~\ref{thm:cvar-dp}}{=} \inf_{\pi \in \Pi(x)}  \sup_{\gamma \in \Gamma(q)} \E\left[ C_{\tau_n}^{x,\pi} Q_{\tau_n}^{q,\gamma} + V\left( X_{\tau_n}^{x,\pi}, Q_{\tau_n}^{q,\gamma} \right) \right]
                        \\&\stackrel{ \text{Prop}~\ref{prop:cvar-dp-minimax}}{=} \sup_{\gamma \in \Gamma(q)} \inf_{\pi \in \Pi(x)} \E\left[ C_{\tau_n}^{x,\pi} Q_{\tau_n}^{q,\gamma} + V\left( X_{\tau_n}^{x,\pi}, Q_{\tau_n}^{q,\gamma} \right) \right]
                        \\&\leq \sup_{\gamma \in \Gamma(q)} \E\left[ C_{\tau_n}^{x,\pi^{(n),\gamma}} Q_{\tau_n}^{q,\gamma} + V\left( X_{\tau_n}^{x,\pi^{(n),\gamma}}, Q_{\tau_n}^{q,\gamma} \right) \right]
                        \\&\stackrel{ \eqref{eq:pi-optimality-sub-2} }{\leq} \sup_{\gamma \in \Gamma(q)} \E\left[ C_{\tau_n}^{x,\pi^{(n),\gamma}} Q_{\tau_n}^{q,\gamma}  \right]  + A_n(x)
                        \\&\stackrel{ \eqref{eq:pi-optimality-sub-3} }{\leq} V^\star( x, q ) + A_n(x).
        \end{align}
        Since $\lim_{n \rightarrow \infty} A_n(x) = 0$, we obtain $V(x,q) \leq V^\star(x,q)$, which concludes the proof.

\subproof{Proof of claim \ref{it:opt-V-star-validity-gamma}.}
        The proof is almost symmetric to that of claim \ref{it:opt-V-star-validity-pi}.
        Fix $\pi \in \Pi(x)$ and $n \in \mathbb{N}$.
        For simplicity, we let $\gamma$ denote $\gamma^{(n),\pi} \in \Gamma(q)$.
        For any $t < \tau_n$, since $g^{(n)}(X_t,Q_t) = g^\star(X_t, Q_t)$, by condition \ref{it:verification-HJB}, we have
        \begin{align}
                 & \left( \frac{\eta}{2} Q_t \pi_t^2 - V_x^\star\left( X_t, Q_t \right) \pi_t \right)
                        +  \left( \frac{1}{2} V_{qq}^\star\left( X_t, Q_t \right) \gamma_t^2 - \sigma X_t \gamma_t \right)
                 \\&= \left( \frac{\eta}{2} Q_t \pi_t^2 - V_x^\star\left( X_t, Q_t \right) \pi_t \right)
                        + \max_{w \in \mathbb{R}} \left\{ \frac{1}{2} V_{qq}^\star\left( X_t, Q_t \right) w^2 - \sigma X_t w \right\}
                 \\&\geq \min_{v \in \mathbb{R}} \left\{ \frac{\eta}{2} Q_t v^2 - V_x^\star\left( X_t, Q_t \right) v \right\}
                        + \max_{w \in \mathbb{R}} \left\{ \frac{1}{2} V_{qq}^\star\left( X_t, Q_t \right) w^2 - \sigma X_t w \right\}
                \\&= 0.
        \end{align}
        Consequently, by Proposition~\ref{prop:cvar-dp-dynkin},
        \begin{align} \label{eq:gamma-optimality-sub-1}
                &\E\left[ C_{\tau_n} Q_{\tau_n} + V^\star\left( X_{\tau_n}, Q_{\tau_n} \right) \right]
                \\&= V^\star\left( x, q \right)
                        + \E\left[ \int_{t=0}^{\tau_n}\left\{
                                \left( \frac{\eta}{2} Q_t \pi_t^2 - V_x^\star\left( X_t, Q_t \right) \pi_t \right)
                                +  \left( \frac{1}{2} V_{qq}^\star\left( X_t, Q_t \right) \gamma_t^2 - \sigma X_t \gamma_t \right)
                         \right\} dt \right]
                \\&\geq V^\star\left( x, q \right).
        \end{align}
        By the definition of $\tau_n$ and the continuity of sample paths, we have either $X_{\tau_n} = \frac{1}{n}$ or $Q_{\tau_n} \in \{ \frac{1}{n}, 1 - \frac{1}{n} \}$, and hence,
        \begin{align} \label{eq:gamma-optimality-sub-2}
                \E\left[ V^\star\left( X_{\tau_n}, Q_{\tau_n} \right) \right]
                        &\leq \E\left[ \max\left\{ V^\star( 1/n, Q_{\tau_n} ),~ V^\star( X_{\tau_n}, 1/n ),~ V^\star( X_{\tau_n}, 1-1/n ) \right\} \right]
                        \\&\leq \sup_{q' \in [0,1]}\{ V^\star( 1/n, q' ) \} + V^\star( M, 1/n ) + V^\star( M, 1-1/n ) := B_n,
        \end{align}
        where the last inequality follows from the feasibility of $\pi$ (i.e., $|X_{\tau_n}| \leq M$) and the monotonicity of $V^\star(\cdot, q)$ (condition \ref{it:verification-symmetry}).
        By condition \ref{it:verification-boundary}, we have $\lim_{n \rightarrow 0} B_n = 0$.

        Since $\pi$ was chosen arbitrary, we further deduce that
        \begin{equation} \label{eq:gamma-optimality-sub-3}
                V^\star(x,q)
                                \stackrel{\eqref{eq:gamma-optimality-sub-1}}{\leq} \inf_{\pi \in \Pi(x)} \E\left[ C_{\tau_n}^{x,\pi} Q_{\tau_n}^{q,\gamma^{(n),\pi}} + V^\star\left( X_{\tau_n}^{x,\pi}, Q_{\tau_n}^{q,\gamma^{(n),\pi}} \right)  \right]
                        \stackrel{\eqref{eq:gamma-optimality-sub-2}}{\leq} \inf_{\pi \in \Pi(x)} \E\left[ C_{\tau_n}^{x,\pi} Q_{\tau_n}^{q,\gamma^{(n),\pi}} \right] + B_n.
        \end{equation}
        By utilizing Theorem~\ref{thm:cvar-dp} (we have shown that $\E[\tau_n] < \infty$ at the end of the proof of claim \ref{it:opt-gamma-feasibility}) and the above results, we obtain
        \begin{align}
                V(x,q)
                        &\stackrel{\text{Thm~\ref{thm:cvar-dp}}}{=} \inf_{\pi \in \Pi(x)} \sup_{\gamma \in \Gamma(q)} \E\left[ C_{\tau_n}^{x,\pi} Q_{\tau_n}^{q,\gamma} + V\left( X_{\tau_n}^{x,\pi}, Q_{\tau_n}^{q,\gamma} \right) \right]
                        \\&\geq \inf_{\pi \in \Pi(x)} \E\left[ C_{\tau_n}^{x,\pi} Q_{\tau_n}^{q,\gamma^{(n),\pi}} + V\left( X_{\tau_n}^{x,\pi}, Q_{\tau_n}^{q,\gamma^{(n),\pi}} \right) \right]
                        \\&\geq \inf_{\pi \in \Pi(x)} \E\left[ C_{\tau_n}^{x,\pi} Q_{\tau_n}^{q,\gamma^{(n),\pi}}  \right]
                        \\&\stackrel{\eqref{eq:gamma-optimality-sub-3}}{\geq} V^\star( x, q ) - B_n,
        \end{align}
        where the first inequality follows from the definition of supremum, and the second inequality follows from the non-negativity of $V(\cdot, \cdot)$.
        By taking $n \nearrow \infty$, we obtain the desired result.

\subproof{Proof of claim \ref{it:opt-pi-optimality}.}
        To simplify notation, let $\pi := \pi^{(n),\gamma}$ and $\tau := \tau_n$ temporarily.
        By Proposition~\ref{prop:time-decomposition}, we have
        \begin{equation}
                J( \pi, \gamma; x, q )
                        = \E\left[ C_\tau^{0,x,\pi} Q_\tau^{0,q,\gamma} + C_\infty^{\tau, X_\tau^{0,x,\pi}, \pi} Q_\infty^{\tau, Q_\tau^{0,q, \gamma}, \gamma}  \right].
        \end{equation}
        (Recall that $C_\infty^{\tau, X_\tau^{0,x,\pi}, \pi} = C_\infty^{0,x,\pi} - C_\tau^{0,x,\pi}$.)
        In the proof of claim \ref{it:opt-V-star-validity-pi}, in \eqref{eq:pi-optimality-sub-1}, we have shown that
        \begin{equation}
                \E\left[ C_\tau^{0,x,\pi} Q_\tau^{0,q,\gamma} \right] \leq V^\star(x,q) = V(x,q),
        \end{equation}
        where the last equality follows from the claims \ref{it:opt-V-star-validity-pi} and \ref{it:opt-V-star-validity-gamma}.
        Therefore,
        \begin{align}
                J( \pi, \gamma; x, q )
                        &\leq V(x,q) + \E\left[ C_\infty^{\tau, X_\tau^{0,x,\pi}, \pi} Q_\infty^{\tau, Q_\tau^{0,q, \gamma}, \gamma}  \right]
                        \\&\leq V(x,q) + \E\left[ \scvar_{Q_\tau^{0,q,\gamma}}\left[ C_\infty^{\tau, X_\tau^{0,x,\pi}, \pi}  \right]  \right].
        \end{align}
        To obtain the desired result, it suffices to show that $\lim_{n \rightarrow \infty} \E\left[ \scvar_{Q_\tau^{0,q,\gamma}}\left[ C_\infty^{\tau, X_\tau^{0,x,\pi}, \pi}  \right]  \right] = 0$.

        Note that after time $\tau_n$, the policy $\pi^{(n),\gamma}$ trades according to the exponential schedule (i.e,. $X_t = X_{\tau_n} e^{-(t-\tau_n)/n}$), and thus $C_\infty^{\tau_n, X_{\tau_n}, \pi}$ is normally distributed conditional on $X_{\tau_n}$: more specifically, as derived in \eqref{eq:mv-exponential}, we have
        \begin{equation}
                \E\left[ \left. C_\infty^{\tau_n, X_{\tau_n}, \pi} \right| X_{\tau_n} \right] = \frac{ \eta X_{\tau_n}^2 }{ 4 n }
                , \quad
                \text{Var}\left[ \left. C_\infty^{\tau_n, X_{\tau_n}, \pi} \right| X_{\tau_n} \right] = \frac{ n \sigma^2 X_{\tau_n}^2  }{ 2 }.
        \end{equation}
        Consequently,
        \begin{equation}
                \scvar_{Q_{\tau_n}}\left[ C_\infty^{\tau_n, X_{\tau_n}, \pi} \right] = \frac{ \eta X_{\tau_n}^2 Q_{\tau_n} }{ 4 n } + \kappa(Q_{\tau_n}) \times \sqrt{ \frac{ n \sigma^2 X_{\tau_n}^2  }{ 2 } }.
        \end{equation}
        By the definition of $\tau_n$, we have either $X_{\tau_n} = \frac{1}{n}$ or $Q_{\tau_n} \in \left\{ \frac{1}{n}, 1 - \frac{1}{n} \right\}$, and thus
        \begin{align}
                \scvar_{Q_{\tau_n}}\left[ C_\infty^{\tau_n, X_{\tau_n}, \pi} \right]
                        &\leq \max\left\{
                                \frac{ \eta Q_{\tau_n} }{ 4n^3 } + \kappa(Q_{\tau_n}) \times \sqrt{ \frac{ \sigma^2 }{ 2n } }
                                , ~
                                \frac{ \eta X_{\tau_n}^2 Q_\tau }{ 4 n } +  \max\left\{ \kappa(1/n), \kappa(1-1/n)  \right\} \times \sqrt{ \frac{ n \sigma^2 X_{\tau_n}^2  }{ 2 } }
                        \right\}
                        \\&\leq \max\left\{
                                \frac{ \eta }{ 4n^3 } + \sup_{q'}\{ \kappa(q') \} \times \sqrt{ \frac{ \sigma^2 }{ 2n } }
                                , ~
                                \frac{ \eta x^2 }{ 4 n } +  \max\left\{ \kappa(1/n), \kappa(1-1/n)  \right\} \times \sqrt{ \frac{ n \sigma^2 x^2  }{ 2 } }
                        \right\}
                        \label{eq:pi-optimality-sub-5}
                        .
        \end{align}
        Since $\sup_{q'} \kappa(q') < \infty$, $\lim_{n \rightarrow \infty} \sqrt{n} \kappa(1/n) =
        0$, and $\lim_{n \rightarrow \infty} \sqrt{n} \kappa(1-1/n) = 0$
        (Lemma~\ref{lem:kappa-limit}), the right side of \eqref{eq:pi-optimality-sub-5} vanishes as $n \rightarrow \infty$.
        Therefore, we have $\lim_{n \rightarrow \infty} \E\left[ \scvar_{Q_{\tau_n}}\left[ C_\infty^{\tau_n, X_{\tau_n}, \pi} \right] \right] = 0$ and this concludes the proof.

\subproof{Proof of claim \ref{it:opt-gamma-optimality}.}
        In the proof of claim \ref{it:opt-V-star-validity-gamma}, we have shown that
        \begin{align}
                \E\left[ C_{\tau_n}^{x,\pi} Q_{\tau_n}^{q,\gamma^{(n),\pi}} + V^\star\left( X_{\tau_n}^{x,\pi}, Q_{\tau_n}^{q,\gamma^{(n),\pi}} \right) \right]
                \geq V^\star( x, q ) = V(x,q),
        \end{align}
        where the last inequality follows from the claims \ref{it:opt-V-star-validity-pi} and \ref{it:opt-V-star-validity-gamma}, and $\lim_{n \rightarrow \infty} \E\left[ V^\star\left( X_{\tau_n}^{x,\pi}, Q_{\tau_n}^{q,\gamma^{(n),\pi}} \right) \right] = 0$.

        Observe that we have $Q_{\infty} = Q_{\tau_n}$ since $g^\star(X_t,Q_t) = 0$ for all $t \geq \tau_n$ under $\gamma^{(n),\pi}$.
        Therefore,
        \begin{align}
                \E\left[ C_{\infty}^{x,\pi} Q_{\infty}^{q,\gamma^{(n),\pi}} \right] - \E\left[ C_{\tau_n}^{x,\pi} Q_{\tau_n}^{q,\gamma^{(n),\pi}} \right]
                        &= \E\left[ \left( C_{\infty} - C_{\tau_n} \right) Q_{\tau_n} \right]
                        \\&= \E\left[ \E\left( \left. \int_{t=\tau_n}^\infty \frac{\eta}{2} \pi_t^2dt - \int_{t=\tau_n}^\infty \sigma X_t dW_t \right| \Fscr_{\tau_n} \right) Q_{\tau_n}  \right]
                        \\&= \E\left[ \left( \int_{t=\tau_n}^\infty \frac{\eta}{2} \pi_t^2dt \right) Q_{\tau_n}^{q,\gamma^{(n)}}  \right]
                        \\&\geq 0.
        \end{align}
        Consequently,
        \begin{equation}
                J( \pi, \gamma^{(n),\pi}; x, q )
                        = \E\left[ C_{\infty}^{x,\pi} Q_{\infty}^{q,\gamma^{(n),\pi}} \right]
                        \geq \E\left[ C_{\tau_n}^{x,\pi} Q_{\tau_n}^{q,\gamma^{(n),\pi}} \right]
                        \geq V(x,q) - \E\left[ V^\star\left( X_{\tau_n}^{x,\pi}, Q_{\tau_n}^{q,\gamma^{(n),\pi}} \right) \right].
        \end{equation}
        By taking $\liminf_{n \rightarrow \infty}$ on both sides, we obtain the desired result.
\end{proof}

\begin{proof}[\proofnamest{Proof of Theorem~\ref{thm:verification}}]
        We aim to show that $V(x,q) = V^\star(x,q)$ for any $x \in \mathbb{R}$ and $q \in [0,1]$.
        From Theorem~\ref{thm:optimality-detail}\ref{it:opt-V-star-validity-pi} and \ref{it:opt-V-star-validity-gamma}, we deduce that $V(x,q) = V^\star(x,q)$ for any $x > 0$ and $q \in (0,1)$.
        By symmetry, the same argument holds for any $x < 0$ and $q \in (0,1)$.
        From Proposition~\ref{prop:V-properties}\ref{it:V-boundary} and condition \ref{it:verification-boundary}, we can also verify that their boundary values match, i.e., $V^\star(x,q) = V(x,q) = 0$ if $x=0$ or $q \in \{0,1\}$.
\end{proof}

\begin{proof}[\proofnamest{Proof of Theorem~\ref{thm:policy-optimality}}]
        \newedit{For any $x \in (0, M]$ and $q \in (0,1)$}, the claims follow from Theorem~\ref{thm:optimality-detail}\ref{it:opt-pi-feasibility},
        \ref{it:opt-gamma-feasibility}, \ref{it:opt-pi-gamma-feasibility},
        \ref{it:opt-pi-optimality}, \ref{it:opt-gamma-optimality}, and
        \ref{it:opt-pi-gamma-optimality}.
        \newedit{The claims also hold true for any $x \in [-M, 0)$ by symmetry.}

        We next examine the cases when $x=0$, $q =0$, or $q=1$.
        Note that $V(0, \cdot) = V(\cdot,0) = V(\cdot,1) = 0$.

        \begin{itemize}
        \item Suppose $x=0$.
                Under $\pi^{(n),\gamma}$, we have $X_t=0$ for any $t \geq 0$ and thus $C_\infty^{x,\pi^{(n),\gamma}} = 0$ almost surely for any $\gamma \in \Gamma(q)$. Therefore, $J(\pi^{(n),\gamma}, \gamma; x,q) = 0$, from which the claims \ref{it:opt-pi-optimality-general} and \ref{it:opt-pi-gamma-optimality-general} follow.
        Consequently, under $\gamma^{(n),\pi}$, we have $Q_t = q$ for any $t \geq 0$ since $g^{(n)}(0, Q_t) =0$, and thus $\E[ C_\infty^{x,\pi} Q_\infty^{q, \gamma^{(n),\pi}} ] = q \E[ C_\infty^{x,\pi} ] \geq 0$ for any $\pi \in \Pi(x)$, from which the claim \ref{it:opt-gamma-optimality-general} follows.

        \item Suppose $q=0$.
        The only admissible adversary's policy is the one satisfying that $Q_t = 0$ for all $t \geq 0$, and thus $J(\pi, \gamma; x,q) = 0$ for any $\pi \in \Pi(x)$, from which claims \ref{it:opt-pi-optimality-general}--\ref{it:opt-pi-gamma-optimality-general} follow.

        \item Suppose $q=1$.
        Again, the only admissible adversary's policy is the one satisfying that $Q_t = 1$ for all $t \geq 0$.
        Then, $\pi^{(n),\gamma}$ is an exponential policy in which we have $X_t = x \exp( - t/n )$ and thus $J(\pi^{(n), \gamma},\gamma; x,q) = \E[ C_\infty^{x,\pi^{(n),\gamma}} ] = \frac{ \eta x^2 }{ 4 n } \searrow 0$ as $n \rightarrow \infty$, from which claims \ref{it:opt-pi-optimality-general} and \ref{it:opt-pi-gamma-optimality-general} follow.
        Consequently, since $\E[ C_\infty^{x,\pi} ] \geq 0$ for any $\pi \in \Pi(x)$, the claim \ref{it:opt-gamma-optimality-general} immediately follows.
      \end{itemize}
\end{proof}

\subsection{Other Proofs}


\begin{proof}[\proofnamest{Proof of Proposition~\ref{prop:cvar-dp-dynkin}}]
By applying It\^o's formula, we obtain
        \begin{equation}
                \widehat{V}(X_t, Q_t)
                        = \widehat{V}(x,q) - \int_{s=0}^t \widehat{V}_x(X_s, Q_s) \pi_s ds
                        + \int_{s=0}^t \widehat{V}_q( X_s, Q_s ) \gamma_s dW_s
                        + \frac{1}{2} \int_{s=0}^t \widehat{V}_{qq}(X_s, Q_s) \gamma_s^2 ds.
        \end{equation}
        Recall that $C_t \defeq \int_{s=0}^t \frac{\eta}{2} \pi_s^2 ds - \int_{s=0}^t \sigma X_s dW_s$.
        By applying It\^o's product rule, we further obtain
        \begin{equation}
                C_t Q_t
                        = Q_0 C_0
                        + \int_{s=0}^t \left( \frac{\eta}{2} \pi_s^2 Q_s - \sigma X_s \gamma_s \right) ds
                        + \int_{s=0}^t \left( - \sigma X_s Q_s + X_s \gamma_s \right) dW_s.
        \end{equation}

        Let $M_t^{(1)} \defeq \int_{s=0}^t \widehat{V}_q(X_s, Q_s) \gamma_s dW_s$, $M_t^{(2)} \defeq - \sigma \int_{s=0}^t X_s Q_s dW_s$, and $M_t^{(3)} \defeq \int_{s=0}^t X_s \gamma_s dW_s$.
        We first verify that  $\big( M_t^{(1)} \big)_{t \geq 0}$ is uniformly integrable:
        \begin{equation}
                        \sup_t \E\left[ |M_t^{(1)}|^2 \right]
                        \leq \E\left[ \int_{s=0}^\infty \widehat{V}_q(X_s, Q_s)^2 \gamma_s^2 ds \right]
                        \leq \left( \sup_{x, q} \widehat{V}_q(x,q) \right)^2 \E\left[ \int_{s=0}^\infty \gamma_s^2 ds \right] < \infty,
        \end{equation}
        where $\sup_{x, q} \widehat{V}_q(x,q) < \infty$ since we have assumed that the domain of $\widehat{V}$ is compact and $\widehat{V}_q$ is continuous, and $\E\left[ \int_{s=0}^\infty \gamma_s^2 ds \right] < \infty$ since $( Q_t )_{t \geq 0}$ is a bounded martingale.
        Consequently,
        \begin{equation}
                \sup_t \E\left[ |M_t^{(2)}|^2 \right]
                        \leq \E\left[ \int_{s=0}^\infty X_s^2 Q_s^2 ds \right]
                        \leq \E\left[ \int_{s=0}^\infty X_s^2 ds \right] < \infty,
        \end{equation}
        where the last inequality directly follows from the definition of $\Pi(x)$.
        Similarly, we have $\sup_t \E\left[ |M_t^{(3)}|^2 \right] \leq \E\left[ \int_{s=0}^\infty X_s^2 \gamma_s^2 ds \right] \leq M^2 \E\left[ \int_{s=0}^\infty \gamma_s^2 ds \right] < \infty$.
        Therefore, $\big( M_t^{(1)} \big)_{t \geq 0}$, $\big( M_t^{(2)} \big)_{t \geq 0}$, and $\big( M_t^{(3)} \big)_{t \geq 0}$ are all uniformly integrable, and hence, $\E[ M_\tau^{(1)} ] = \E[ M_\tau^{(2)} ] = \E[ M_\tau^{(3)} ] = 0$ for any stopping time $\tau$.

        Combining these results, we deduce that
        \begin{align}
                \E\left[  C_\tau Q_\tau + \widehat{V}(X_\tau, Q_\tau ) \right]
                        &= \widehat{V}(x,q) + \E\left[ \int_{t=0}^\tau \left\{ \frac{\eta}{2} \pi_t^2 Q_t - \widehat{V}_x(X_t,Q_t)\pi_t + \frac{1}{2} \widehat{V}_{qq}(X_t,Q_t) \gamma_t^2 - \sigma X_t \gamma_t \right\} dt \right].
        \end{align}
        We obtain the claim by rearranging terms.
\end{proof}

\begin{proof}[\proofnamest{Proof of Theorem~\ref{thm:value-function}}]
We prove that the function $V^\star$, defined in \eqref{eq:V-opt}, satisfies the conditions \ref{it:verification-HJB}--\ref{it:verification-g-monotone} in Theorem~\ref{thm:verification}.

\subproof{Verification of condition \ref{it:verification-HJB}.}
Observe that $\varphi''(q) = - \frac{q}{\varphi^2(q)} < 0$ for any $q \in (0,1)$, and thus and $V_{qq}^\star(x,q) < 0$.
The minimum/maximum in \eqref{eq:HJB} are attainable, and we have
\begin{equation}
        \min_{v \in \mathbb{R}} \left\{ \frac{\eta}{2} q v^2 - V^\star_x\left( x, q \right) v \right\}
                +
                \max_{w \in \mathbb{R}} \left\{ \frac{1}{2} V^\star_{qq}\left( x, q \right) w^2 - \sigma x w \right\}
                = - \frac{\big(V_x^\star(x,q)\big)^2}{2\eta q} - \frac{\sigma^2 x^2}{ 2V_{qq}^\star(x,q) }
\end{equation}
Therefore, it suffices show that ${V_x^\star}^2 \times V_{qq}^\star = - \sigma^2 \eta \times x^2 q$:
\begin{align*}
        {V_x^\star}^2 \times V_{qq}^\star
                &= \left( (3/4)^{\frac{2}{3}} \times \sigma^{\frac{2}{3}} \eta^{\frac{1}{3}} \times \frac{4}{3} x^{\frac{1}{3}} \times \varphi(q)\right)^2 \times \left((3/4)^{\frac{2}{3}} \times \sigma^{\frac{2}{3}} \eta^{\frac{1}{3}} \times |x|^{\frac{4}{3}} \times \varphi''(q) \right)
                \\&= \sigma^2 \eta \times x^2 \times \varphi^2(q) \times \varphi''(q)
                = - \sigma^2 \eta x^2 q.
\end{align*}

\subproof{Verification of conditions \ref{it:verification-boundary}--\ref{it:verification-symmetry}.}
These conditions can be easily verified by inspection.

\subproof{Verification of condition \ref{it:verification-f-monotone}.}
Note that $\frac{V^\star_x(x,q)}{q} = (3/4)^{-\frac{1}{3}} \times \sigma^{\frac{2}{3}} \eta^{\frac{1}{3}} \times x^{\frac{1}{3}} \times \frac{\varphi(q)}{q}$.
Its continuity and monotonicity (with respect to $x$) can be verified immediately, and it suffices to show that $\frac{\varphi(q)}{q}$ is decreasing in $q$.

Observe that for any $q_1, q_2$ such that $0 < q_1 \leq q_2 \leq 1$, we have $\varphi(q_1) \geq \frac{q_1}{q_2} \varphi(q_2) + \left(1 - \frac{q_1}{q_2}\right) \varphi(0) = q_1 \times \frac{\varphi(q_2)}{q_2}$ due to its concavity, and therefore $\frac{\varphi(q_1)}{q_1} \geq \frac{\varphi(q_2)}{q_2}$. As a result, $\frac{\varphi(q)}{q}$ is monotonically decreasing in $q$ on $(0,1)$.

\subproof{Verification of condition \ref{it:verification-g-monotone}.}
Note that $\frac{x}{V^\star_{qq}(x,q)} = (3/4)^{-\frac{2}{3}} \times \sigma^{-\frac{2}{3}} \eta^{-\frac{1}{3}} \times x^{-\frac{1}{3}} \times \frac{\varphi^2(q)}{q}$. The continuity and monotonicity (with respect to $x$) is trivial.
\end{proof}

\begin{proof}[\proofnamest{Proof of Proposition~\ref{prop:Emden-Fowler}}]
        Emden--Fowler equation deals with a differential equation with a form of $\frac{d^2 \varphi }{dq^2} = A q^n \varphi^m$, and the differential equation that we have corresponds to the case of $A=-1$, $n=1$, and $m=-1/2$.
        In this case, the solutions can be written in parametric form \citep[2.3.27]{ODEhandbook}:
        for constants $a$ and $b$ such that $A = -\frac{9}{2} (b/a)^3$,
        \begin{equation} \label{eq:Emden-Fowler-solution-left}
                q(\theta) = a \theta^{-\frac{2}{3}}\left[ \left( \theta Z'(\theta) + \frac{1}{3} Z(\theta) \right)^2 - \theta^2 Z^2(\theta) \right]
                , ~
                \varphi(\theta) = b \theta^{\frac{2}{3}} Z^2(\theta)
                , ~
                Z(\theta) = C_1 I_{1/3}(\theta) + C_2 K_{1/3}(\theta),
        \end{equation}
        or
        \begin{equation} \label{eq:Emden-Fowler-solution-right}
                q(\theta) = a \theta^{-\frac{2}{3}}\left[ \left( \theta Z'(\theta) + \frac{1}{3} Z(\theta) \right)^2 + \theta^2 Z^2(\theta) \right]
                , ~
                \varphi(\theta) = b \theta^{\frac{2}{3}} Z^2(\theta)
                , ~
                Z(\theta) = C_1 J_{1/3}(\theta) + C_2 Y_{1/3}(\theta),
        \end{equation}
        where $\theta \in \mathbb{R}_+$ is the parameter, and $C_1$ and $C_2$ are arbitrary constants.

        Note that the expression \eqref{eq:Emden-Fowler-left} is obtained by taking $C_1 = - \frac{2}{\pi}$ and $C_2 = 0$ in \eqref{eq:Emden-Fowler-solution-left}, and the expression \eqref{eq:Emden-Fowler-right} is obtained by taking $C_1 = \sqrt{3}$ and $C_2 = -1$ in \eqref{eq:Emden-Fowler-solution-right}.
        Therefore it suffices to show that the curve $\left\{  \left( q_L(\theta), \varphi_L(\theta) \right) \right\}_{\theta \in (0, \infty]} \bigcup \left\{ \left( q_R(\theta), \varphi_R(\theta) \right) \right\}_{\theta \in (0, \bar{\theta}]}$ is a valid graph satisfying the boundary conditions, i.e.,
        \begin{enumerate}[label=(\roman*)]
                \item $\lim_{\theta \nearrow \infty} \big( q_L(\theta), \varphi_L(\theta) \big) = (0, 0)$.
                \item $(q_R(\bar{\theta}), \varphi_R(\bar{\theta})) = (1,0)$.
                \item The left part and the right part meet at a point, i.e., $\lim_{\theta \searrow 0} \big( q_L(\theta), \varphi_L(\theta) \big) = \lim_{\theta \searrow 0} \big( q_R(\theta), \varphi_R(\theta) \big)$.
                \item They have the same slope at the contact point, i.e., $\lim_{\theta \searrow 0} \frac{d\phi_L/d\theta}{dq_L/d\theta} = \lim_{\theta \searrow 0} \frac{d\phi_R/d\theta}{dq_R/d\theta}$.
        \end{enumerate}

        First, observe that $\lim_{\theta \nearrow \infty} Z_L(\theta) = - \frac{2}{\pi} \lim_{\theta \nearrow \infty} K_{1/3}(\theta) = 0$ and $\lim_{\theta \nearrow \infty} Z_L'(\theta) = \frac{1}{\pi} \lim_{\theta \nearrow \infty} \big( K_{-2/3}(\theta) + K_{4/3}(\theta) \big) = 0$.
        Therefore, $\lim_{\theta \nearrow \infty} q_L(\theta) = 0$ and $\lim_{\theta \nearrow \infty} \varphi_L(\theta) = 0$, which proves (i).

        Next, observe that the value of $\bar{\theta}$ was chosen to satisfy $Z_R(\bar{\theta}) = 0$, and the value of $a$ was chosen to satisfy $q_R(\bar{\theta}) = a \times \bar{\theta}^{\frac{4}{3}} \left( Z_R'(\bar{\theta}) \right)^2  = 1$, which proves (ii).

        Finally, with some algebra, it can be shown that
        \begin{equation}
                \lim_{\theta \searrow 0} q_L(\theta) = \lim_{\theta \searrow 0} q_R(\theta) = \frac{2^{\frac{10}{3}} a}{3 \Gamma^2(\frac{1}{3})}
                , \quad
                \lim_{\theta \searrow 0} \varphi_L(\theta) = \lim_{\theta \searrow 0} \varphi_R(\theta) = \frac{2^{\frac{8}{3}} b}{3 \Gamma^2(\frac{2}{3})},
        \end{equation}
        and
        \begin{equation}
                \lim_{\theta \searrow 0} \frac{d\phi_L(\theta)/d\theta}{dq_L(\theta)/d\theta} = \lim_{\theta \searrow 0}  \frac{d\phi_L(\theta)/d\theta}{dq_L(\theta)/d\theta} = \frac{3 \times 2^{\frac{1}{3}}b \Gamma(\frac{2}{3}) }{a \Gamma(\frac{1}{3})},
        \end{equation}
        where $\Gamma$ is the Gamma function.
        These results prove (iii) and (iv).

        \noindent
        \textbf{Some notes on the determination of constants.}
        In the representation of the general solution,  \eqref{eq:Emden-Fowler-solution-left} and \eqref{eq:Emden-Fowler-solution-right}, there are seven constants that need to be identified: $a$, $b$, $\bar{\theta}$, $C_{L1}$, $C_{L2}$, $C_{R1}$, and $C_{R2}$.
        The upper limit $\bar{\theta}$ and the constant $a$ are uniquely determined by $(C_{L1}, C_{L2}, C_{R1}, C_{R2})$ due to (ii), and the constant $b$ is also uniquely determined by $a$ due to the identity $(b/a)^3=9/2$.
        We can also observe that the curve $\{ (q(\theta), \varphi(\theta)) \}_{\theta \in \mathbb{R}_+}$ is invariant to a uniform scaling of $(C_{L1}, C_{L2}, C_{R1}, C_{R2})$, and thus we can set $C_{R4}=-1$ without loss of generality.
        We can further obtain a system of equations from the other conditions: $C_{L2} = 0$ from (i), $C_{L1}^2 \pi^2 = 4 C_{R2}^2$ and $C_{L1}^2 \pi^2 = 3 C_{R1}^2 + 2\sqrt{3} C_{R1} C_{R2} + 3 C_{R2}^2$ from (iii), and $C_{R2} = - \sqrt{3} C_{R1}$ from (iv).
        These three equations uniquely determine the values of $C_{L1}$, $C_{L2}$, and $C_{R1}$.
\end{proof}


\end{document}